\documentclass[11pt]{article}%
\usepackage{amsmath}
\usepackage{amsfonts}
\usepackage{amssymb}
\usepackage{graphicx}
\usepackage{longtable}
\usepackage{hyperref}
\usepackage{pdflscape}
\usepackage{multirow}
\usepackage{natbib}
\usepackage[letterpaper,tmargin=1in,bmargin=1in,lmargin=1in,rmargin=1in,dvips]%
{geometry}
\usepackage[onehalfspacing]{setspace}%
\usepackage{mathtools}

\allowdisplaybreaks
\providecommand{\U}[1]{\protect\rule{.1in}{.1in}}
\newtheorem{theorem}{Theorem}

\newenvironment{proof}[1][Proof]{\noindent\textbf{#1.} }{\ \rule{0.5em}{0.5em}}
\ifx\pdfoutput\relax\let\pdfoutput=\undefined\fi
\newcount\msipdfoutput
\ifx\pdfoutput\undefined\else
\ifcase\pdfoutput\else
\ifx\paperwidth\undefined\else
\ifdim\paperheight=0pt\relax\else\pdfpageheight\paperheight\fi
\ifdim\paperwidth=0pt\relax\else\pdfpagewidth\paperwidth\fi
\fi\fi\fi

\begin{document}

\title{\textbf{A Test for Kronecker Product Structure Covariance Matrix}%
\thanks{Guggenberger gratefully acknowledges the hospitality of the EUI in
Florence while parts of the paper were drafted. Mavroeidis gratefully
acknowledges the research support of the European Research Council via
Consolidator grant number 647152. We would like to thank the Editor Serena Ng, Associate Editor and two anonymous referees for helpful comments and Lewis McLean for
research assistance.}}
\author{$%
\begin{array}
[c]{c}%
\text{Patrik Guggenberger}\\
\text{Department of Economics}\\
\text{Pennsylvania State University}%
\end{array}
$
\and $%
\begin{array}
[c]{c}%
\text{Frank Kleibergen}\\
\text{Amsterdam School of Economics}\\
\text{University of Amsterdam}%
\end{array}
$
\and \vspace{0.25in}$%
\begin{array}
[c]{c}%
\text{Sophocles Mavroeidis}\\
\text{Department of Economics}\\
\text{University of Oxford}%
\end{array}
$}
\date{First Version: October, 2019\\
Revised: \today\vspace{0.25in}}
\maketitle

\begin{abstract}
We propose a test for a covariance matrix to have Kronecker Product Structure
(KPS). KPS implies a reduced rank restriction on a certain transformation of
the covariance matrix and the new procedure is an adaptation of the
\cite{kleibergen2006grr} reduced rank test. To derive the limiting
distribution of the Wald type test statistic proves challenging partly because
of the singularity of the covariance matrix estimator that appears in the
weighting matrix. We show that the test statistic has a $\chi^{2}$ limiting
null distribution with degrees of freedom equal to the number of restrictions
tested. Local asymptotic power results are derived. Monte Carlo simulations
reveal good size and power properties of the test. Re-examining fifteen highly
cited papers conducting instrumental variable regressions, we find that KPS is
not rejected in 56 out of 118 specifications at the 5\% nominal size. \medskip

Keywords: covariance matrix, heteroskedasticity, invariance, Kronecker product
structure, linear instrumental variables regression model, reduced rank, weak identification

JEL codes: C12, C26

\end{abstract}

\section{Introduction}

The robustness properties of nonparametric covariance matrix estimators, like
those proposed by \cite{White80} against heteroskedasticity and by, for
example, \cite{Newe87} and \cite{Andr91} against heteroskedasticity and
autocorrelation, have led to the current default of conducting
semi-parametric inference in econometrics. It is well understood that compared
to parametrically specified covariance matrix estimators, these robustness
properties come at the cost of a large number of additional estimated
components. The latter affects the precision of semi-parametric estimators of
the structural parameters compared to parametric ones.

For some structural models estimated by the generalized method of moments
(GMM), see \cite{Hans82l}, use of nonparametric covariance matrix estimators
may also lead to computational challenges for estimation of the structural
parameters when using the continuous updating estimator (CUE) of
\cite{Hans96l}. Prominent examples of such models are the linear instrumental
variables (IV) regression model and the linear factor model in asset pricing.
When using a nonparametric covariance matrix estimator, the CUE objective
function is often ill behaved, that is, it is flat and/or has many local
extrema, making the CUE difficult to compute. Part of the appeal of the CUE
stems from a number of weak-identification-robust tests based on statistics
centered around the CUE for hypotheses involving the structural parameters,
see e.g. \cite{Kleibergen2005}.

When one uses a Kronecker Product Structure (KPS) covariance matrix estimator
instead of a nonparametric one in the CUE\ objective function in linear IV and
factor asset pricing models, the CUE (which is then typically referred to as
the limited information maximum likelihood (LIML) estimator) is
straightforward to compute. Furthermore, weak-identification-robust tests
specified on a subvector of the structural parameter vector with uniformly
better power than projected robust full-vector tests are available, see e.g.
\cite{gkm19}, \cite{GKM21}, and \cite{kf21}. The KPS structure of the
covariance matrix also allows for an analytical computation of the confidence
sets of the structural parameters using the algorithm from
\cite{dufour2005projection}.\footnote{\cite{dufour2005projection} actually
assume homoskedasticity, which is a special case of KPS covariance, but their
algorithm can be modified to cover the more general case of KPS covariance.}

The above illustrates the trade-off between, on the one hand, the robustness
provided by a nonparametric covariance matrix estimator and, on the other
hand, the computational ease and accurate statistical inference provided by a
KPS covariance matrix estimator. To help empirical researchers decide when the
use of a KPS covariance matrix estimator is justified, we develop a test for
the null hypothesis that the covariance matrix $R=E\left(  \frac{1}{n}%
\sum_{i=1}^{n}f_{i}f_{i}^{\prime}\right)  $ has KPS, where $f_{i}=V_{i}\otimes
Z_{i}$ and $V_{i}\in\mathbb{R}^{p}$ and $Z_{i}\in\mathbb{R}^{k}$ are
uncorrelated random vectors. Here $V_{i}$ are unobserved error variables (for
which consistent estimators are available) and $Z_{i}$ are observed
regressors. This setup encompasses, for example, the linear IV and factor asset
pricing models.

The test is based on the insight that KPS implies that a certain invertible
transformation $\mathcal{R}\left( R\right)  $ of $R$ has rank one, see
\cite{vanloanpit93} and (\ref{eq: red rank cov}) below, and our procedure
adapts the \cite{kleibergen2006grr} reduced rank statistic to test for
KPS.\footnote{Another adaptation of the \cite{kleibergen2006grr} reduced-rank
statistic is by \cite{dfp07}, who develop a test for singularity of a
symmetric matrix.} More precisely, the new test statistic is given as a
quadratic form in $vec(\hat{\Lambda})$ with weighting matrix that depends on
$vec(\mathcal{R}(\hat{R})),$ where $\hat{R}$ is a sample analogue of $R$ and
$\hat{\Lambda}$ is an estimator for a certain matrix that is known to be rank
restricted under the null hypothesis (see (\ref{eq: lambda}%
)-(\ref{eq: hyp rank}) below). The adaptation of \cite{kleibergen2006grr} is
nontrivial partly because the covariance matrix of $vec(\mathcal{R}(\hat
{R})),$ that appears in the modified test statistic, is singular. As a
consequence, it is a priori not obvious whether the use of the Moore-Penrose
generalized inverse in the expression of the \cite{kleibergen2006grr} reduced
rank statistic still leads to a $\chi^{2}$ limiting distribution. To answer
that question, we first derive the limiting distribution of $\hat{\Lambda}$
and show the limit to be degenerate Normal. We next establish that the
probability limit of the Moore-Penrose inverse of the covariance matrix
involved in the \cite{kleibergen2006grr} rank statistic is such that it
offsets this degeneracy. As the final result, we conclude that the new KPS
test statistic has a $\chi^{2}$ limiting null distribution with degrees of
freedom equal to the number of tested restrictions. We also consider an
asymptotic setup where $p,$ $k,$ and $n$ jointly go to infinity and show that
the asymptotic null rejection probability of the test is controlled as long as
$(pk)^{16}=o(n^{3}).$ For power considerations, we establish that under
sequences of covariance matrices local to KPS the test statistic has a
limiting noncentral chi square distribution.

As an important property we show that the proposed test is invariant to
orthonormal transformations of the data. In contrast, we show that this is not
true for certain alternative tests for KPS that are based on an application of
the \cite{kleibergen2006grr} test statistic to a different transformation of
the sample covariance matrix estimator that does not lead to a singular
covariance matrix (and may therefore a priori seem the more natural choice).

We provide comprehensive Monte Carlo simulations that document good size and
power properties of the suggested test. Finally, we apply the new KPS test to
various specifications of linear IV models employed in fifteen highly cited
empirical studies recently published in top ranked economic journals. We find
that for the specifications with independent data and moderate numbers of
observations, KPS is not rejected in 24 out of 30 cases at the 5\%
significance level, while for smaller numbers of observations, it is rejected in
14 out of 28 cases. In specifications with clustered data, KPS is not rejected
in 7 out of 17 cases with moderate sample sizes, and 11 out of 35 cases with
smaller samples. Overall, KPS is not rejected in 56 out of the 118
specifications that we tested. The relatively high number of non-rejections
illustrates the potential importance of the KPS test for applied work.

In a companion paper, \cite{GKM21}, we show how the new KPS\ test can be used
as a key ingredient in a testing procedure with correct asymptotic size for a
null hypothesis that restricts the values of a subvector of the structural
parameter vector in the linear IV model with a general covariance matrix. The
first step of the algorithm uses the KPS test to test the null of a KPS of the
covariance matrix of the unrestricted reduced-form sample moment vector. In
the second step of the algorithm, the null hypothesis involving the structural
parameter is tested using the improved subvector Anderson-Rubin test from
\cite{gkm19} when the test in the first step does not reject and using the
size correct AR%
%TCIMACRO{\TEXTsymbol{\backslash}}%
%BeginExpansion
$\backslash$%
%EndExpansion
AR test procedure from \cite{and17} otherwise. The AR%
%TCIMACRO{\TEXTsymbol{\backslash}}%
%BeginExpansion
$\backslash$%
%EndExpansion
AR procedure from \cite{and17} is an asymptotically size correct inference
procedure for testing hypotheses on a subvector of the structural parameters
for general covariance matrices but is less powerful than the improved
subvector Anderson-Rubin test from \cite{gkm19} in the linear IV regression
model. However, the latter test is asymptotically size correct only when the
covariance matrix has KPS. \cite{GKM21} establish that the resulting two-step
procedure has correct asymptotic size and conduct Monte-Carlo experiments
which show that it leads to more powerful subvector inference than the AR%
%TCIMACRO{\TEXTsymbol{\backslash}}%
%BeginExpansion
$\backslash$%
%EndExpansion
AR test in \cite{and17}.

As in the linear IV\ regression model, a KPS structure of the covariance
matrix of the sample moment vector of the linear regression model encompassing
linear asset pricing models also leads to improvements in terms of the power
of identification robust tests on individual elements of the vector of risk
premia and computational ease of obtaining the estimator of the risk premia.
There is increasing awareness that risk premia of many risk factors are only
weakly identified, see e.g. \cite{kz99}, \cite{kf09}, and \cite{kz2020}.
Therefore, it is important to analyze them using inference methods that are
robust to weak identification. The current state of the art for conducting
weak-factor-robust inference on risk premia is to assume homoskedasticity.
Extending homoskedasticity to KPS or even further by extending the switching
test procedure from \cite{GKM21} would extend the scope of the
weak-factor-robust inference methods for analyzing the individual risk premia
in linear asset pricing models. The KPS test would be an integral part of such extensions.

KPS or separability, which is how other fields sometimes refer to KPS, of the
covariance matrix is also studied in the statistics and signal processing
literature. The distance to a covariance matrix with KPS\ is considered in
\cite{Genton07} and \cite{VH17}, while \cite{LuZim05} and \cite{mgg06} analyze
the likelihood ratio test of KPS of the covariance matrix of Normally
distributed data. They estimate the elements of the KPS covariance matrix
using a switching algorithm. Exploiting the reduced rank restriction imposed
on the reordered covariance matrix by KPS is also done in \cite{wjs08}. Their
results are, however, based on a complex Gaussian distribution for the data,
which leads to a degrees of freedom parameter of the $\chi^{2}$ limiting
distribution of their test that is different from the one derived here.

KPS is an example of dimension reduction of a covariance matrix. Other
examples of dimension reduction result from shrinking the covariance matrix to
a matrix with (much) fewer unrestricted elements to estimate, for example, a
scalar multiple of the identity matrix, see e.g. \cite{lw12}, or by shrinking
the population eigenvalues, see e.g. \cite{lw15} and \cite{lw18}.

The paper is organized as follows. In the second section, we introduce the new
test for a KPS covariance matrix and derive the asymptotic null distribution
of the test statistic, which we denote as KPST. The third section contains the
limiting distribution of the KPST statistic under local alternatives while the
fourth section conducts a simulation study to analyze the size and power of
the new KPS test. The fifth section summarizes the extensive analysis of
testing for a KPS\ reduced-form covariance matrix in a considerable number of
prominent articles. The final sixth section concludes. Proofs and detailed
empirical results are given in the Appendix.

We use the vec operator of the matrix $A$, $vec(A):=(a_{1}^{\prime}\ldots
a_{k}^{\prime})^{\prime}\in\mathbb{R}^{mk}$ for an $m\times k$ dimensional
matrix $A=(a_{1},\ldots,a_{k}).$ For a symmetric $m\times m$ dimensional
matrix $A,$ we also use the $m^{2}\times\frac{1}{2}m(m+1)$ dimensional,
so-called, duplication matrix $D_{m}$ which selects the $\frac{1}{2}m(m+1)$
unique elements of $A$ in the $\frac{1}{2}m(m+1)$ dimensional vector $vech(A)$
that vectorizes only the lower triangular part of $A$:
\[
vech(A)=(D_{m}^{\prime}D_{m})^{-1}D_{m}^{\prime}vec(A)\text{ \ \ \ and
\ \ }vec(A)=D_{m}vech(A).
\]

\section{A Test for Kronecker Product Structure Covariance Matrix}

We propose a test for a covariance matrix $R\in\mathbb{R}^{kp\times kp}$ to
have KPS, where%
\begin{equation}%
\begin{array}
[c]{c}%
R:=E\left(  \frac{1}{n}\sum_{i=1}^{n}f_{i}f_{i}^{\prime}\right)  ,
\end{array}
\label{eq: cov}%
\end{equation}
for mean zero, independently distributed random vectors $f_{i}\in
\mathbb{R}^{kp},$ $i=1,\ldots,n,$ which satisfy
\begin{equation}
f_{i}:=(V_{i}\otimes Z_{i}) \label{eq: fi}%
\end{equation}
with $V_{i}\in\mathbb{R}^{p}$ and $Z_{i}\in\mathbb{R}^{k}$ uncorrelated random
vectors.\footnote{The matrix $R$ can depend on the sample size $n$ but for
simplicity of notation we do not index $R$ by $n.$} The specification of
$f_{i}$ fits, for example, a setting where $V_{i}$ contains the errors of a
number of regression equations and $Z_{i}$ contains the regressors, so that
$R$ is then the covariance matrix of the sample covariance between these
errors and the regressors.

It follows that the covariance matrix has a block structure%
\begin{equation}
R:=\left(
\begin{array}
[c]{ccc}%
R_{11} & \cdots & R_{1p}\\
\vdots & \ddots & \vdots\\
R_{p1} & \cdots & R_{pp}%
\end{array}
\right)  , \label{eq: block}%
\end{equation}
where $R_{jl}\in\mathbb{R}^{k\times k},$ $j,l=1,\ldots,p$. Because
$R_{jl}=E\left(  \frac{1}{n}\sum_{i=1}^{n}V_{ij}V_{il}Z_{i}Z_{i}^{\prime
}\right)  \allowbreak=R_{lj}^{\prime},$ for $V_{i}=(V_{i1}\ldots
V_{ip})^{\prime},$ it follows that $R_{jl}$ is symmetric. We are interested in
testing if the covariance matrix $R$ has KPS:%
\begin{equation}
\text{H}_{0}:R=G_{1}\otimes G_{2} \label{eq: KPS hyp}%
\end{equation}
with $G_{1}\in\mathbb{R}^{p\times p}$ and $G_{2}\in\mathbb{R}^{k\times k}$
symmetric positive definite matrices, against the alternative hypothesis of
not having KPS. For normalization purposes, we set one diagonal element equal
to one (say the upper left element of $G_{1}$%
).\footnote{\label{foot: normalization}Normalizing $G_{1,11}$ to one is an
obvious normalization because $G_{1}$ is a positive definite matrix (because
$(G_{1}\otimes G_{2})$ is positive definite) so its diagonal elements are all
strictly larger than zero. The normalization does therefore not imply a
restriction.} When $p=1$ or $k=1$ the null is always true and from now on we
assume that $\min\{p,k\}\geq2.$ To measure the distance of the sample
covariance matrix estimator from a KPS\ covariance matrix, we use a convenient
(invertible) transformation proposed by \cite{vanloanpit93}.

\noindent For a matrix $A\in\mathbb{R}^{kp\times kp}$ with block structure as
in (\ref{eq: block}) define%
\begin{equation}
\mathcal{R}\left(  A\right)  :=%
\begin{pmatrix}
A_{1}\\
\vdots\\
A_{p}%
\end{pmatrix}
\in\mathbb{R}^{p^{2}\times k^{2}}\quad\text{with }A_{j}:=%
\begin{pmatrix}
vec\left(  A_{1j}\right)  ^{\prime}\\
\vdots\\
vec\left(  A_{pj}\right)  ^{\prime}%
\end{pmatrix}
\in\mathbb{R}^{p\times k^{2}}, \label{eq: trans cov}%
\end{equation}
for $j=1,...,p.$ One can easily show that%
\begin{equation}
\mathcal{R}\left(  G_{1}\otimes G_{2}\right)  =vec(G_{1})vec(G_{2})^{\prime}
\label{eq: red rank cov}%
\end{equation}
and by Theorem 2.1 in \cite{vanloanpit93}, we have
\[
\left\Vert R-G_{1}\otimes G_{2}\right\Vert _{F}=\allowbreak\left\Vert
\mathcal{R}(R)-\allowbreak vec(G_{1})vec(G_{2})^{\prime}\right\Vert _{F},
\]
with $\left\Vert .\right\Vert _{F}$ the Frobenius or trace norm of a matrix,
$\left\Vert A\right\Vert _{F}^{2}:=tr(A^{\prime}A)=vec(A)^{\prime}vec(A),$ for
any rectangular matrix $A$. Because $\mathcal{R}\left(  G_{1}\otimes
G_{2}\right)  $ is a matrix of rank one, when testing for a KPS, it is more
convenient to test for the rank of $\mathcal{R}\left(  R\right)  $ to be one
instead of directly testing for KPS\ of $R$.

Consider the covariance matrix estimator%
\begin{equation}%
\begin{array}
[c]{c}%
\hat{R}:=\frac{1}{n}%
%TCIMACRO{\tsum \limits_{i=1}^{n}}%
%BeginExpansion
{\textstyle\sum\limits_{i=1}^{n}}
%EndExpansion
\hat{f}_{i}\hat{f}_{i}^{\prime}\in\mathbb{R}^{kp\times kp}%
\end{array}
\label{eq: sample cov}%
\end{equation}
which uses sample values $\hat{f}_{i}:=\hat{V}_{i}\otimes Z_{i}$ of the random
vectors $f_{i}$ for some estimated residuals $\hat{V}_{i}.$ We assume that
$\hat{f}_{i}=f_{i}+o_{p}(1),$ uniformly over $i=1,\ldots,n,$ as $n\rightarrow
\infty.$ Define the distance from a KPS covariance matrix by the Frobenius
norm%
\begin{equation}
DS:=\min_{G_{1}>0,G_{2}>0, G_{1,11}=1}\left\Vert \mathcal{R}(\hat{R})-vec(G_{1}%
)vec(G_{2})^{\prime}\right\Vert _{F}, \label{eq: dis KPS}%
\end{equation}
where $G_{1},$ $G_{2}>0$ indicates that $G_{1}\in\mathbb{R}^{p\times p}$ and
$G_{2}\in\mathbb{R}^{k\times k}$ are positive definite symmetric matrices, and $G_{1,11}=1$ states that the upper left element of $G_1$ is normalized to 1.

We test for $\mathcal{R}(\hat{R})$ being a rank one matrix using the
\cite{kleibergen2006grr} rank statistic. To describe the
\cite{kleibergen2006grr} rank statistic consider first a singular value
decomposition (SVD) of $\mathcal{R}(\hat{R})$:%
\begin{equation}
\mathcal{R(}\hat{R})=\hat{L}\hat{\Sigma}\hat{N}^{\prime},
\label{eq: SVD general}%
\end{equation}
where $\hat{\Sigma}:=diag(\hat{\sigma}_{1},\ldots,\hat{\sigma}_{\min
(p^{2},k^{2})})$ denotes a $p^{2}\times k^{2}$ dimensional diagonal matrix
with the singular values $\hat{\sigma}_{j}$ ($j=1,...,\min(p^{2},k^{2})$) on
the main diagonal ordered non-increasingly, and with $\hat{L}\in
\mathbb{R}^{p^{2}\times p^{2}}$ and $\hat{N}\in\mathbb{R}^{k^{2}\times k^{2}}$
orthonormal matrices. Decompose%
\begin{equation}
\hat{L}:=%
\begin{pmatrix}
\hat{L}_{11} & \hat{L}_{12}\\
\hat{L}_{21} & \hat{L}_{22}%
\end{pmatrix}
=\left(  \hat{L}_{1}\text{ }\vdots\text{ }\hat{L}_{2}\right)  ,\text{ }%
\hat{\Sigma}:=%
\begin{pmatrix}
\hat{\sigma}_{1} & 0\\
0 & \hat{\Sigma}_{2}%
\end{pmatrix}
,\text{ }\hat{N}:=%
\begin{pmatrix}
\hat{N}_{11} & \hat{N}_{12}\\
\hat{N}_{21} & \hat{N}_{22}%
\end{pmatrix}
=(\hat{N}_{1}\text{ }\vdots\text{ }\hat{N}_{2}), \label{eq: svdpar}%
\end{equation}
with $\hat{L}_{11}:1\times1,$ $\hat{L}_{12}:1\times(p^{2}-1),$ $\hat{L}%
_{21}:(p^{2}-1)\times1,$ $\hat{L}_{22}:(p^{2}-1)\times(p^{2}-1),$ $\hat
{\sigma}_{1}:1\times1,$ $\hat{\Sigma}_{2}:(p^{2}-1)\times(k^{2}-1),$ $\hat
{N}_{11}:1\times1,$ $\hat{N}_{12}:1\times(k^{2}-1),$ $\hat{N}_{21}%
:(k^{2}-1)\times1,$ $\hat{N}_{22}:(k^{2}-1)\times(k^{2}-1)$ dimensional matrices.

\begin{theorem}
\label{th: frob norm}Suppose $\hat{R}$ is positive definite. The distance
measure DS in (\ref{eq: dis KPS}) equals the square root of the sum of squares
of all but the largest singular value of $\mathcal{R}(\hat{R})\in
\mathbb{R}^{p^{2}\times k^{2}},$ i.e. $DS^{2}=\sum_{i=2}^{\min(p^{2},k^{2}%
)}\hat{\sigma}_{i}^{2},$ where $\hat{\sigma}_{1}\geq...\geq\hat{\sigma}%
_{\min(p^{2},k^{2})}$ are the ordered singular values of $\mathcal{R}(\hat
{R}).$ Furthermore, if $\hat{\sigma}_{1}>\hat{\sigma}_{2},$ then the positive definite symmetric 
minimizers $\hat{G}_{1},\hat{G}_{2}$ of (\ref{eq: dis KPS}) have the following unique expression:%
\begin{equation}%
\begin{array}
[c]{ll}%
vec(\hat{G}_{1}):=\hat{L}_{1}/\hat{L}_{11} & \in\mathbb{R}^{p^{2}\times1},\\
vec(\hat{G}_{2})^{\prime}:=\hat{L}_{11}\hat{\sigma}_{1}\hat{N}_{1}^{\prime
}\quad & \in\mathbb{R}^{1\times k^{2}}.
\end{array}
\label{eq: g1g2 est}%
\end{equation}

\end{theorem}

\begin{proof}
See the Appendix.\smallskip
\end{proof}

If $R$ is positive definite, and $\hat{R}\overset{p}{\rightarrow}R$, then
$\hat{R}$ will be positive definite with probability approaching one
(w.p.a.1). The choice of normalization in (\ref{eq: g1g2 est}) conforms with
the normalization of $G_{1},G_{2}$ in (\ref{eq: KPS hyp}) discussed in
Footnote \ref{foot: normalization} above.

\paragraph{The KPST statistic}

We use the distance between $\mathcal{R}(\hat{R})$ and a matrix of rank one to
test for a KPS of $R$. The test is based on the limiting distribution of the
unique elements of $\hat{R}$ or equivalently $\mathcal{R}(\hat{R}).$ These
elements result from using the $k^{2}\times\frac{1}{2}k(k+1)$ and $p^{2}%
\times\frac{1}{2}p(p+1)$ dimensional duplication matrices $D_{k}$ and
$D_{p}:$
\begin{equation}%
\begin{array}
[c]{rl}%
\mathcal{R}(\hat{R})= & \mathcal{R}\left(  \frac{1}{n}\sum_{i=1}^{n}(\hat
{V}_{i}\hat{V}_{i}^{\prime}\otimes Z_{i}Z_{i}^{\prime})\right) \\
= & \frac{1}{n}\sum_{i=1}^{n}vec(\hat{V}_{i}\hat{V}_{i}^{\prime}%
)vec(Z_{i}Z_{i}^{\prime})^{\prime}\\
= & D_{p}\hat{R}^{\ast}D_{k}^{\prime},
\end{array}
\label{eq: R-spec}%
\end{equation}
with
\begin{equation}%
\begin{array}
[c]{c}%
\hat{R}^{\ast}:=\frac{1}{n}\sum_{i=1}^{n}vech(\hat{V}_{i}\hat{V}_{i}^{\prime
})vech(Z_{i}Z_{i}^{\prime})^{\prime}.
\end{array}
\label{eq: R star}%
\end{equation}
The $\frac{1}{2}p(p+1)\times\frac{1}{2}k(k+1)$ dimensional matrix $\hat
{R}^{\ast}$ contains the unique elements of $\hat{R}$ and $\mathcal{R}(\hat
{R}).$ We assume $vec(\hat{R}^{\ast})$ satisfies a central limit theorem:
\begin{equation}%
\begin{array}
[c]{cc}%
\sqrt{n}(vec(\hat{R}^{\ast})-vec(R^{\ast})) & \underset{d}{\rightarrow}%
\psi=vec(\Psi),
\end{array}
\label{eq: CLT R_n}%
\end{equation}
with $\psi\sim N(0,V_{R^{\ast}}),$ $\Psi$ a $\frac{1}{2}p(p+1)\times\frac
{1}{2}k(k+1)$ dimensional normally distributed random matrix and
\begin{equation}%
\begin{array}
[c]{rl}%
R^{\ast}:= & E\left(  \frac{1}{n}\sum_{i=1}^{n}vech(V_{i}V_{i}^{\prime
})vech(Z_{i}Z_{i}^{\prime})^{\prime}\right)  ,\\
V_{R^{\ast}}:= & \lim_{n\rightarrow\infty}\left[  E\left(  \frac{1}{n}%
\sum_{i=1}^{n}\left(  vech(Z_{i}Z_{i}^{\prime})vech(Z_{i}Z_{i}^{\prime
})^{\prime}\otimes vech(V_{i}V_{i}^{\prime})vech(V_{i}V_{i}^{\prime})^{\prime
}\right)  \right)  \right. \\
& \left.  -E\left(  vec(\frac{1}{n}\sum_{i=1}^{n}vech(V_{i}V_{i}^{\prime
})vech(Z_{i}Z_{i}^{\prime})^{\prime})\right)  E\left(  vec(\frac{1}{n}%
\sum_{i=1}^{n}vech(V_{i}V_{i}^{\prime})vech(Z_{i}Z_{i}^{\prime})^{\prime
})\right)  ^{\prime}\right]  .
\end{array}
\label{eq: lim_mean_var1n}%
\end{equation}
In fact, we assume a slightly stronger result, namely, that $\hat{R}^{\ast
}=R^{\ast}+\frac{1}{\sqrt{n}}\Psi+o_{p}(n^{-\frac{1}{2}}),$ holds. A central
limit theorem (\ref{eq: CLT R_n}) for (possibly) non-identical distributed
independent random variables holds under mild conditions, e.g. under the
Liapounov's or Lindeberg's condition, see \cite{wh84}.

Define%
\begin{equation}%
\begin{array}
[c]{rll}%
\hat{\Lambda}:= & \left(  \hat{L}_{22}\hat{L}_{22}^{\prime}\right)
^{-1/2}\hat{L}_{22}\hat{\Sigma}_{2}\hat{N}_{22}^{\prime}\left(  \hat{N}%
_{22}\hat{N}_{22}^{\prime}\right)  ^{-1/2} & :(p^{2}-1)\times(k^{2}-1).
\end{array}
\label{eq: lambda}%
\end{equation}
It can be shown that $\hat{\Lambda}=vec(\hat{G}_{1})_{\perp}^{\prime
}\mathcal{R}(\hat{R})vec(\hat{G}_{2})_{\perp},$ where
\[%
\begin{array}
[c]{ll}%
vec(\hat{G}_{1})_{\perp}:=\hat{L}_{2}\hat{L}_{22}^{-1}(\hat{L}_{22}\hat{L}%
_{22}^{\prime})^{1/2} & :p^{2}\times(p^{2}-1),\\
vec(\hat{G}_{2})_{\perp}^{\prime}:=\left(  \hat{N}_{22}\hat{N}_{22}^{\prime
}\right)  ^{1/2}\hat{N}_{22}^{\prime-1}\hat{N}_{2}^{\prime} & :(k^{2}-1)\times
k^{2},
\end{array}
\]
see
\cite[page 102]{kleibergen2006grr}. We then have%
\begin{equation}
\mathcal{R}(\hat{R})=vec(\hat{G}_{1})vec(\hat{G}_{2})^{\prime}+vec(\hat{G}%
_{1})_{\perp}\hat{\Lambda}vec(\hat{G}_{2})_{\perp}^{\prime}.
\label{eq: rank speci}%
\end{equation}

Using $\mathcal{R(}R),$ our hypothesis of interest H$_{0}$ (\ref{eq: KPS hyp})
is transformed into%
\begin{equation}
\text{H}_{0}:\mathcal{R(}R)=vec(G_{1})vec(G_{2})^{\prime}\text{ or H}%
_{0}:vec(G_{1})_{\perp}^{\prime}\mathcal{R(}R)vec(G_{2})_{\perp}=0,
\label{eq: hyp rank}%
\end{equation}
where $vec(G_{1})_{\perp}$ and $vec(G_{2})_{\perp}$ are $p^{2}\times(p^{2}-1)$
and $k^{2}\times(k^{2}-1)$ dimensional matrices that contain the orthogonal
complements of $vec(G_{1})$ and $vec(G_{2}),$ $vec(G_{1})_{\perp}^{\prime
}vec(G_{1})\equiv0,$ $vec(G_{1})_{\perp}^{\prime}vec(G_{1})_{\perp}\equiv
I_{p^{2}-1},$ $vec(G_{2})_{\perp}^{\prime}vec(G_{2})\equiv0,$ $vec(G_{2}%
)_{\perp}^{\prime}vec(G_{2})_{\perp}\equiv I_{k^{2}-1}.$ The KPST test uses
the sample analog of the last component in (\ref{eq: hyp rank}) to test
H$_{0}.$ It further results from identifying $vec(G_{1})$ and $vec(G_{2})$
using the eigenvectors associated with the first singular value of
$\mathcal{R(}R)$.

The \cite{kleibergen2006grr} rank test statistic is a quadratic form of the
vectorization of $\hat{\Lambda}$ in (\ref{eq: lambda}). Its specification
directly extends to the new KPS test but because the covariance matrix of
$vec\left(  \mathcal{R}(\hat{R})\right)  $ is singular, the (degenerate)
asymptotic normal distribution of $vec(\hat{\Lambda})$ and the resulting
degrees of freedom parameter of the $\chi^{2}$ limiting distribution of the
\cite{kleibergen2006grr} rank test statistic are not obvious.

We define the statistic KPST\ for testing H$_{0}$ in (\ref{eq: KPS hyp}) as%
\begin{equation}
KPST:=n\times vec\left(  \hat{\Lambda}\right)  ^{\prime}\left(  \hat
{J}^{\prime}\hat{V}\hat{J}\right)  ^{-}vec\left(  \hat{\Lambda}\right)  ,
\label{eq: kpst}%
\end{equation}
where%
\begin{equation}
\hat{J}:=\left(  vec(\hat{G}_{2})_{\perp}\otimes vec(\hat{G}_{1})_{\perp
}\right)  ,\quad\hat{V}:=\widehat{\text{cov}}\left(  vec\left(  \mathcal{R}%
(\hat{R})\right)  \right)  \in\mathbb{R}^{p^{2}k^{2}\times p^{2}k^{2}},
\label{eq: Jhat Vhat}%
\end{equation}
and
\begin{equation}%
\begin{array}
[c]{rl}%
\widehat{\text{cov}}\left(  vec\left(  \mathcal{R}(\hat{R})\right)  \right)
= & \frac{1}{n}\sum_{i=1}^{n}\left(  vec(Z_{i}Z_{i}^{\prime})vec(Z_{i}%
Z_{i}^{\prime})^{\prime}\otimes vec(\hat{V}_{i}\hat{V}_{i}^{\prime}%
)vec(\hat{V}_{i}\hat{V}_{i}^{\prime})^{\prime}\right)  \\
& -vec\left(  \mathcal{R}(\hat{R})\right)  vec\left(  \mathcal{R}(\hat
{R})\right)  ^{\prime}\\
= & (D_{k}\otimes D_{p})\widehat{\text{cov}}\left(  vec\left(  \hat{R}^{\ast
}\right)  \right)  (D_{k}\otimes D_{p})^{\prime},\\
\widehat{\text{cov}}\left(  vec\left(  \hat{R}^{\ast}\right)  \right)  = &
\frac{1}{n}\sum_{i=1}^{n}\left(  vech(Z_{i}Z_{i}^{\prime})vech(Z_{i}%
Z_{i}^{\prime})^{\prime}\otimes vech(\hat{V}_{i}\hat{V}_{i}^{\prime}%
)vech(\hat{V}_{i}\hat{V}_{i}^{\prime})^{\prime}\right) \\ & -vec\left(  \hat
{R}^{\ast}\right)  vec\left(  \hat{R}^{\ast}\right)  ^{\prime}.
\end{array}
\label{eq: covhat}%
\end{equation}

In the Appendix it is shown that the KPST statistic in (\ref{eq: kpst}) can
be simplified as follows:%
\begin{equation}%
\begin{array}
[c]{rl}%
KPST= & n\times\left(  vec\left(  \hat{\Sigma}_{2}\right)  \right)  ^{\prime
}\left[  (\hat{N}_{2}\otimes\hat{L}_{2})^{\prime}\hat{V}\left(  \hat{N}%
_{2}\otimes\hat{L}_{2}\right)  \right]  ^{-}\left(  vec\left(  \hat{\Sigma
}_{2}\right)  \right)  .
\end{array}
\label{eq: simple KPST}%
\end{equation}
This provides an expression for KPST which is easier to compute. On the other
hand, it cannot be directly used to obtain the $\chi^{2}$ limiting
distribution because $\hat{\Sigma}_{2}$ does not have an asymptotic normal
distribution while $vec(\hat{\Lambda})$ does.

\paragraph{\textbf{The }KPST$^{\ast}$\textbf{ statistic}}

For comparison, we now introduce an alternative test statistic KPST$^{\ast}$
that fits more naturally into the \cite{kleibergen2006grr} framework. However,
unlike KPST, KPST$^{\ast}$ turns out not to be invariant to orthonormal
transformations of the data. Because $vec(G_{1})=D_{p}vech(G_{1}),$
$vec(G_{2})=D_{k}vech(G_{2}),$ the hypothesis of interest (\ref{eq: hyp rank})
can also be specified as:%
\begin{equation}
\text{H}_{0}:R^{\ast}=vech(G_{1})vech(G_{2})^{\prime}\text{ or H}%
_{0}:vech(G_{1})_{\perp}^{\prime}R^{\ast}vech(G_{2})_{\perp}=0,
\label{eq: hyp rankvech}%
\end{equation}
where $vech(G_{1})_{\perp}$ and $vech(G_{2})_{\perp}$ are $\frac{1}%
{2}p(p+1)\times(\frac{1}{2}p(p+1)-1)$ and $\frac{1}{2}k(k+1)\times(\frac{1}%
{2}k(k+1)-1)$ dimensional matrices that contain the orthogonal complements of
$vech(G_{1})$ and $vech(G_{2}),$ $vech(G_{1})_{\perp}^{\prime}vech(G_{1}%
)=0,$ $vech(G_{1})_{\perp}^{\prime}vech(G_{1})_{\perp}= I_{\frac
{1}{2}p(p+1)-1},$ $vech(G_{2})_{\perp}^{\prime}vech(G_{2})=0,$
$vec(G_{2})_{\perp}^{\prime}\allowbreak vec(G_{2})_{\perp}= I_{\frac
{1}{2}k(k+1)-1}.$ This specification of the hypothesis fits directly in the
setup of the \cite{kleibergen2006grr} rank test because the covariance matrix
of $\hat{R}^{\ast}$ is non-singular. Therefore, the corresponding
specification of $vec(\hat{\Lambda})$ converges to a Normally distributed
random vector. The specification of null hypothesis in (\ref{eq: hyp rankvech}%
) allows us to easily infer the number of restrictions tested, which equals
$\left(  \frac{1}{2}k(k+1)-1\right)  \left(  \frac{1}{2}p(p+1)-1\right)  $,
but the resulting rank statistic does not equal KPST in (\ref{eq: kpst}).
Specifically, define the SVD of $\hat{R}^{\ast}%
=\hat{L}^{\ast}\hat{\Sigma}^{\ast}\hat{N}^{\ast\prime}$, where%
\[
\hat{L}^{\ast}:=%
\begin{pmatrix}
\hat{L}_{11}^{\ast} & \hat{L}_{12}^{\ast}\\
\hat{L}_{21}^{\ast} & \hat{L}_{22}^{\ast}%
\end{pmatrix}
,\text{ }\hat{\Sigma}^{\ast}:=%
\begin{pmatrix}
\hat{\sigma}_{1}^{\ast} & 0\\
0 & \hat{\Sigma}_{2}^{\ast}%
\end{pmatrix}
,\text{ }\hat{N}^{\ast}:=%
\begin{pmatrix}
\hat{N}_{11}^{\ast} & \hat{N}_{12}^{\ast}\\
\hat{N}_{21}^{\ast} & \hat{N}_{22}^{\ast}%
\end{pmatrix}
,
\]
with $\hat{L}_{11}^{\ast}:1\times1,$ $\hat{L}_{12}^{\ast}:1\times\left(
\frac{1}{2}p(p+1)-1\right)  ,$ $\hat{L}_{21}^{\ast}:(\frac{1}{2}%
p(p+1)-1)\times1,$ $\hat{L}_{22}^{\ast}:(\frac{1}{2}p(p+1)-1)\times(\frac
{1}{2}p(p+1)-1),$ $\hat{\sigma}_{1}^{\ast}:1\times1,$ $\hat{\Sigma}_{2}^{\ast
}:(\frac{1}{2}p(p+1)-1)\times(\frac{1}{2}k(k+1)-1),$ $\hat{N}_{11}^{\ast
}:1\times1,$ $\hat{N}_{12}^{\ast}:1\times(\frac{1}{2}k(k+1)-1),$ $\hat{N}%
_{21}^{\ast}:(\frac{1}{2}k(k+1)^{2}-1)\times1,$ $\hat{N}_{22}^{\ast}:(\frac
{1}{2}k(k+1)-1)\times(\frac{1}{2}k(k+1)-1)$ dimensional matrices and $\hat
{L}_{2}^{\ast}=(\hat{L}_{12}^{\ast\prime}$ $\vdots$ $\hat{L}_{22}^{\ast\prime
})^{\prime},$ $\hat{N}_{2}^{\ast}=(\hat{N}_{12}^{\ast\prime}$ $\vdots$
$\hat{N}_{22}^{\ast\prime})^{\prime}$. The \cite{kleibergen2006grr} statistic
for testing (\ref{eq: hyp rankvech}) using $\hat{R}^{\ast}$ is
\begin{equation}%
\begin{array}
[c]{l}%
KPST^{\ast}:=n\times vec\left(  \hat{\Lambda}^{\ast}\right)  ^{\prime}\left(
\hat{J}^{\ast\prime}\hat{V}^{\ast}\hat{J}^{\ast}\right)  ^{-}vec\left(
\hat{\Lambda}^{\ast}\right)  ,\text{ where}\\
\hat{\Lambda}^{\ast}:=\left(  \hat{L}_{22}^{\ast}\hat{L}_{22}^{\ast\prime
}\right)  ^{-1/2}\hat{L}_{22}^{\ast}\hat{\Sigma}_{2}^{\ast}\hat{N}_{22}%
^{\ast\prime}\left(  \hat{N}_{22}^{\ast}\hat{N}_{22}^{\ast\prime}\right)
^{-1/2},\\
\hat{J}^{\ast}:=\left(  \left(  N_{22}^{\ast}N_{22}^{\ast\prime}\right)
^{1/2}N_{22}^{\ast\prime-1}\left[  N_{12}^{\ast\prime}\text{ }\vdots\text{
}N_{22}^{\ast\prime}\right]  \otimes\left(  L_{22}^{\ast}L_{22}^{\ast\prime
}\right)  ^{1/2}L_{22}^{\ast\prime-1}\left[  L_{12}^{\ast\prime}\text{ }%
\vdots\text{ }L_{22}^{\ast\prime}\right]  \right)  ,\\
\hat{V}^{\ast}:=\widehat{\text{cov}}\left(  vec\left(  \hat{R}^{\ast}\right)
\right)  \in\mathbb{R}^{\left(  \frac{1}{4}p(p+1)k(k+1)\right)  \times\left(
\frac{1}{4}p(p+1)k(k+1)\right)  },
\end{array}
\label{eq: kpst*}%
\end{equation}
see Corollary 1 in \cite{kleibergen2006grr}.

\paragraph{Asymptotic theory and invariance to orthonormal transformations}

The statistics KPST in (\ref{eq: kpst}) and KPST$^{\ast}$ in
(\ref{eq: kpst*}) are not identical, and, unlike the proposed KPST
statistic, tests of the KPS hypothesis based on the KPST$^{\ast}$ statistic
are not invariant to orthonormal transformations of the data, as stated in the
following Theorem.\smallskip

\begin{theorem}
\label{th: kps test}Assume $E\left(  \left\Vert f_{i}\right\Vert ^{8}\right)
<\kappa$ for some $\kappa<\infty,$ $\hat{f}_{i}=f_{i}+o_{p}(1),$ uniformly for
$i=1,\ldots,n,$ as $n\rightarrow\infty,$ and the central limit theorem in
(\ref{eq: CLT R_n}) holds in the slightly stronger version $\hat{R}^{\ast
}=R^{\ast}+\frac{1}{\sqrt{n}}\Psi+o_{p}(n^{-\frac{1}{2}})$. Then, under
H$_{0},$ for KPST and KPST$^{\ast}$ defined in (\ref{eq: kpst}) and
(\ref{eq: kpst*}), respectively, the following hold:\newline\textbf{a.}\[
KPST\underset{d}{\rightarrow}\chi_{df}^{2}
\]  as
$n\rightarrow\infty$ (for fixed $p$ and $k$) with degrees of freedom%
\begin{equation}%
\begin{array}
[c]{c}%
df:=\left(  \frac{1}{2}k(k+1)-1\right)  \left(  \frac{1}{2}p(p+1)-1\right)  .
\end{array}
\label{eq: a_spec}%
\end{equation}
\newline\textbf{b. }
\[KPST^{\ast}\underset{d}{\rightarrow}\chi_{df}^{2}\]  as
$n\rightarrow\infty$ (for fixed $p$ and $k$) with $df$ as given in
(\ref{eq: a_spec}). \newline\textbf{c. }The statistics KPST and KPST$^{\ast}$ are in general not
numerically identical. While KPST is invariant to orthonormal
transformations of the data in $\hat{V}_{i}$ and $Z_{i},$ KPST$^{\ast}$ is not
invariant to such transformations.\newline\textbf{d. }For sequences $p,$ $k,$
$n$ that satisfy%
\begin{equation}%
\begin{array}
[c]{c}%
\frac{(pk)^{16}}{n^{3}}\rightarrow0,
\end{array}
\label{eq: joint conv}%
\end{equation}
we have
\begin{equation}%
\begin{array}
[c]{c}%
\lim_{n,p,k\rightarrow\infty}\Pr\left[  KPST<\chi_{df,1-\alpha}^{2}\right]
\leq\alpha,
\end{array}
\label{eq: lim many pkn}%
\end{equation}
where $\chi_{df,1-\alpha}^{2}$ denotes the $1-\alpha$ quantile of a $\chi
_{df}^{2}$ distribution.
\end{theorem}

\begin{proof}
see the Appendix.\footnote{We thank an anonymous associate editor for pointing
at the vech operator and duplication matrix to simplify the proof and
exposition.}\smallskip
\end{proof}

We define the new KPST test as follows: it rejects H$_{0}$ in
(\ref{eq: KPS hyp}) at nominal size $\alpha$ if
\begin{equation}
KPST>\chi_{df,1-\alpha}^{2}. \label{eq: KPST test}%
\end{equation}
Based on Theorem \ref{th: kps test}a and d, the resulting test has limiting
null rejection probability bounded by $\alpha.$

Theorem \ref{th: kps test}c shows that the rank-one tests KPST and
KPST$^{\ast}$ that are based on $\mathcal{R(}\hat{R})$ and $\hat{R}^{\ast}$
respectively are not identical even though they are testing the same
underlying hypothesis that the covariance $R$ matrix has KPS. This difference
occurs because these are Wald tests, and Wald statistics are in general not
invariant to non-linear transformations.

Theorem \ref{th: kps test}d provides a sufficient condition for uniform
convergence of $\hat{\Lambda}$ and its covariance matrix estimator for
settings where $p,$ $k,$ and $n$ jointly go to infinity so the main results
for the limiting distribution of KPST remain unaltered. It is needed to assess
the validity of the asymptotic approximation for settings where $p$ and $k$
are relatively large compared to the number of observations $n.$

The conditions in Theorem \ref{th: kps test}d are weaker than those in
\cite{newey2009gmm}. \cite{newey2009gmm} prove the validity of the asymptotic
approximation of test statistics where the number of observations grows faster
than the cube of the number of moment restrictions. The number of moment
restrictions here is proportional to $(pk)^{2}$ so their rate would be
$(pk)^{6}/n\rightarrow0$ which is more restrictive than the rate in
(\ref{eq: joint conv}).

\paragraph{Invariance to nonsingular transformations}

Theorem \ref{th: kps test}c shows the KPST is invariant to orthonormal
transformations, but it is still not invariant to general nonsingular
transformations of the data. To ensure invariance to nonsingular
transformations, we need to normalize the data as in \cite{kleibergen2006grr}.
Specifically, the KPST statistic is computed using the moment vector%
\begin{equation}
\hat{f}_{i}=C_{1}^{\prime}\widehat{V}_{i}\otimes C_{2}^{\prime}Z_{i},
\label{eq: f par boot}%
\end{equation}
where $C_{1}$ and $C_{2}$ are the Choleski factors of the inverse of the
second moments of $\hat{V}_{i}$ and $Z_{i},$ i.e., $C_{1}C_{1}^{\prime
}=\left(  \frac{1}{n}\sum_{i=1}^{n}\widehat{V}_{i}\widehat{V}_{i}^{\prime
}\right)  ^{-1}$ and $C_{2}C_{2}^{\prime}=\left(  \frac{1}{n}\sum_{i=1}%
^{n}Z_{i}Z_{i}^{\prime}\right)  ^{-1}$. To see why this normalization yields
invariance, let $A$ be a nonsingular $k\times k$ matrix, and define the
transformed instruments $Z_{A,i}:=AZ_{i}.$ Let $C_{2A}$ denote the Choleski
factor of $\left(  \frac{1}{n}\sum_{i=1}^{n}AZ_{i}Z_{i}^{\prime}A^{\prime
}\right)  ^{-1}.$ The KPST statistic with the original instruments $Z_{i}$ is
computed using $C_{2}^{\prime}Z_{i}$ in the moment vector
(\ref{eq: f par boot}), while the KPST with the transformed instruments
$AZ_{i}$ uses $C_{2A}^{\prime}AZ_{i}$ in the same formula
(\ref{eq: f par boot}). Therefore, the transformation from $C_{2}^{\prime
}Z_{i}$ to $C_{2A}^{\prime}Z_{A,i}$ is given by $T_{A}:=\allowbreak
C_{2A}^{\prime}AC_{2}^{-1\prime},$ i.e., $C_{2A}^{\prime}Z_{A,i}%
\allowbreak=\allowbreak C_{2A}^{\prime}AZ_{i}\allowbreak=\allowbreak
T_{A}\left(  C_{2}^{\prime}Z_{i}\right)  .$ Now, observe that $T_{A}$ is an
orthonormal matrix, because $T_{A}^{\prime}T_{A}\allowbreak=\allowbreak
C_{2}^{-1}A^{\prime}C_{2A}\allowbreak C_{2A}^{\prime}AC_{2}^{-1\prime
}\allowbreak=\allowbreak C_{2}^{-1}A^{\prime}\allowbreak\left(  \frac{1}%
{n}\sum_{i=1}^{n}AZ_{i}Z_{i}^{\prime}A^{\prime}\right)  ^{-1}\allowbreak
AC_{2}^{-1\prime}\allowbreak=\allowbreak C_{2}^{-1}\allowbreak\left(  \frac
{1}{n}\sum_{i=1}^{n}Z_{i}Z_{i}^{\prime}\right)  ^{-1}\allowbreak
C_{2}^{-1\prime}\allowbreak=\allowbreak I_{k}.$ Hence, invariance follows from
Theorem \ref{th: kps test}c. The exact same argument can be made about
rotations of the reduced form errors $\hat{V}.$

\paragraph{Clustered data}

In case of clustered data, we assume there are $n$ clusters of $N_{i}$
observations each, so the total number of data points is $\sum_{i=1}^{n}%
N_{i}:$
\begin{equation}%
\begin{array}
[c]{c}%
f_{i}=\sum_{j=1}^{N_{i}}f_{ij},
\end{array}
\label{eq: f_i}%
\end{equation}
for mean zero $kp$ dimensional random vectors $f_{ij},$ $j=1,\ldots,N_{i},$
$i=1,\ldots,n.$ Observations $f_{ij}$ within cluster $i$ can be arbitrarily
dependent, i.e., $E\left(  f_{ij}f_{is}\right)  $ is unrestricted for all $j,$
$s=1,...,N_{i},$ while observations across clusters are independent. The
$kp\times kp$ dimensional (positive semi-definite) covariance matrix of the
sample moments then results as:%
\begin{equation}%
\begin{array}
[c]{c}%
R=\frac{1}{n}\sum_{i=1}^{n}E(f_{i}f_{i}^{\prime}).
\end{array}
\label{eq: cluster covariance}%
\end{equation}

\section{Limiting distribution of KPST under local alternatives}

To analyze the power of KPST under local alternatives, we construct the
limiting distribution of the KPST statistic under alternatives where the
covariance matrix of the moments $R\in\mathbb{R}^{kp\times kp}$ is local to
KPS:%
\begin{equation}%
\begin{array}
[c]{c}%
\text{H}_{1}:R=(G_{1}\otimes G_{2})+\frac{1}{\sqrt{n}}A_{0},
\end{array}
\label{eq: R_n local}%
\end{equation}
where $G_{1}\in\mathbb{R}^{p\times p}$ and $G_{2}\in\mathbb{R}^{k\times k}$
are symmetric positive definite matrices, and $A_{0} \in\mathbb{R}^{kp\times kp}$
is a fixed symmetric matrix. The
best-fitting KPS approximation of $R$ under H$_{1}$ w.r.t. Frobenius norm,
defined as $\bar{G}_{1,n}\otimes\bar{G}_{2,n}$, where $\bar{G}_{1,n},\bar
{G}_{2,n}$ solve $\min_{\bar{G}_{1}>0,\bar{G}_{2}>0}\left\Vert (G_{1}\otimes
G_{2})\allowbreak+\frac{1}{\sqrt{n}}A_{0}\allowbreak-\bar{G}_{1}\otimes\bar
{G}_{2}\right\Vert _{F}$,$\,$will differ from $G_{1}\otimes G_{2}$. That is,
$\bar{G}_{1,n}\neq G_{1}$ and $\bar{G}_{2,n}\neq G_{2},$ unless $A_{0}$ lies
in the span of the orthogonal complement of $G_{1}\otimes G_{2}$. However,
under the local alternatives (\ref{eq: R_n local}), $\bar{G}_{1,n}\rightarrow
G_{1}$ and $\bar{G}_{2,n}\rightarrow G_{2}$. This needs to be taken into
account when we characterize the asymptotic distribution of the KPST statistic
under the local alternatives in (\ref{eq: R_n local}).

The re-arranged matrix $\mathcal{R}(R)$ under H$_{1}$ is:%
\begin{equation}%
\begin{array}
[c]{rl}%
\mathcal{R}(R)= & vec(G_{1})vec(G_{2})^{\prime}+\frac{1}{\sqrt{n}}%
\mathcal{R}(A_{0})\\
= & vec(\bar{G}_{1,n})vec(\bar{G}_{2,n})^{\prime}+
vec(\bar{G}_{1,n})_{\perp}\Lambda_{n}vec(\bar{G}_{2,n})_{\perp}^{\prime},
\end{array}
\label{eq: rearranged A_n}%
\end{equation}
with 
\begin{equation}
\Lambda _{n}=vec(\bar{G}_{1,n})_{\perp }^{\prime }\mathcal{R}(R)vec(\bar{G}%
_{2,n})_{\perp }.  \label{eq: Lambda_n}
\end{equation}%
The decomposition in the last line of
(\ref{eq: rearranged A_n}) is identical to the one in (\ref{eq: rank speci}).

\begin{theorem}
\label{th: kps test power}Under local to KPS sequences of covariance
matrices as in (\ref{eq: R_n local}) and for mean zero, independently
distributed random vectors $f_{i}\in \mathbb{R}^{kp}$ with finite eighth
moments, 
\begin{equation*}
KPST\underset{d}{\rightarrow }\chi _{df}^{2}(\delta )
\end{equation*}%
as $n\rightarrow \infty $ (with $k,p$ fixed), where 
\begin{equation}
\begin{array}{rl}
\delta := & vec(a_{0})^{\prime }\left[ \left( vec(G_{2})_{\perp }^{\prime
}\otimes vec(G_{1})_{\perp }^{\prime }\right) (D_{k}\otimes D_{p})V_{R^{\ast
}}\right. \\ 
& \left. (D_{k}\otimes D_{p})^{\prime }\left( vec(G_{2})_{\perp }\otimes
vec(G_{1})_{\perp }\right) \right] ^{-}vec(a_{0}),%
\end{array}
\label{eq: explicit delta}
\end{equation}%
$V_{R^{\ast }}$ has been defined in (\ref{eq: lim_mean_var1n}), and 
\begin{equation*}
a_{0}:=vec(G_{1})_{\perp }^{\prime }\mathcal{R}(A_{0})vec(G_{2})_{\perp }\in 
\mathbb{R}^{\left( \frac{1}{2}k(k+1)-1\right) \times \left( \frac{1}{2}%
p(p+1)-1\right) }.
\end{equation*}%
\end{theorem}

\begin{proof}
See the Appendix. \smallskip
\end{proof}

\section{Simulation study on size and power}

\paragraph{Size}

We evaluate the accuracy of the limiting distribution in Theorem
\ref{th: kps test} to approximate the finite sample distribution of the
KPST\ statistic. We do so in a small simulation experiment using the linear
regression model:%
\begin{equation}
Y_{i}=Z_{i}^{\prime}\Pi+V_{i},\qquad i=1,\ldots,n, \label{eq: regression}%
\end{equation}
where $Y_{i}$ is a $p$ dimensional vector of dependent variables, $Z_{i}$ is a
$k$ dimensional vector of explanatory (exogenous) variables and $V_{i}$ is a
$p$ dimensional vector of errors. We further set $\Pi$ to zero (which is
without loss of generality because KPST uses the residual vectors) and
generate the $Z_{i}$'s independently from $N(0,I_{k})$ distributions and
$V_{i}$ given $Z_{i}$ independently from a $N\left(  0,h\left(  Z_{i}\right)
I_{p}\right)  $ distribution. We consider two different specifications of
$h(Z_{i}).$ The first leads to homoskedasticity and has $h(Z_{i})=1$ while the
second leads to (scalar) heteroskedasticity and has $h\left(  Z_{i}\right)
=\left\Vert Z_{i}\right\Vert ^{2}/k.$ For each case, we compute null rejection
probabilities (NRPs) using the three conventional nominal significance levels
of 10\%, 5\% and 1\%. The NRPs are computed using 40,000 Monte Carlo
replications for the KPST\ test that uses chi-square critical values based on
the results from Theorem \ref{th: kps test}. Table \ref{Tab: NRP} reports the
NRPs when the sample size depends on the dimensions $p$ and $k,$ specifically
$n=\left(  kp\right)  ^{16/3}$, in accordance with Theorem \ref{th: kps test}.
We notice only a slight underrejection in some cases, but in the remaining
cases the NRPs are not significantly different from the test's nominal levels.
Table \ref{Tab: NRP 2} reports NRPs with a smaller sample size $n=\left(
pk\right)  ^{4}$. In this case, we find some modest deviations from the
nominal size but these are generally quite small.

To investigate NRPs in smaller samples, Figures \ref{fig: p_2} to
\ref{fig: p_4_5} show the NRPs as a function of the sample size $n$ for
smaller sample sizes than in Tables \ref{Tab: NRP}, \ref{Tab: NRP 2} for
different settings of $p$ and $k.$ Depending on the value of the latter, the
NRPs are close to the nominal level for values of $n$ much smaller than
$\left(  pk\right)  ^{4}.$ For larger values of $pk,$ we therefore do not
(like for the smaller values of $pk)$ show the rejection frequencies all the
way up to $n=(pk)^{\frac{16}{3}},$ i.e. the value indicated by Theorem
\ref{th: kps test}d, but just to $(pk)^{4},$ which is for $p=2,$ $k=7$ at the
bottom right hand side of Figure \ref{fig: p_2}, equal to approximately
40,000, and for $p=5,$ $k=4$ at the bottom right hand side of Figure
\ref{fig: p_3} equal to 160,000 (note that the horizontal axis is in
log-scale). In many cases, the NRPs are still much closer to their nominal
significance levels than indicated by this rate. For example, when $p=k=2$ and
testing at the 5\% significance level, the NRP is close to the nominal level
for sample size of around 100. More striking is that when $p=2$ and $k=5$ the
KPST test at 5\% nominal size has NRPs close to the nominal size for values of
$n$ around 200. Figures \ref{fig: p_2}-\ref{fig: p_4_5} also show that the
KPST test generally over-rejects for small $n$. Moreover, the over-rejection is increasing in the dimensions $k$ and $p$ and can be very substantial for very small $n$, see Figure \ref{fig: p_4_5}, as is the case for any Wald test when the number of restrictions is large relative to the sample size. Therefore, it is of interest to investigate the possibility of small-sample corrections, e.g., following the bootstrap approach of \cite{Chen2019}. From a practical perspective, this over-rejection means that rejection of
KPS with small sample sizes, which happens only a few times in the
applications reported in Section \ref{s: empirical}, could be due to a
significantly higher type 1 error probability than the nominal size of the test.%
\footnote{When KPST is used as a pre-test in a two-step procedure, such as the subvector Anderson-Rubin test of \cite{GKM21}, that involves choosing a second-step test that is robust to violation of KPS when the KPST rejects in the first step, over-rejection will only affect the power but not the overall size of the two-step procedure.}

%TCIMACRO{\TeXButton{B}{\begin{table}[tbp] \centering}}%
%BeginExpansion
\begin{table}[tbp] \centering
%EndExpansion%
\begin{tabular}
[c]{rrrrr|rrr|rrr}\hline\hline
\multicolumn{5}{r}{\textit{Data Generating Process}:} &
\multicolumn{3}{|c|}{homoskedastic} & \multicolumn{3}{c}{scalar hetero}%
\\\hline
p & k & n & a & m & 10\% & 5\% & 1\% & 10\% & 5\% & 1\%\\\hline
2 & 2 & 1626 & 4 & 9 & 10.0 & 5.1 & 1.0 & 9.7 & 4.4 & 0.7\\
2 & 3 & 14130 & 10 & 18 & 10.0 & 5.0 & 0.8 & 9.3 & 4.2 & 0.7\\
2 & 4 & 65536 & 18 & 30 & 9.4 & 5.0 & 0.9 & 9.7 & 4.9 & 0.9\\
2 & 5 & 215444 & 28 & 45 & 9.8 & 4.7 & 0.9 & 9.8 & 5.1 & 1.0\\
3 & 2 & 14130 & 10 & 18 & 10.2 & 5.0 & 0.9 & 10.0 & 4.7 & 0.9\\\hline
3 & 3 & 122827 & 25 & 36 & 9.7 & 4.9 & 1.0 & 9.8 & 5.0 & 0.9\\\hline\hline
\end{tabular}
\caption{Rejection frequencies (in percentages) of KPST test at various
significance levels. $\chi^2_{df}$ critical values.
$[n=(pk)^{16/3}]$, $df$: number of restrictions given in eq.~(\ref{eq:
a_spec}), $m$: number of estimated parameters. Computed using 40,000 MC replications.}\label{Tab: NRP}%
%TCIMACRO{\TeXButton{E}{\end{table}}}%
%BeginExpansion
\end{table}%
%EndExpansion
%

%TCIMACRO{\TeXButton{B}{\begin{table}[tbp] \centering}}%
%BeginExpansion
\begin{table}[tbp] \centering
%EndExpansion%
\begin{tabular}
[c]{rrrrr|rrr|rrr}\hline\hline
\multicolumn{5}{r}{\textit{Data Generating Process}:} &
\multicolumn{3}{|c|}{homoskedastic} & \multicolumn{3}{c}{scalar hetero}%
\\\hline
p & k & n & a & m & 10\% & 5\% & 1\% & 10\% & 5\% & 1\%\\\hline
2 & 2 & 256 & 4 & 9 & 11.2 & 5.3 & 0.9 & 11.4 & 4.8 & 0.5\\
2 & 3 & 1296 & 10 & 18 & 10.2 & 4.9 & 0.9 & 9.3 & 4.0 & 0.5\\
2 & 4 & 4096 & 18 & 30 & 9.9 & 5.1 & 1.0 & 9.1 & 4.2 & 0.8\\
2 & 5 & 10000 & 28 & 45 & 9.7 & 4.6 & 0.8 & 8.8 & 4.0 & 0.6\\
2 & 6 & 20736 & 40 & 63 & 10.0 & 5.1 & 1.0 & 9.5 & 4.5 & 0.7\\
2 & 7 & 38416 & 54 & 84 & 9.8 & 4.8 & 0.9 & 9.5 & 4.5 & 0.8\\
3 & 2 & 1296 & 10 & 18 & 9.9 & 4.8 & 0.7 & 9.0 & 3.7 & 0.5\\
3 & 3 & 6561 & 25 & 36 & 9.8 & 5.0 & 0.9 & 9.6 & 4.4 & 0.7\\
3 & 4 & 20736 & 45 & 60 & 10.7 & 5.6 & 1.2 & 10.2 & 5.1 & 0.9\\
3 & 5 & 50625 & 70 & 90 & 10.4 & 5.2 & 1.0 & 10.2 & 5.0 & 0.7\\\hline
3 & 6 & 104976 & 100 & 126 & 10.2 & 5.0 & 1.1 & 10.1 & 5.0 & 1.0\\\hline
3 & 7 & 194481 & 135 & 168 & 10.2 & 5.0 & 1.0 & 10.0 & 5.0 & 1.0\\\hline\hline
\end{tabular}
\caption{Rejection frequencies (in percentages) of KPST test at various
significance levels. $\chi^2_{df}$ critical values.
$n=(pk)^4$, $df$: number of restrictions given in eq.~(\ref{eq:
a_spec}), $m$: number of estimated parameters. Computed using 40,000 MC replications.}\label{Tab: NRP 2}%
%TCIMACRO{\TeXButton{E}{\end{table}}}%
%BeginExpansion
\end{table}%
%EndExpansion

\begin{figure}[ptb]%
\centering
\includegraphics[
height=3.6288in,
width=5.4293in
]%
{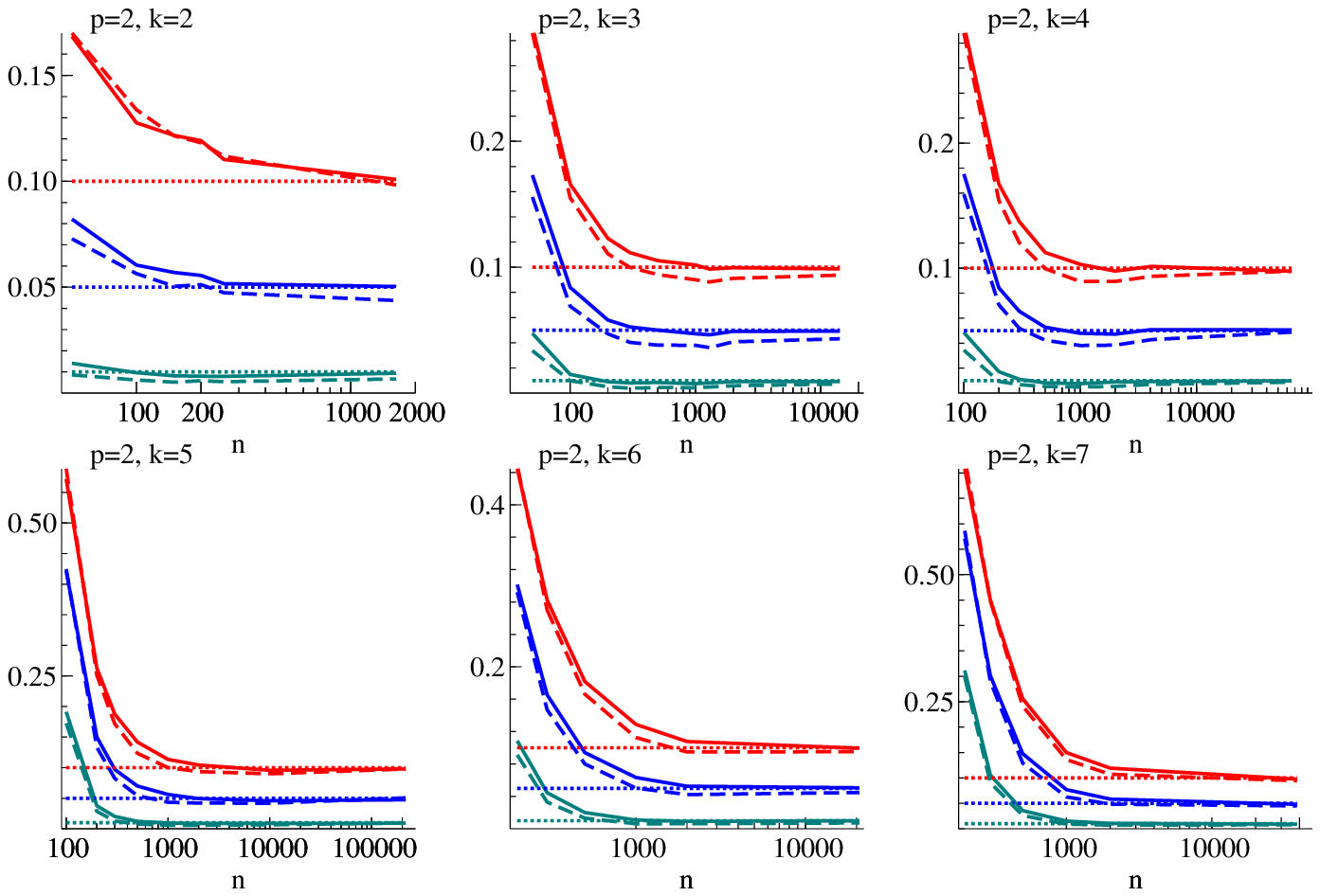}%
\caption{Null rejection probabilities of KPST test as a function of sample size $n$ at different significance levels: 10\% (red), 5\% (blue) and 1\% (green); and different data generating processes: homoskedastic (solid) and scalar heteroskedastic (dashed). Computed using 40,000 MC replications.}%
\label{fig: p_2}%
\end{figure}

\begin{figure}[ptb]%
\centering
\includegraphics[
height=3.6288in,
width=5.4293in
]%
{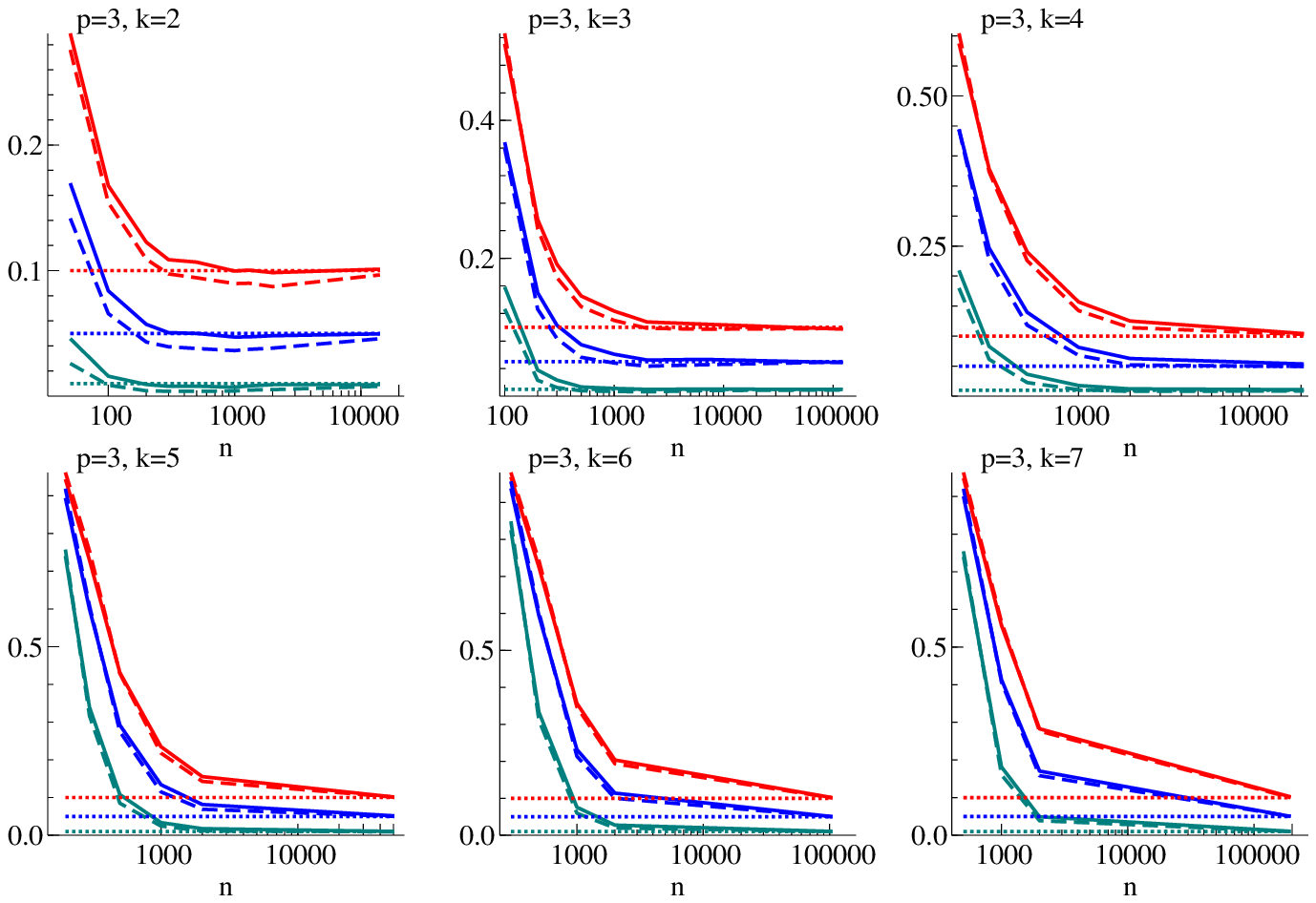}%
\caption{Null rejection probabilities of KPST test as a function of sample size $n$ at different significance levels: 10\% (red), 5\% (blue) and 1\% (green); and different data generating processes: homoskedastic (solid) and scalar heteroskedastic (dashed). Computed using 40,000 MC replications.}%
\label{fig: p_3}%
\end{figure}

\begin{figure}[ptb]%
\centering
\includegraphics[
height=3.224in,
width=4.83in
]%
{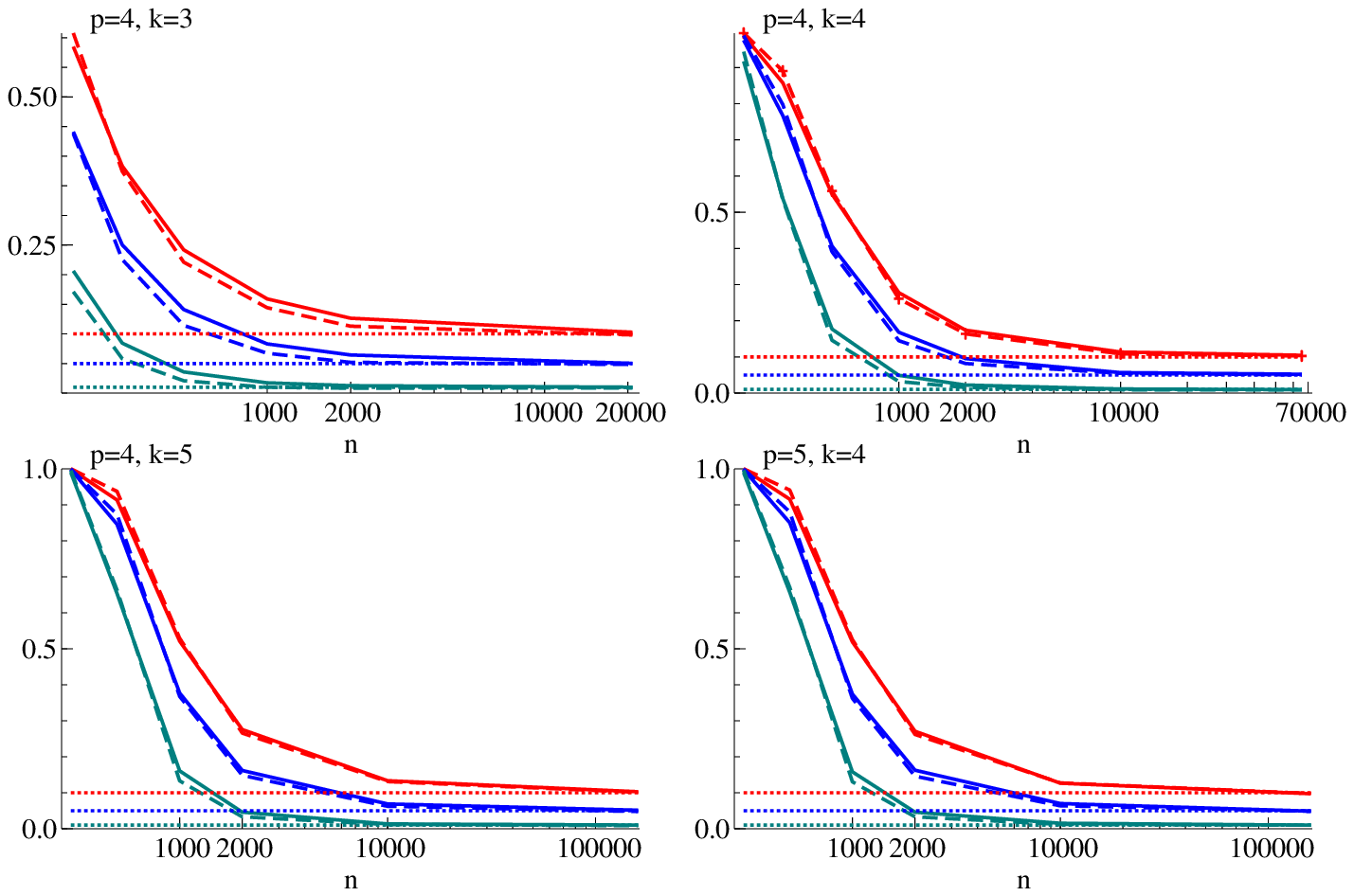}%
\caption{Null rejection probabilities of KPST test as a function of sample size $n$ at different significance levels: 10\% (red), 5\% (blue) and 1\% (green); and different data generating processes: homoskedastic (solid) and scalar heteroskedastic (dashed). Computed using 40,000 MC replications.}%
\label{fig: p_4_5}%
\end{figure}

\paragraph{Power}

We simulate the power of the KPST test using the asymptotic $\chi^{2}$
critical values stated in Theorem \ref{th: kps test}. The Data Generating
Process (DGP) is generated by a model with $p=k=2,$ where $Y_{i}=Z_{i}%
\Pi+V_{i}$ and $\Pi=0,$ see (\ref{eq: regression}). The two dimensional
vectors containing the regressors $Z_{i}$ and errors $V_{i}$ are simulated
according to:
\begin{equation}
V_{i}\sim iid\left\{
\begin{array}
[c]{l}%
N\left(  0,\Omega_{1}\right)  ,\\
N\left(  0,\Omega_{2}\right)  ,
\end{array}
\right.  \text{ }Z_{i}\sim iid\left\{
\begin{array}
[c]{l}%
N\left(  0,Q_{zz,1}\right)  ,\quad i=1,...,\left[  n/2\right] \\
N\left(  0,Q_{zz,2}\right)  ,\quad i=[n/2]+1,...,n,
\end{array}
\right.  \label{eq: inst_errors}%
\end{equation}
with $\Omega_{1}=diag\left(  b,1\right)  ,$ $\Omega_{2}=diag\left(
1,b\right)  ,$ $Q_{zz,1}=diag\left(  1,c\right)  ,$ $Q_{zz,2}=diag\left(
c,1\right)  ,$ and%
\begin{equation}
b:=\frac{1}{2}\frac{\sigma}{\sqrt{n}}-\frac{1}{2}\sqrt{\frac{\sigma}{\sqrt{n}%
}\left(  \frac{\sigma}{\sqrt{n}}+8\right)  }+1,\quad c:=\frac{1}{2}%
\frac{\sigma}{\sqrt{n}}+\frac{1}{2}\sqrt{\frac{\sigma}{\sqrt{n}}\left(
\frac{\sigma}{\sqrt{n}}+8\right)  }+1,
\end{equation}
for $\sigma\in\lbrack0,\sqrt{n}).$ The covariance matrix $R$ is then such
that:
\begin{equation}%
\begin{array}
[c]{rl}%
R= & \frac{1}{n}var\left(  \sum_{i=1}^{n}(V_{i}\otimes Z_{i})\right)
=\frac{1}{2}diag\left(  b+c,1+bc,1+bc,b+c\right) \\
= & \underbrace{I_{4}}_{G_{1}\otimes G_{2}}+\frac{\sigma}{\sqrt{n}}{\times
}diag\left(  1,-1,-1,1\right)  ,
\end{array}
\label{eq: R_N}%
\end{equation}
and $G_{1}=G_{2}=I_{2}$. Because
\begin{equation}
\mathcal{R(}diag\left(  1,-1,-1,1\right)  )=\left(
\begin{array}
[c]{cccc}%
1 & 0 & 0 & -1\\
0 & 0 & 0 & 0\\
0 & 0 & 0 & 0\\
-1 & 0 & 0 & 1
\end{array}
\right)  ,
\end{equation}
$vec(G_{1})^{\prime}\mathcal{R(}diag\left(  1,-1,-1,1\right)  )vec(G_{2})=0,$
the re-arranged specification of $R$ in (\ref{eq: rearranged A_n}) equals$:$%
\begin{equation}%
\begin{array}
[c]{rl}%
\mathcal{R}(R)= & vec(G_{1})vec(G_{2})^{\prime}+\frac{\sigma}{\sqrt{n}}\left(
\begin{array}
[c]{cccc}%
1 & 0 & 0 & -1\\
0 & 0 & 0 & 0\\
0 & 0 & 0 & 0\\
-1 & 0 & 0 & 1
\end{array}
\right) \\
= & vec(G_{1})vec(G_{2})^{\prime}+\frac{1}{\sqrt{n}}vec(G_{1})_{\perp}%
a_{0}vec(G_{2})_{\perp}^{\prime},
\end{array}
\end{equation}
where
\begin{equation}%
\begin{array}
[c]{c}%
vec(G_{1})_{\perp}=vec(G_{2})_{\perp}=\frac{1}{\sqrt{2}}\left(
\begin{array}
[c]{ccc}%
1 & 0 & 0\\
0 & \sqrt{2} & 0\\
0 & 0 & \sqrt{2}\\
-1 & 0 & 0
\end{array}
\right)  ,\text{ }a_{0}=\sigma\left(
\begin{array}
[c]{ccc}%
1 & 0 & 0\\
0 & 0 & 0\\
0 & 0 & 0
\end{array}
\right)  =\sigma e_{1}e_{1}^{\prime},\text{ }e_{1}=(1,0,0)^{\prime},
\end{array}
\end{equation}
is such that the local deviation from KPS lies in the orthogonal complement of
$vec(G_{1})$ and $vec(G_{2}).$ The non-centrality parameter of the non-central
$\chi^{2}$ limiting distribution follows from (\ref{eq: explicit delta}). Note
that
\begin{equation}%
\begin{array}
[c]{l}%
(e_{1}\otimes e_{1})^{\prime}\left[  (\left[  vec(G_{2})\right]  _{\perp
}^{\prime}\otimes\left[  vec(G_{1})\right]  _{\perp}^{\prime})\right. \\
\left.  \text{cov}\left(  vec\left(  \mathcal{R}(\hat{R})\right)  \right)
(\left[  vec(G_{2})\right]  _{\perp}\otimes\left[  vec(G_{1})\right]  _{\perp
})\right]  ^{-}(e_{1}\otimes e_{1})=\frac{1}{4},
\end{array}
\label{Eq: non-central1}%
\end{equation}
where $G_{i}=I_{2}$ for $i=1,2$. Note also that $vec\left(  a_{0}\right)
=2\sigma(e_{1}\otimes e_{1})$. Thus, the non-centrality parameter is%
\begin{equation}%
\begin{array}
[c]{c}%
\delta=\frac{1}{4}\sigma^{2}.
\end{array}
\label{Eq: non-central2}%
\end{equation}
For $\sigma=0,$ $R$ has KPS, so the null hypothesis in (\ref{eq: KPS hyp})
holds. For the limiting case of $\sigma=\sqrt{n}:$ $b=0,$ so $\Omega_{1}$ and
$\Omega_{2}$ are singular.

We compute the power function of the KPST test at three significance levels
10\%, 5\% and 1\% using 10,000 Monte Carlo replications. For comparison, we
also compute the power of the non-invariant KPST$^{\ast}$ test that rejects
H$_{0}$ if the statistic KPST$^{\ast}$ in (\ref{eq: kpst*}) exceeds the
corresponding $1-\alpha$ quantile of $\chi_{df}^{2}$ with degrees of freedom
$df$ given in Theorem \ref{th: kps test}b, which are the same critical values
as for the KPST test in (\ref{eq: KPST test}). The results are reported
graphically in Figure \ref{fig: power}. The left-hand-side graphs in Figure
\ref{fig: power} show that for a moderate sample of size $n=200$ both tests
have good and essentially identical power. Moreover, as the sample size
increases, the power function of both tests approaches the noncentral
$\chi^{2}$ asymptotic approximation in Theorem \ref{th: kps test power}
indicated in blue on the right-hand-side graphs of Figure \ref{fig: power} for
$n=100,000.$ Results for other sample sizes are qualitatively similar and are
omitted in the interest of brevity. In particular, KPST has nontrivial power
even for small samples.%

\begin{figure}[ptb]%
\centering
\includegraphics[
height=4.4222in,
width=6.5303in
]%
{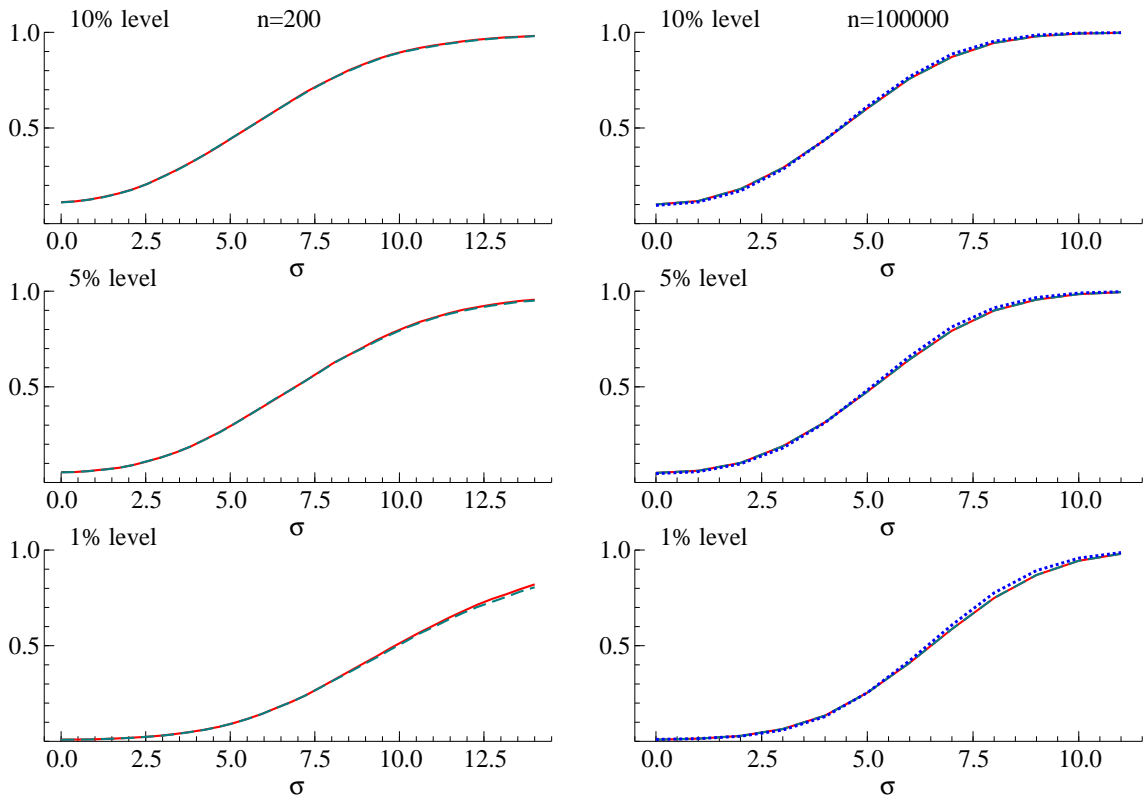}%
\caption{Power of KPST (solid red) and KPST$^{\ast}$ (dashed green) tests with
sample size $n=200$ (left) and $n=10^{5}$ (right). Asymptotic approximation
from Theorem \ref{th: kps test power} (dotted blue) superimposed on the right.
$\sigma$ measures deviation from KPS in Frobenius norm. Computed using 10,000
Monte Carlo replications.}%
\label{fig: power}%
\end{figure}

\section{Empirical applications\label{s: empirical}}

We investigate whether KPS covariance matrices are potentially relevant for
applied work. To do so, we apply the KPST\ test to the covariance matrices of
estimators in published empirical studies. We consider fifteen highly cited
papers conducting linear IV regressions from top journals in economics and
test for KPS of the joint covariance matrix of the (unrestricted reduced form)
least squares estimators which result from regressing all endogenous variables
on the instruments.\footnote{Both the endogenous variables and the instruments
are first regressed on the control, or included exogenous, variables and only
the residuals from these regressions are used.} Tables \ref{tab:KPS app} and
\ref{tab:clust} in the Supplementary Appendix report the results of the KPST
test for the 118 different specifications we analyzed. Table \ref{tab:KPS app}
does so for the studies using independent data (sixty specifications) while
Table \ref{tab:clust} lists the results for studies with clustered data (fifty
eight specifications). Because these tables are rather extensive, Tables
\ref{TableKey nocluster} and \ref{TableKey cluster} report a summary of our
findings on the KPST\ tests.

Table \ref{TableKey nocluster}, summarizing our results on KPS tests for the
papers using independent data, shows considerable support for KPS covariance
matrices especially when the number of observations is not too large. For the
60 different specifications using independent data reported in Table
\ref{TableKey nocluster}, KPS is rejected at the 5\% nominal size for only
about one third of them, namely for 22.

Table \ref{TableKey cluster}, summarizing the test results for papers using
clustered data, shows that for the 58 different specifications with clustered
data, KPS\ is rejected at the 5\% nominal size for 46 specifications when
using the unrestricted covariance matrix estimator (\ref{eq: sample cov}) and
for 40 when using the clust qered covariance matrix estimator
(\ref{eq: cluster covariance}). The number of observations in the involved
papers using clustered data is typically much larger than for the papers using
independent observations which largely explains our different findings for
independent compared to clustered observations.

Summarising, our analysis of the KPS of covariance matrices of moment
condition vectors in a considerable number of prominent empirical studies
shows that KPS is often not rejected especially for moderate sample sizes.

\nocite{ACJR2011}\nocite{AD2013}\nocite{ADG2013}\nocite{AGN2013}%
\nocite{AJ2005}\nocite{AJRY2008}

\nocite{DL2012}\nocite{DT2011}\nocite{HG2010}\nocite{JPS2006}\nocite{MSS2004}%
\nocite{Nunn2008}

\nocite{PSJM2013}\nocite{TCN2010}\nocite{Voors2012}\nocite{Yogo2004}\bigskip%

%TCIMACRO{\TeXButton{B}{\begin{table}[tbp] \centering}}%
%BeginExpansion
\begin{table}[tbp] \centering
%EndExpansion%
\begin{tabular}
[c]{|l|c|c|c|}\hline
\textbf{Paper} & \#\textbf{ specifications} & \textbf{KPS rejection} &
\textbf{\# observations}\\\hline
\cite{TCN2010} & 2 & none & moderate\\
\cite{Nunn2008} & 4 & 4 & small\\
\cite{AJ2005} & 24 & 10 & small\\
\cite{HG2010} & 2 & 2 & huge\\
\cite{AGN2013} & 6 & 1 & moderate\\
\cite{Yogo2004} & 22 & 5 & moderate\\\hline
\end{tabular}
\caption{Summary of results of 5\% significance level KPST tests for specifications in papers using independent observations}
\label{TableKey nocluster}%
%TCIMACRO{\TeXButton{E}{\end{table}}}%
%BeginExpansion
\end{table}%
%EndExpansion
%

%TCIMACRO{\TeXButton{B}{\begin{table}[tbp] \centering}}%
%BeginExpansion
\begin{table}[tbp] \centering
%EndExpansion%
\begin{tabular}
[c]{|l|c|c|c|c|c|}\hline
&  &  &  & \textbf{clustered} & \\
\textbf{Paper} & \#\textbf{ specific.} & \textbf{KPS rej.} & \textbf{\# obs.}
& \textbf{KPS rej.} & \textbf{\# clusters}\\\hline
\cite{DT2011} & 8 & 6 & large & 5 & moderate\\
\cite{AJRY2008} & 9 & 7 & large & 5 & moderate\\
\cite{JPS2006} & 4 & 4 & huge & 4 & huge\\
\cite{PSJM2013} & 2 & 2 & huge & 2 & huge\\
\cite{ADG2013} & 18 & 18 & large & 13 & small\\
\cite{AD2013} & 7 & 7 & huge & 7 & small\\
\cite{ACJR2011} & 1 & 1 & small & 1 & very small\\
\cite{MSS2004} & 3 & 0 & large & 3 & small\\
\cite{Voors2012} & 6 & 1 & moderate & 0 & small\\\hline
\end{tabular}
\caption{Summary of results of 5\% signficance level KPST tests for specifications in papers using clustered observations}
\label{TableKey cluster}%
%TCIMACRO{\TeXButton{E}{\end{table}}}%
%BeginExpansion
\end{table}%
%EndExpansion

\section{Conclusion}

We propose a test for the null of a covariance matrix of a vector of moment
equations to have a KPS. The test is an extension of the
\cite{kleibergen2006grr}\ rank test and is easy to use. We apply it to data
used in a considerable number of prominent applied studies conducting IV
regressions and find that KPS of the covariance matrix of the least squares
estimator of the unrestricted reduced form is often not rejected for moderate
sample sizes. In linear IV regression, a KPS covariance matrix brings
considerable advantages for both computation and inference in weakly
identified settings. Given the common occurrence of weak identification in
applications, our empirical findings underscore the contribution that the use
of KPS covariance matrices can make in applied work.

In a companion paper, \cite{GKM21}, we develop a two-step test procedure that
in the first step uses the new KPS covariance matrix test and, depending on
its outcome, in the second step conducts a weak-identification-robust test on
a subset of the structural parameters based either on an improved powerful
subvector AR test or based on the AR/AR test that is robust to arbitrary forms
of conditional heteroskedasticity. The two-step procedure is constructed such
that its asymptotic size is bounded by the nominal size. A promising area for
application of testing for KPS is in linear factor models for establishing
risk premia. The default setting in this area is to assume homoskedasticity
and weak identification is often present.

To further improve the approximation of the finite sample distribution of the
KPST statistic, it would also be of interest to investigate whether the
bootstrap can deliver refinements as \cite{Chen2019} show for rank tests on
general matrices. We leave this important extension for future work.

\newpage\appendix

\section*{\textbf{\noindent Appendix}}

\section{Proofs}

\textbf{\noindent Proof of Theorem \ref{th: frob norm}:} For a {given nonzero
matrix $A\in\mathbb{R}^{m\times n}$ with SVD 
$A=Udiag(\sigma_{1},...,\sigma_{p})V^{\prime}$ for singular values $\sigma
_{1}\geq\sigma_{2}\geq...\geq\sigma_{p}\geq0$ with $p=\min\{m,n\}$,
rectangular $diag(\sigma_{1},...,\sigma_{p})\in$}$\mathbb{R}^{m\times n}${,
orthogonal matrices $U=[u_{1},...,u_{m}]\in\mathbb{R}^{m\times m},$ and
$V=[v_{1},...,v_{n}]\in\mathbb{R}^{n\times n},$ }Lemma 2 in \cite{GKM21}
states that {a minimizing argument in }t{he minimization problem $\ $}%
\begin{equation}
{\min_{B\in\mathbb{R}^{m\times n},\text{ }rk(B)=1}||A-B||_{F}^{2}}
\label{lemma2}%
\end{equation}
{\ is given by }$\widehat{{B}}${$=\sigma_{1}u_{1}v_{1}^{\prime}$ and the
minimum equals $%
%TCIMACRO{\tsum \nolimits_{i=2}^{p}}%
%BeginExpansion
{\textstyle\sum\nolimits_{i=2}^{p}}
%EndExpansion
\sigma_{i}^{2}.$ Furthermore, it is shown that if $\sigma_{1}>\sigma_{2}$ then
}$\widehat{{B}}${$=\sigma_{1}u_{1}v_{1}^{\prime}$ is the unique minimizer.}

{Theorem 5.8 in \cite{vanloanpit93} states that if }$A\in\mathbb{R}^{pk\times
pk}$ is symmetric and positive definite then minimizers $\tilde{G}_{1}$ and
$\tilde{G}_{2}$ for the problem
\begin{equation}
\min_{G_{1}\in\mathbb{R}^{p\times p},G_{2}\in\mathbb{R}^{k\times k}}%
||A-G_{1}\otimes G_{2}||_{F}^{2} \label{theorem 5.8}%
\end{equation}
exist that are also symmetric and positive definite. Because $\hat{R}$ is
symmetric by construction and positive definite by assumption,
(\ref{theorem 5.8}) and $||\hat{R}-G_{1}\otimes G_{2}||_{F}=||\mathcal{R(}%
\hat{R}\mathcal{)}-vec(G_{1})vec(G_{2})^{\prime}||_{F}$ (which holds by\ by
Theorem 2.1 in {\cite{vanloanpit93})} imply that symmetric positive definite
matrices $\hat{G}_{1}$ and $\hat{G}_{2}$ exist that minimize $||\mathcal{R(}%
\hat{R}\mathcal{)}-vec(G_{1})vec(G_{2})^{\prime}||_{F}^{2}$ over $G_{1}%
\in\mathbb{R}^{p\times p},G_{2}\in\mathbb{R}^{k\times k}.$ Therefore, because
the rank of $vec(\hat{G}_{1})vec(\hat{G}_{2})^{\prime}$ is one, $\widehat{B}%
:=vec(\hat{G}_{1})vec(\hat{G}_{2})^{\prime}$ is a minimizer in the problem
${\min_{B\in\mathbb{R}^{p^{2}\times k^{2}},\text{ }rk(B)=1}||\mathcal{R(}%
\hat{R}\mathcal{)}-B||_{F}^{2}}$. But by (\ref{lemma2}) with $A$ playing the
role of $\mathcal{R(}\hat{R}\mathcal{)}$ we know that the minimum equals {$%
%TCIMACRO{\tsum \nolimits_{i=2}^{\min\{p^{2},k^{2}\}}}%
%BeginExpansion
{\textstyle\sum\nolimits_{i=2}^{\min\{p^{2},k^{2}\}}}
%EndExpansion
\hat{\sigma}_{i}^{2},$ where }${\hat{\sigma}}${$_{1}\geq\hat{\sigma}_{2}%
\geq...\geq\hat{\sigma}_{\min\{p^{2},k^{2}\}}\geq0$ denote the singular values
of }$\mathcal{R(}\hat{R}\mathcal{)}\in\mathbb{R}^{p^{2}\times k^{2}}.$ This
establishes the claimed formula for $DS^{2}$.

Next, if $\hat{\sigma}_{1}>\hat{\sigma}_{2}$, then, {by the uniqueness part of
}Lemma 2 in \cite{GKM21},{ the minimizing }$\hat{G}_{1}$ and $\hat{G}_{2}$ in
{(\ref{theorem 5.8}) satisfy }$vec(\hat{G}_{1})vec(\hat{G}_{2})^{\prime}%
={\hat{\sigma}_{1}}\hat{L}_{1}\hat{N}_{1}^{\prime}.$ There are many different
ways we can define the argmins such that $vec\left(  \hat{G}_{1}\right)  $ and
$vec\left(  \hat{G}_{2}\right)  $ are proportional to $\hat{L}_{1}$ and
$\hat{N}_{1},$ respectively. Because sign-definiteness and symmetry are not
affected by normalization by a positive constant, any normalization will
produce estimates $\hat{G}_{1},\hat{G}_{2}$ that are symmetric and positive
definite. The normalization we use is $vec(\hat{G}_{1})=\hat{L}_{1}/\hat{L}%
_{11},$ so that $vec(\hat{G}_{2})=\hat{L}_{11}\hat{\sigma}_{1}\hat{N}_{1}$,
i.e., the upper left element of $\hat{G}_{1}$ is normalized to 1. We know that
$\hat{L}_{11}\neq0$ for otherwise, one diagonal block of the covariance matrix
$\hat{R}$ would be zero, which would contradict the assumption that $\hat{R}$
is positive definite. This establishes (\ref{eq: g1g2 est}). \bigskip

\noindent\textbf{Proof of Equation (\ref{eq: simple KPST}). }Because $\hat
{L}_{22}^{-1}(\hat{L}_{22}\hat{L}_{22}^{\prime})^{1/2}$ and $\hat{N}_{22}%
^{-1}\left(  \hat{N}_{22}\hat{N}_{22}^{\prime}\right)  ^{1/2}$ are invertible
which follows from expressions stated in the proof of Theorem 2a, KPST can be
rewritten as:%
\[%
\begin{array}
[c]{rl}%
KPST= & n\times\left[  vec(vec(\hat{G}_{1})_{\perp}^{\prime}\mathcal{R}%
(\hat{R})vec(\hat{G}_{2})_{\perp})\right]  ^{\prime}\\
& \left[  \left(  vec(\hat{G}_{2})_{\perp}^{\prime}\otimes vec(\hat{G}%
_{1})_{\perp}^{\prime}\right)  \hat{V}\left(  vec(\hat{G}_{2})_{\perp}\otimes
vec(\hat{G}_{1})_{\perp}\right)  \right]  ^{-}\\
& \left[  vec(vec(\hat{G}_{1})_{\perp}^{\prime}\mathcal{R}(\hat{R})vec(\hat
{G}_{2})_{\perp})\right] \\
= & \left(  vec\left(  \hat{L}_{2}^{\prime}\hat{L}\hat{\Sigma}\hat{N}^{\prime
}\hat{N}_{2}\right)  \right)  ^{\prime}\left(  \hat{N}_{22}^{-1}\left(  \hat
{N}_{22}\hat{N}_{22}^{\prime}\right)  ^{1/2}\otimes\hat{L}_{22}^{-1}(\hat
{L}_{22}\hat{L}_{22}^{\prime})^{1/2}\right) \\
& \left[  \left(  \left(  \hat{N}_{22}\hat{N}_{22}^{\prime}\right)  ^{1/2}%
\hat{N}_{22}^{\prime-1}\otimes(\hat{L}_{22}\hat{L}_{22}^{\prime})^{1/2}\bar
{L}_{22}^{-1\prime}\right)  \right.  \left(  \hat{N}_{2}\otimes\hat{L}%
_{2}\right)  ^{\prime}(D_{k}\otimes D_{p})\hat{V}_{\hat{R}^{\ast}}\\
& \left.  (D_{k}\otimes D_{p})^{\prime}\left(  \hat{N}_{2}\otimes\hat{L}%
_{2}\right)  \left(  \hat{N}_{22}^{-1}\left(  \hat{N}_{22}\hat{N}_{22}^{\prime
}\right)  ^{1/2}\otimes\hat{L}_{22}^{-1}(\hat{L}_{22}\hat{L}_{22}^{\prime
})^{1/2}\right)  \right]  ^{-}\\
& \left(  \left(  \hat{N}_{22}\hat{N}_{22}^{\prime}\right)  ^{1/2}\hat{N}%
_{22}^{-1\prime}\otimes(\hat{L}_{22}\hat{L}_{22}^{\prime})^{1/2}\hat{L}%
_{22}^{-1\prime}\right)  \left(  vec\left(  \hat{L}_{2}^{\prime}\hat{L}%
\hat{\Sigma}\hat{N}^{\prime}\hat{N}_{2}\right)  \right) \\
= & n\times\left(  vec\left(  \hat{\Sigma}_{2}\right)  \right)  ^{\prime
}\left[  (\hat{N}_{2}\otimes\hat{L}_{2})^{\prime}\hat{V}\left(  \hat{N}%
_{2}\otimes\hat{L}_{2}\right)  \right]  ^{-}\left(  vec\left(  \hat{\Sigma
}_{2}\right)  \right)  .
\end{array}
\smallskip
\]

\noindent\textbf{Proof of Theorem \ref{th: kps test}a: }The hypothesis of
interest in (\ref{eq: hyp rank}) is: H$_{0}:vec(G_{1})_{\perp}^{\prime
}\mathcal{R}(R)vec(G_{2})_{\perp}=0.$ We test this hypothesis using a SVD of
$\mathcal{R}(\hat{R}):$%
\[
\mathcal{R}(\hat{R})=\hat{L}\hat{\Sigma}\hat{N}^{\prime},
\]
whose elements using (\ref{eq: R-spec}) and the orthonormality of $\hat{L}$ and $\hat{N}$ can be specified as%
\[%
\begin{array}
[c]{rlllrl}%
\hat{L}= & (D_{p}\hat{A}\text{ }\vdots\text{ }D_{p\perp}), &  &  & \hat{N}= &
(D_{k}\hat{B}\text{ }\vdots\text{ }D_{k\perp}),\\
\hat{\Sigma}= & \left(
\begin{array}
[c]{cc}%
\hat{\sigma}_{1} & 0\\
0 & \hat{\Sigma}_{2}%
\end{array}
\right)  , &  &  & \hat{\Sigma}_{2}= & \left(
\begin{array}
[c]{cc}%
\hat{\Sigma}_{22} & 0\\
0 & 0
\end{array}
\right)  ,
\end{array}
\]
where $\hat{A}$ is a $\frac{1}{2}p(p+1)\times\frac{1}{2}p(p+1)$ dimensional
matrix, $\hat{A}^{\prime}D_{p}^{\prime}D_{p}\hat{A}=I_{\frac{1}{2}p(p+1)},$
$\hat{B}$ is a $\frac{1}{2}k(k+1)\times\frac{1}{2}k(k+1)$ dimensional matrix,
$\hat{B}^{\prime}D_{k}^{\prime}D_{k}\hat{B}=I_{\frac{1}{2}k(k+1)},$
$\hat{\Sigma}_{22}$ is a diagonal $(\frac{1}{2}p(p+1)-1)\times(\frac{1}%
{2}k(k+1)-1)$ dimensional matrix, $D_{p\perp}$ and $D_{k\perp}$ are
$p^{2}\times\frac{1}{2}p(p-1)$ and $k^{2}\times\frac{1}{2}k(k-1)$ dimensional
matrices which are the orthogonal complements of $D_{p}$ and $D_{k},$
$D_{p}^{\prime}D_{p\perp}\equiv0,$ $D_{p\perp}^{\prime}D_{p\perp}\equiv
I_{\frac{1}{2}p(p-1)},$ $D_{k}^{\prime}D_{k\perp}\equiv0$ and $D_{k\perp
}^{\prime}D_{k\perp}\equiv I_{\frac{1}{2}k(k-1)}.$ We also use an identical
SVD of the population counterpart $\mathcal{R}(R)$ of $\mathcal{R}(\hat{R}):$
\[
\mathcal{R}(R)=L\Sigma N^{\prime},
\]
with an identical specification of its elements (but without \textquotedblleft%
$\symbol{94}$\textquotedblright) and where under H$_{0}:\Sigma_{22}=0.$

To obtain the limit distribution of the sample analog of the parameter tested
under H$_{0}$ recall from below (\ref{eq: lambda}) that%
\begin{equation}%
\begin{array}
[c]{rl}%
\hat{\Lambda}= & vec(\hat{G}_{1})_{\perp}^{\prime}\mathcal{R}(\hat{R}%
)vec(\hat{G}_{2})_{\perp}.
\end{array}
\label{eq: lambdahat}%
\end{equation}
Next, we use that $vec(\hat{G}_{1})_{\perp}=vec(G_{1})_{\perp}+O_{p}%
(n^{-\frac{1}{2}}),$ $vec(\hat{G}_{2})_{\perp}=vec(G_{2})_{\perp}%
+O_{p}(n^{-\frac{1}{2}})$ (which holds under our imposed conditions, see
Kleibergen and Paap (2006)), the assumption $\hat{R}^{\ast}=R^{\ast}+\frac
{1}{\sqrt{n}}\Psi+o_{p}(n^{-\frac{1}{2}}),$ and $D_{p}R^{\ast}D_{k}^{\prime
}=vec(G_{1})vec(G_{2})^{\prime}$ which holds under H$_{0}.$ Thus, under
H$_{0}$
\[%
\begin{array}
[c]{rl}%
\hat{\Lambda}= & \left[  vec(G_{1})_{\perp}+O_{p}(n^{-\frac{1}{2}})\right]
^{\prime}\left[  vec(G_{1})vec(G_{2})^{\prime}+\frac{1}{\sqrt{n}}D_{p}\Psi
D_{k}^{\prime}+o_{p}(n^{-\frac{1}{2}})\right]  \left[  vec(G_{2})_{\perp
}+O_{p}(n^{-\frac{1}{2}})\right]  \\
= & \frac{1}{\sqrt{n}}vec(G_{1})_{\perp}^{\prime}D_{p}\Psi D_{k}^{\prime
}vec(G_{2})_{\perp}+o_{p}(n^{-\frac{1}{2}}).
\end{array}
\]
To construct the limit distribution of $\hat{\Lambda}$, recall that $A$ and $B$ were defined from $L= (D_{p}A\text{ }\vdots\text{ }D_{p\perp})$ and $N= 
(D_{k}B\text{ }\vdots\text{ }D_{k\perp})$,
and partition them as
\[
A = \left(
\begin{array}
[c]{cc}%
a_{1} & A_{2}%
\end{array}
\right), \qquad B =  \left(
\begin{array}
[c]{cc}%
b_{1} & B_{2}%
\end{array}
\right),
\]
where $a_{1}:\frac{1}{2}p(p+1)\times1,$ $A_{2}:\frac{1}{2}p(p+1)\times(\frac{1}%
{2}p(p+1)-1),$ $b_{1}:\frac{1}{2}k(k+1)\times1,$ $B_{2}:\frac{1}%
{2}k(k+1)\times(\frac{1}{2}k(k+1)-1)$. Then,
\[%
\begin{array}
[c]{rlllrl}%
vec(G_{1})_{\perp}= & L_{2}L_{22}^{-1}(L_{22}L_{22}^{\prime})^{1/2}, &  &  &
vec(G_{2})_{\perp}= & N_{2}N_{22}^{-1}\left(  N_{22}N_{22}^{\prime}\right)
^{1/2},\\
L_{2}= & \left(
\begin{array}
[c]{cc}%
e_{1,\frac{1}{2}p(p+1)}^{\prime}A_{2} & 0\\
D_{2,p}A_{2} & D_{2,p\perp}%
\end{array}
\right)  , &  &  & N_{2}= & \left(
\begin{array}
[c]{cc}%
e_{1,\frac{1}{2}k(k+1)}^{\prime}B_{2} & 0\\
D_{2,k}B_{2} & D_{2,k\perp}%
\end{array}
\right)  ,\\
L_{22}= & \left(
\begin{array}
[c]{cc}%
D_{2,p}A_{2} & D_{2,p\perp}%
\end{array}
\right)  , &  &  & N_{22}= & \left(
\begin{array}
[c]{cc}%
D_{2,k}B_{2} & D_{2,k\perp}%
\end{array}\right),
\end{array}
\]
where we use that $D_{p}=(e_{1,\frac{1}{2}p(p+1)}$ $\vdots$ $D_{2,p}^{\prime
})^{\prime},$ for  $D_{2,p}:(p^{2}-1)\times\frac{1}{2}p(p+1)$ and $D_{k}%
=(e_{1,\frac{1}{2}k(k+1)}$ $\vdots$ $D_{2,k}^{\prime})^{\prime},$ for
$D_{2,k}:(k^{2}-1)\times\frac{1}{2}k(k+1)$ with $e_{1,i}$ the first $i$
dimensional unity vector (i.e. the first column of $I_{i}).$ We partition
 $D_{p\perp}=(0$ $\vdots$ $D_{2,p\perp
}^{\prime})^{\prime},$ where $D_{2,p\perp}:(p^{2}-1)\times\frac{1}{2}p(p-1),$
$D_{2,p\perp}^{\prime}D_{2,p\perp}=I_{\frac{1}{2}p(p-1)},$ and $D_{k\perp}=(0$
$\vdots$ $D_{2,k\perp}^{\prime})^{\prime},$ where $D_{2,k\perp}:(k^{2}-1)\times
\frac{1}{2}k(k-1),$ $D_{2,k\perp}^{\prime}D_{2,k\perp}=I_{\frac{1}{2}k(k-1)},$
where the specifications of $D_{p\perp}$ and $D_{k\perp}$ result from those of
$D_{p}$ and $D_{k}.$

We next use the spectral decompositions of $A_{2}^{\prime}D_{2,p}^{\prime
}D_{2,p}A_{2}:(\frac{1}{2}p(p+1)-1)\times(\frac{1}{2}p(p+1)-1)$ and
$B_{2}^{\prime}D_{2,k}^{\prime}D_{2,k}B_{2}:(\frac{1}{2}k(k+1)-1)\times
(\frac{1}{2}k(k+1)-1):$%
\[%
\begin{array}
[c]{cc}%
A_{2}^{\prime}D_{2,p}^{\prime}D_{2,p}A_{2}= & L_{D_{2p}A_{2}}\Lambda
_{D_{2,p}A_{2}}^{2}L_{D_{2p}A_{2}}^{\prime}\\
B_{2}^{\prime}D_{2,k}^{\prime}D_{2,k}B_{2}= & L_{D_{2k}B_{2}}\Lambda
_{D_{2,k}B_{2}}^{2}L_{D_{2k}B_{2}}^{\prime},
\end{array}
\]
with $L_{D_{2p}A_{2}}$ and $L_{D_{2k}B_{2}}$ orthonormal $(\frac{1}%
{2}p(p+1)-1)\times(\frac{1}{2}p(p+1)-1)$ and $(\frac{1}{2}k(k+1)-1)\times
(\frac{1}{2}k(k+1)-1)$ dimensional matrices and $\Lambda_{D_{2,p}A_{2}}^{2}$
and $\Lambda_{D_{2,k}B_{2}}^{2}$ diagonal $(\frac{1}{2}p(p+1)-1)\times
(\frac{1}{2}p(p+1)-1)$ and $(\frac{1}{2}k(k+1)-1)\times(\frac{1}{2}k(k+1)-1)$
dimensional matrices with the squared singular values in non-increasing order
on the diagonal. We note that $A_{2}^{\prime}D_{2,p}^{\prime}D_{2,p}A_{2}$ is invertible.
This results since $A_{2}^{\prime}D_{p}^{\prime}D_{p}A_{2}=A_{2}^{\prime
}D_{2,p}^{\prime}D_{2,p}A_{2}+A_{2}^{\prime}e_{1,\frac{1}{2}p(p+1)}%
e_{1,\frac{1}{2}p(p+1)}^{\prime}A_{2}=I_{\frac{1}{2}p(p+1)}$ so $A_{2}%
^{\prime}D_{2,p}^{\prime}D_{2,p}A_{2}=I_{\frac{1}{2}p(p+1)}-A_{2}^{\prime
}e_{1,\frac{1}{2}p(p+1)}e_{1,\frac{1}{2}p(p+1)}^{\prime}A_{2}.$ Only when
$e_{1,\frac{1}{2}p(p+1)}^{\prime}A_{2}A_{2}^{\prime}e_{1,\frac{1}{2}p(p+1)}=1$
is this of lower rank since the specification then corresponds with a
projection matrix. This is, however, not possible given the specification of
$L=(D_{p}A$ $\vdots$ $D_{p\perp})$ which is orthonormal so $L^{\prime
}L=LL^{\prime}=I_{p^{2}}.$ The quadratic form (inner product) of the top row
of $L$ is thus equal to one. Given the specification of $D_{p},$ $D_{p\perp}$
has only zeros on the first row. Next, the $L_{11}$ element is unequal to zero
because $R_{11}$ is a positive definite covariance matrix. Since the $L_{11}$ element is unequal to zero, the length of the vector of the
remaining elements on the first row of $L$ can not be equal to one.\ This
implies that $e_{1,\frac{1}{2}p(p+1)}^{\prime}A_{2}A_{2}^{\prime}e_{1,\frac
{1}{2}p(p+1)}\neq1$ so $A_{2}^{\prime}D_{2,p}^{\prime}D_{2,p}A_{2}$ is invertible and $B_{2}^{\prime}D_{2,k}^{\prime}D_{2,k}B_{2}$ as well. A further consequence is that $L_{22}$ and $N_{22}$ are invertible and similarly $\hat{L}_{22}$ and $\hat{N}_{22}$.

The above spectral decompositions feature in the SVDs of
$L_{22},$ and $N_{22},$ which we can specify as:
\begin{align*}
L_{22}= &  \left(
\begin{array}
[c]{cc}%
D_{2,p}A_{2} & D_{2,p\perp}%
\end{array}
\right)  \\
= &  \left(
\begin{array}
[c]{cc}%
D_{2,p}A_{2}(L_{D_{2p}A_{2}}\Lambda_{D_{2,p}A_{2}}^{2}L_{D_{2p}A_{2}}^{\prime
})^{-\frac{1}{2}}L_{D_{2p}A_{2}}\Lambda_{D_{2,p}A_{2}}L_{D_{2,p}A_{2}}%
^{\prime} & D_{2,p\perp}%
\end{array}
\right)  \\
= &  \left(
\begin{array}
[c]{cc}%
D_{2,p}A_{2}(L_{D_{2p}A_{2}}\Lambda_{D_{2,p}A_{2}}^{2}L_{D_{2p}A_{2}}^{\prime
})^{-\frac{1}{2}}L_{D_{2p}A_{2}} & D_{2,p\perp}%
\end{array}
\right)  \left(
\begin{array}
[c]{cc}%
\Lambda_{D_{2,p}A_{2}} & 0\\
0 & I_{\frac{1}{2}p(p-1)}%
\end{array}
\right)  \\
&  \left(
\begin{array}
[c]{cc}%
L_{D_{2,p}A_{2}}^{\prime} & 0\\
0 & I_{\frac{1}{2}p(p-1)}%
\end{array}
\right)  ,\\
(L_{22}L_{22}^{\prime})^{\frac{1}{2}}= &  \left(
\begin{array}
[c]{cc}%
D_{2,p}A_{2}(L_{D_{2p}A_{2}}\Lambda_{D_{2,p}A_{2}}^{2}L_{D_{2p}A_{2}}^{\prime
})^{-\frac{1}{2}}L_{D_{2p}A_{2}} & D_{2,p\perp}%
\end{array}
\right)  \left(
\begin{array}
[c]{cc}%
\Lambda_{D_{2,p}A} & 0\\
0 & I_{\frac{1}{2}p(p-1)}%
\end{array}
\right)  \\
&  \left(
\begin{array}
[c]{cc}%
D_{2,p}A_{2}(L_{D_{2p}A_{2}}\Lambda_{D_{2,p}A_{2}}^{2}L_{D_{2p}A_{2}}^{\prime
})^{-\frac{1}{2}}L_{D_{2p}A_{2}} & D_{2,p\perp}%
\end{array}
\right)  ^{\prime},\\
L_{22}^{-1}\left(  L_{22}L_{22}^{\prime}\right)  ^{1/2}= &  \left(
\begin{array}
[c]{cc}%
L_{D_{2,p}A_{2}} & 0\\
0 & I_{\frac{1}{2}p(p-1)}%
\end{array}
\right)  \left(
\begin{array}
[c]{c}%
L_{D_{2p}A_{2}}^{\prime}(L_{D_{2p}A_{2}}\Lambda_{D_{2,p}A_{2}}^{2}%
P_{D_{2p}A_{2}}^{\prime})^{-\frac{1}{2}}A_{2}^{\prime}D_{2,p}^{\prime}\\
D_{2,p\perp}^{\prime}%
\end{array}
\right)  \\
= &  \left(
\begin{array}
[c]{c}%
(L_{D_{2p}A_{2}}\Lambda_{D_{2,p}A_{2}}^{2}P_{D_{2p}A_{2}}^{\prime})^{-\frac
{1}{2}}A_{2}^{\prime}D_{2,p}^{\prime}\\
D_{2,p\perp}^{\prime}%
\end{array}
\right)  =\left(
\begin{array}
[c]{c}%
(A_{2}^{\prime}D_{2,p}^{\prime}D_{2,p}A_{2})^{-\frac{1}{2}}A_{2}^{\prime
}D_{2,p}^{\prime}\\
D_{2,p\perp}^{\prime}%
\end{array}
\right)
\end{align*}
and%
\begin{equation}%
\begin{array}
[c]{rl}%
vec\left(  G_{1}\right)  _{\perp}=L_{2}L_{22}^{-1}\left(  L_{22}L_{22}^{\prime
}\right)  ^{1/2}= & \left(
\begin{array}
[c]{cc}%
D_{p}A_{2} & D_{p\perp}%
\end{array}
\right)  \left(
\begin{array}
[c]{c}%
(A_{2}^{\prime}D_{2,p}^{\prime}D_{2,p}A_{2})^{-\frac{1}{2}}A_{2}^{\prime
}D_{2,p}^{\prime}\\
D_{2,p\perp}^{\prime}%
\end{array}
\right)  ,\\
vec\left(  G_{2}\right)  _{\perp}=N_{2}N_{22}^{-1}\left(  N_{22}N_{22}^{\prime
}\right)  ^{1/2}= & \left(
\begin{array}
[c]{cc}%
D_{k}B_{2} & D_{k\perp}%
\end{array}
\right)  \left(
\begin{array}
[c]{c}%
(B_{2}^{\prime}D_{2,k}^{\prime}D_{2,k}B_{2})^{-\frac{1}{2}}B_{2}^{\prime
}D_{2,k}^{\prime}\\
D_{2,k\perp}^{\prime}%
\end{array}
\right)  ,
\end{array}
\label{eq: vecGperps}%
\end{equation}
where in the third line of the decomposition of $L_{22},$ we have the three
components that result from a SVD of $L_{22}$.

Then, under H$_{0}:$%
\begin{equation}%
\begin{array}
[c]{rl}%
\sqrt{n}\hat{\Lambda}= & vec(G_{1})_{\perp}^{\prime}D_{p}\Psi D_{k}^{\prime
}vec(G_{2})_{\perp}+o_{p}(1)\\
= & (L_{22}L_{22}^{\prime})^{1/2}L_{22}^{-1}L_{2}^{\prime}D_{p}\Psi
D_{k}^{\prime}N_{2}N_{22}^{-1}\left(  N_{22}N_{22}^{\prime}\right)
^{1/2}+o_{p}(1)\\
= & \left(
\begin{array}
[c]{c}%
(A_{2}^{\prime}D_{2,p}^{\prime}D_{2,p}A_{2})^{-\frac{1}{2}}A_{2}^{\prime
}D_{2,p}^{\prime}\\
D_{2,p\perp}^{\prime}%
\end{array}
\right)  ^{\prime}\left(
\begin{array}
[c]{cc}%
D_{p}A_{2} & D_{p\perp}%
\end{array}
\right)  ^{\prime}D_{p}\Psi D_{k}^{\prime}\\
& \left(
\begin{array}
[c]{cc}%
D_{k}B_{2} & D_{k\perp}%
\end{array}
\right)  \left(
\begin{array}
[c]{c}%
(B_{2}^{\prime}D_{2,k}^{\prime}D_{2,k}B_{2})^{-\frac{1}{2}}B_{2}^{\prime
}D_{2,k}^{\prime}\\
D_{2,k\perp}^{\prime}%
\end{array}
\right)  +o_{p}(1)\\
= & \left(
\begin{array}
[c]{c}%
(A_{2}^{\prime}D_{2,p}^{\prime}D_{2,p}A_{2})^{-\frac{1}{2}}A_{2}^{\prime
}D_{2,p}^{\prime}\\
D_{2,p\perp}^{\prime}%
\end{array}
\right)  ^{\prime}\left(
\begin{array}
[c]{cc}%
A_{2}^{\prime}D_{p}^{\prime}D_{p}\Psi D_{k}^{\prime}D_{k}B_{2} & 0\\
0 & 0
\end{array}
\right) \\
& \left(
\begin{array}
[c]{c}%
(B_{2}^{\prime}D_{2,k}^{\prime}D_{2,k}B_{2})^{-\frac{1}{2}}B_{2}^{\prime
}D_{2,k}^{\prime}\\
D_{2,k\perp}^{\prime}%
\end{array}
\right)  +o_{p}(1)\\
= & D_{2,p}A_{2}\bar{\Lambda}B_{2}^{\prime}D_{2,k}^{\prime}+o_{p}(1),
\end{array}
\label{eq: lim lambdahat}%
\end{equation}
\bigskip where%
\[%
\begin{array}
[c]{rl}%
\bar{\Lambda}:= & (A_{2}^{\prime}D_{2,p}^{\prime}D_{2,p}A_{2})^{-\frac{1}{2}%
}A_{2}^{\prime}D_{p}^{\prime}D_{p}\Psi D_{k}^{\prime}D_{k}B_{2}(B_{2}^{\prime
}D_{2,k}^{\prime}D_{2,k}B_{2})^{-\frac{1}{2}},
\end{array}
\]
which is a $\left(  \frac{1}{2}p(p+1)-1\right)  \times\left(  \frac{1}%
{2}k(k+1)-1\right)  $ normally distributed random matrix with mean zero. The
covariance matrix of $vec(\bar{\Lambda})$ equals%
\[%
\begin{array}
[c]{cl}%
V_{vec(\bar{\Lambda})}= & \left(  (B_{2}^{\prime}D_{2,k}^{\prime}D_{2,k}%
B_{2})^{-\frac{1}{2}}B_{2}^{\prime}D_{k}^{\prime}D_{k}\otimes(A_{2}^{\prime
}D_{2,p}^{\prime}D_{2,p}A_{2})^{-\frac{1}{2}}A_{2}^{\prime}D_{p}^{\prime}%
D_{p}\right)  V_{R^{\ast}}\\
& \times\left(  (B_{2}^{\prime}D_{2,k}^{\prime}D_{2,k}B_{2})^{-\frac{1}{2}%
}B_{2}^{\prime}D_{k}^{\prime}D_{k}\otimes(A_{2}^{\prime}D_{2,p}^{\prime
}D_{2,p}A_{2})^{-\frac{1}{2}}A_{2}^{\prime}D_{p}^{\prime}D_{p}\right)
^{\prime}.
\end{array}
\]
The above implies that the limiting distribution of $\sqrt{n}\hat{\Lambda}$ is
degenerate Normal because $D_{2,p}A_{2}$ and $D_{2,k}B_{2}$ are $(p^{2}%
-1)\times(\frac{1}{2}p(p+1)-1)$ and $(k^{2}-1)\times(\frac{1}{2}k(k+1)-1)$
dimensional matrices, respectively, and so the number of rows exceeds the
number of columns.\smallskip

We now apply a weak law of large numbers to the sample average $\hat{V}$
defined in (\ref{eq: kpst}). The matrix $\hat{V}$ contains summands of eighth
order products of $f_{i}$ and the weak law of large numbers holds by the
assumption that $E\left(  \left\Vert f_{i}\right\Vert ^{8}\right)  <\kappa.$
To derive the limit of the covariance matrix estimator in the KPST statistic,
the following derivations are important:
\begin{align}
&  \left(  vec(\hat{G}_{2})_{\perp}\otimes vec(\hat{G}_{1})_{\perp}\right)
^{\prime}\hat{V}\left(  vec(\hat{G}_{2})_{\perp}\otimes vec(\hat{G}%
_{1})_{\perp}\right)  \nonumber\label{eq: box}\\
\underset{p}{\rightarrow} &  \left(  vec(G_{2})_{\perp}\otimes vec(G_{1}%
)_{\perp}\right)  ^{\prime}(D_{k}\otimes D_{p})V_{R^{\ast}}(D_{k}\otimes
D_{p})^{\prime}\left(  vec(G_{2})_{\perp}\otimes vec(G_{1})_{\perp}\right)
\nonumber\\
= &  \left(  \left(
\begin{array}
[c]{cc}%
D_{k}B_{2} & D_{k\perp}%
\end{array}
\right)  \left(
\begin{array}
[c]{c}%
(B_{2}^{\prime}D_{2,k}^{\prime}D_{2,k}B_{2})^{-\frac{1}{2}}B_{2}^{\prime
}D_{2,k}^{\prime}\\
D_{2,k\perp}^{\prime}%
\end{array}
\right)  \otimes\right.  \nonumber\\
&  \left.  \left(
\begin{array}
[c]{cc}%
D_{p}A_{2} & D_{p\perp}%
\end{array}
\right)  \left(
\begin{array}
[c]{c}%
(A_{2}^{\prime}D_{2,p}^{\prime}D_{2,p}A_{2})^{-\frac{1}{2}}A_{2}^{\prime
}D_{2,p}^{\prime}\\
D_{2,p\perp}^{\prime}%
\end{array}
\right)  \right)  ^{\prime}(D_{k}\otimes D_{p})V_{R^{\ast}}(D_{k}\otimes
D_{p})^{\prime}\nonumber\\
&  \left(  \left(
\begin{array}
[c]{cc}%
D_{k}B_{2} & D_{k\perp}%
\end{array}
\right)  \left(
\begin{array}
[c]{c}%
(B_{2}^{\prime}D_{2,k}^{\prime}D_{2,k}B_{2})^{-\frac{1}{2}}B_{2}^{\prime
}D_{2,k}^{\prime}\\
D_{2,k\perp}^{\prime}%
\end{array}
\right)  \otimes\right.  \nonumber\\
&  \left.  \left(
\begin{array}
[c]{cc}%
D_{p}A_{2} & D_{p\perp}%
\end{array}
\right)  \left(
\begin{array}
[c]{c}%
(A_{2}^{\prime}D_{2,p}^{\prime}D_{2,p}A_{2})^{-\frac{1}{2}}A_{2}^{\prime
}D_{2,p}^{\prime}\\
D_{2,p\perp}^{\prime}%
\end{array}
\right)  \right)  \nonumber\\
= &  \left(  \left(
\begin{array}
[c]{cc}%
D_{2,k}B_{2} & D_{2,k\perp}%
\end{array}
\right)  \otimes\left(
\begin{array}
[c]{cc}%
D_{2,p}A_{2} & D_{2,p\perp}%
\end{array}
\right)  \right)  \\
&  \left(  \left(
\begin{array}
[c]{c}%
(B_{2}^{\prime}D_{2,k}^{\prime}D_{2,k}B_{2})^{-\frac{1}{2}}B_{2}^{\prime}%
D_{k}^{\prime}D_{k}\\
0
\end{array}
\right)  \otimes\left(
\begin{array}
[c]{c}%
(A_{2}^{\prime}D_{2,p}^{\prime}D_{2,p}A_{2})^{-\frac{1}{2}}A_{2}^{\prime}%
D_{p}^{\prime}D_{p}\\
0
\end{array}
\right)  \right)  V_{R^{\ast}}\nonumber\\
&  \left(  \left(
\begin{array}
[c]{c}%
(B_{2}^{\prime}D_{2,k}^{\prime}D_{2,k}B_{2})^{-\frac{1}{2}}B_{2}^{\prime}%
D_{k}^{\prime}D_{k}\\
0
\end{array}
\right)  \otimes\left(
\begin{array}
[c]{c}%
(A_{2}^{\prime}D_{2,p}^{\prime}D_{2,p}A_{2})^{-\frac{1}{2}}A_{2}^{\prime}%
D_{p}^{\prime}D_{p}\\
0
\end{array}
\right)  \right)  ^{\prime}\nonumber\\
&  \left(  \left(
\begin{array}
[c]{cc}%
D_{2,k}B_{2} & D_{2,k\perp}%
\end{array}
\right)  \otimes\left(
\begin{array}
[c]{cc}%
D_{2,p}A_{2} & D_{2,p\perp}%
\end{array}
\right)  \right)  ^{\prime}\nonumber\\
= &  \left(  \left(
\begin{array}
[c]{cc}%
D_{2,k}B_{2} & D_{2,k\perp}%
\end{array}
\right)  \otimes\left(
\begin{array}
[c]{cc}%
D_{2,p}A_{2} & D_{2,p\perp}%
\end{array}
\right)  \right)  \left(
\begin{array}
[c]{cccc}%
V_{vec(\bar{\Lambda})} & 0 & 0 & 0\\
0 & 0 & 0 & 0\\
0 & 0 & 0 & 0\\
0 & 0 & 0 & 0
\end{array}
\right)  \nonumber\\
&  \left(  \left(
\begin{array}
[c]{cc}%
D_{2,k}B_{2} & D_{2,k\perp}%
\end{array}
\right)  \otimes\left(
\begin{array}
[c]{cc}%
D_{2,p}A_{2} & D_{2,p\perp}%
\end{array}
\right)  \right)  ^{\prime},\nonumber
\end{align}
where the specifications of $vec(G_{1})_{\perp}$ and $vec(G_{2})_{\perp}$
result from (\ref{eq: vecGperps}). The convergence behavior of KPST is then
characterized by:
\begin{align*}
KPST= &  n\times\left[  vec(vec(G_{1})_{\perp}^{\prime}\mathcal{R}(\hat
{R})vec(G_{2})_{\perp})\right]  ^{\prime}\\
&  \left[  \left(  vec(G_{2})_{\perp}\otimes vec(G_{1})_{\perp}\right)
^{\prime}(D_{k}\otimes D_{p})V_{R^{\ast}}(D_{k}\otimes D_{p})^{\prime}\left(
vec(G_{2})_{\perp}\otimes vec(G_{1})_{\perp}\right)  \right]  ^{-}\\
&  \left[  vec(vec(G_{1})_{\perp}^{\prime}\mathcal{R}(\hat{R})vec(G_{2}%
)_{\perp})\right]  +o_{p}(1)\\
= &  vec(\bar{\Lambda})^{\prime}(D_{2,k}B_{2}\otimes D_{2,p}A_{2})^{\prime
}\left(  \left(
\begin{array}
[c]{cc}%
D_{2,k}B_{2}(B_{2}^{\prime}D_{2,k}^{\prime}D_{2,k}B_{2})^{-1} & D_{2,k\perp}%
\end{array}
\right)  \otimes\right.  \\
&  \left.  \left(
\begin{array}
[c]{cc}%
D_{2,p}A_{2}(A_{2}^{\prime}D_{2,p}^{\prime}D_{2,p}A_{2})^{-1} & D_{2,p\perp}%
\end{array}
\right)  \right)  \left(
\begin{array}
[c]{cccc}%
V_{vec(\bar{\Lambda})} & 0 & 0 & 0\\
0 & 0 & 0 & 0\\
0 & 0 & 0 & 0\\
0 & 0 & 0 & 0
\end{array}
\right)  ^{-}\\
&  \left(  \left(
\begin{array}
[c]{c}%
(B_{2}^{\prime}D_{2,k}^{\prime}D_{2,k}B_{2})^{-1}B_{2}^{\prime}D_{2,k}%
^{\prime}\\
D_{2,k\perp}^{\prime}%
\end{array}
\right)  \otimes\left(
\begin{array}
[c]{c}%
(A_{2}^{\prime}D_{2,p}^{\prime}D_{2,p}A_{2})^{-1}A_{2}^{\prime}D_{2,p}%
^{\prime}\\
D_{2,p\perp}^{\prime}%
\end{array}
\right)  \right)  \\
&  (D_{2,k}B_{2}\otimes D_{2,p}A_{2})vec(\bar{\Lambda})+o_{p}(1)\\
= &  vec(\bar{\Lambda})^{\prime}V_{vec(\bar{\Lambda})}^{-1}vec(\bar{\Lambda
})+o_{p}(1)\underset{d}{\rightarrow}\chi_{df}^{2},
\end{align*}
with $df=\left(  \frac{1}{2}p(p+1)-1\right)  \left(  \frac{1}{2}%
k(k+1)-1\right)  .$ The first equality substitutes $vec(\hat{G}_{1})_{\perp}$
and $vec(\hat{G}_{2})_{\perp}$ by their limits$.$ The second equality follows
from (\ref{eq: lambdahat}), vectorizing (\ref{eq: lim lambdahat}), and the
last line of (\ref{eq: box}). It also uses that the Moore-Penrose inverse of
the expression on the last line of (\ref{eq: box}) equals
\begin{align*}
&  \left[  \left(  \left(
\begin{array}
[c]{cc}%
D_{2,k}B_{2} & D_{2,k\perp}%
\end{array}
\right)  \otimes\left(
\begin{array}
[c]{cc}%
D_{2,p}A_{2} & D_{2,p\perp}%
\end{array}
\right)  \right)  \left(
\begin{array}
[c]{cccc}%
V_{vec(\bar{\Lambda})} & 0 & 0 & 0\\
0 & 0 & 0 & 0\\
0 & 0 & 0 & 0\\
0 & 0 & 0 & 0
\end{array}
\right)  \right.  \\
&  \left.  \left(  \left(
\begin{array}
[c]{cc}%
D_{2,k}B_{2} & D_{2,k\perp}%
\end{array}
\right)  \otimes\left(
\begin{array}
[c]{cc}%
D_{2,p}A_{2} & D_{2,p\perp}%
\end{array}
\right)  \right)  ^{\prime}\right]  ^{-}\\
&  =\left(  \left(
\begin{array}
[c]{cc}%
D_{2,k}B_{2}(B_{2}^{\prime}D_{2,k}^{\prime}D_{2,k}B_{2})^{-1} & D_{2,k\perp}%
\end{array}
\right)  \otimes\right.  \\
&  \left.  \left(
\begin{array}
[c]{cc}%
D_{2,p}A_{2}(A_{2}^{\prime}D_{2,p}^{\prime}D_{2,p}A_{2})^{-1} & D_{2,p\perp}%
\end{array}
\right)  \right)  \left(
\begin{array}
[c]{cccc}%
V_{vec(\bar{\Lambda})} & 0 & 0 & 0\\
0 & 0 & 0 & 0\\
0 & 0 & 0 & 0\\
0 & 0 & 0 & 0
\end{array}
\right)  ^{-}\\
&  \left(  \left(
\begin{array}
[c]{c}%
(B_{2}^{\prime}D_{2,k}^{\prime}D_{2,k}B_{2})^{-1}B_{2}^{\prime}D_{2,k}%
^{\prime}\\
D_{2,k\perp}^{\prime}%
\end{array}
\right)  \otimes\left(
\begin{array}
[c]{c}%
(A_{2}^{\prime}D_{2,p}^{\prime}D_{2,p}A_{2})^{-1}A_{2}^{\prime}D_{2,p}%
^{\prime}\\
D_{2,p\perp}^{\prime}%
\end{array}
\right)  \right)
\end{align*}
which follows because
\[
\left(
\begin{array}
[c]{cc}%
D_{2,k}B_{2}(B_{2}^{\prime}D_{2,k}^{\prime}D_{2,k}B_{2})^{-1} & D_{2,k\perp}%
\end{array}
\right)  ^{\prime}\left(
\begin{array}
[c]{cc}%
D_{2,k}B_{2} & D_{2,k\perp}%
\end{array}
\right)  =\left(
\begin{array}
[c]{cc}%
I_{\frac{1}{2}k(k+1)-1} & 0\\
0 & I_{\frac{1}{2}k(k-1)}%
\end{array}
\right)  =I_{k^{2}-1}%
\]
as $D_{2,k\perp}=I_{\frac{1}{2}k(k-1)}$ and the same argument can be applied
to the other component. The third equality then follows from%
\[%
\begin{array}
[c]{rl}%
\left(
\begin{array}
[c]{cc}%
D_{2,p}A_{2} & D_{2,p\perp}%
\end{array}
\right)  ^{-1}= & \left(
\begin{array}
[c]{cc}%
(A_{2}^{\prime}D_{2,p}^{\prime}D_{2,p}A_{2})^{-1} & 0\\
0 & I_{\frac{1}{2}p(p-1)}%
\end{array}
\right)  \left(
\begin{array}
[c]{cc}%
D_{2,p}A_{2} & D_{2,p\perp}%
\end{array}
\right)  ^{\prime}\\
\left(
\begin{array}
[c]{cc}%
D_{2,k}B_{2} & D_{2,k\perp}%
\end{array}
\right)  ^{-1}= & \left(
\begin{array}
[c]{cc}%
(B_{2}^{\prime}D_{2,k}^{\prime}D_{2,k}B_{2})^{-1} & 0\\
0 & I_{\frac{1}{2}k(k-1)}%
\end{array}
\right)  \left(
\begin{array}
[c]{cc}%
D_{2,k}B_{2} & D_{2,k\perp}%
\end{array}
\right)  ^{\prime}.
\end{array}
\]

\textbf{b. }We show that if $\mathcal{R}(\hat{R})=D_{p}\hat{R}^{\ast}%
D_{k}^{\prime}$ is replaced with
\[%
\begin{array}
[c]{rl}%
\bar{R}:= & D_{p}(D_{p}^{\prime}D_{p})^{-\frac{1}{2}}\hat{R}^{\ast}%
(D_{k}^{\prime}D_{k})^{-\frac{1}{2}}D_{k}^{\prime}%
\end{array}
\]
in the definition of KPST, one obtains KPST$^{\ast}$ in (\ref{eq: kpst*}). To
show this, we use SVDs of $\bar{R}=\bar{L}\bar{\Sigma}\bar{N}^{\prime}$ and
$\hat{R}^{\ast}=\hat{L}^{\ast}\hat{\Sigma}^{\ast}\hat{N}^{\ast\prime}$ which
are related through:%
\[%
\begin{array}
[c]{rl}%
\bar{L}= & \left(  D_{p}(D_{p}^{\prime}D_{p})^{-\frac{1}{2}}\hat{L}^{\ast
}\text{ }\vdots\text{ }D_{p\perp}\right) \\
\bar{\Sigma}= & \left(
\begin{array}
[c]{cc}%
\hat{\Sigma}^{\ast} & 0\\
0 & 0
\end{array}
\right) \\
\bar{N}= & \left(  D_{k}(D_{k}^{\prime}D_{k})^{-\frac{1}{2}}\hat{N}^{\ast
}\text{ }\vdots\text{ }D_{k\perp}\right)  .
\end{array}
\]
To show that KPST using $\bar{R}$, indicated by KPST$_{\bar{R}},$ equals
KPST$^{\ast}$, we analyze KPST$_{\bar{R}}$:
\begin{align*}
KPST_{\bar{R}}=  &  n\times\left[  vec(vec(\bar{G}_{1})_{\perp}^{\prime}%
\bar{R}vec(\bar{G}_{2})_{\perp})\right]  ^{\prime}\\
&  \left[  \left(  vec(\bar{G}_{2})_{\perp}^{\prime}\otimes vec(\bar{G}%
_{1})_{\perp}^{\prime}\right)  \hat{V}_{\bar{R}}\left(  vec(\bar{G}%
_{2})_{\perp}\otimes vec(\bar{G}_{1})_{\perp}\right)  \right]  ^{-}\\
&  \left[  vec(vec(\bar{G}_{1})_{\perp}^{\prime}\bar{R}vec(\bar{G}_{2}%
)_{\perp})\right]  ,\qquad\hat{V}_{\bar{R}}:=\widehat{\text{cov}}\left(
vec\left(  \bar{R}\right)  \right) \\
=  &  vec\left(  \left(
\begin{array}
[c]{cc}%
\hat{\Sigma}_{2}^{\ast} & 0\\
0 & 0
\end{array}
\right)  \right)  ^{\prime}\left(  \left(  \bar{N}_{22}\bar{N}_{22}^{\prime
}\right)  ^{1/2}\bar{N}_{22}^{\prime-}\otimes(\bar{L}_{22}\bar{L}_{22}%
^{\prime})^{1/2}\bar{L}_{22}^{-\prime}\right)  ^{\prime}\\
&  \left\{  \left(  \left(  \bar{N}_{22}\bar{N}_{22}^{\prime}\right)
^{1/2}\bar{N}_{22}^{\prime-}\otimes(\bar{L}_{22}\bar{L}_{22}^{\prime}%
)^{1/2}\bar{L}_{22}^{-\prime}\right)  \rule{0pt}{14pt}\right. \\
&  \left[  \left(  \left(  D_{k}(D_{k}^{\prime}D_{k})^{-\frac{1}{2}}\hat
{N}_{2}^{\ast}\text{ }\vdots\text{ }D_{k\perp}\right)  \otimes\left(
D_{p}(D_{p}^{\prime}D_{p})^{-\frac{1}{2}}\hat{L}_{2}^{\ast}\text{ }%
\vdots\text{ }D_{p\perp}\right)  \right)  ^{\prime}\right. \\
&  \left(  D_{k}(D_{k}^{\prime}D_{k})^{-\frac{1}{2}}\otimes D_{p}%
(D_{p}^{\prime}D_{p})^{-\frac{1}{2}}\right)  \hat{V}_{\hat{R}^{\ast}}\left(
D_{k}(D_{k}^{\prime}D_{k})^{-\frac{1}{2}}\otimes D_{p}(D_{p}^{\prime}%
D_{p})^{-\frac{1}{2}}\right)  ^{\prime}\\
&  \left(  \left(  D_{k}(D_{k}^{\prime}D_{k})^{-\frac{1}{2}}\hat{N}_{2}^{\ast
}\text{ }\vdots\text{ }D_{k\perp}\right)  \otimes\left(  D_{p}(D_{p}^{\prime
}D_{p})^{-\frac{1}{2}}\hat{L}_{2}^{\ast}\text{ }\vdots\text{ }D_{p\perp
}\right)  \right) \\
&  \left.  \left(  \bar{N}_{22}^{-1}\left(  \bar{N}_{22}\bar{N}_{22}^{\prime
}\right)  ^{\frac{1}{2}}\otimes\bar{L}_{22}^{-1}(\bar{L}_{22}\bar{L}%
_{22}^{\prime})^{1/2}\right)  \right\}  ^{-1}\\
&  \left(  \left(  \bar{N}_{22}\bar{N}_{22}^{\prime}\right)  ^{\frac{1}%
{2}\prime}\bar{N}_{22}^{-\prime}\otimes(\bar{L}_{22}\bar{L}_{22}^{\prime
})^{1/2}\bar{L}_{22}^{-\prime}\right)  vec\left(  \left(
\begin{array}
[c]{cc}%
\hat{\Sigma}_{2}^{\ast} & 0\\
0 & 0
\end{array}
\right)  \right) \\
=  &  vec\left(  \left(
\begin{array}
[c]{cc}%
\hat{\Sigma}_{2}^{\ast} & 0\\
0 & 0
\end{array}
\right)  \right)  ^{\prime}\left(  \left(  \bar{N}_{22}\bar{N}_{22}^{\prime
}\right)  ^{1/2}\bar{N}_{22}^{\prime-}\otimes(\bar{L}_{22}\bar{L}_{22}%
^{\prime})^{1/2}\bar{L}_{22}^{-\prime}\right)  ^{\prime}\\
&  \left[  \left(  \left(  \bar{N}_{22}\bar{N}_{22}^{\prime}\right)
^{1/2}\bar{N}_{22}^{\prime-}\otimes(\bar{L}_{22}\bar{L}_{22}^{\prime}%
)^{1/2}\bar{L}_{22}^{-\prime}\right)  \right. \\
&  \left(  \left(
\begin{array}
[c]{cc}%
\hat{N}_{2}^{\ast} & 0
\end{array}
\right)  \otimes\left(
\begin{array}
[c]{cc}%
\hat{L}_{2}^{\ast} & 0
\end{array}
\right)  \right)  ^{\prime}\hat{V}_{\hat{R}^{\ast}}\left(  \left(
\begin{array}
[c]{cc}%
\hat{N}_{2}^{\ast} & 0
\end{array}
\right)  \otimes\left(
\begin{array}
[c]{cc}%
\hat{L}_{2}^{\ast} & 0
\end{array}
\right)  \right) \\
&  \left.  \left(  \bar{N}_{22}^{-1}\left(  \bar{N}_{22}\bar{N}_{22}^{\prime
}\right)  ^{\frac{1}{2}}\otimes\bar{L}_{22}^{-1}(\bar{L}_{22}\bar{L}%
_{22}^{\prime})^{1/2}\right)  \right]  ^{-1}\\
&  \left(  \left(  \bar{N}_{22}\bar{N}_{22}^{\prime}\right)  ^{\frac{1}%
{2}\prime}\bar{N}_{22}^{-\prime}\otimes(\bar{L}_{22}\bar{L}_{22}^{\prime
})^{1/2}\bar{L}_{22}^{-\prime}\right)  vec\left(  \left(
\begin{array}
[c]{cc}%
\hat{\Sigma}_{2}^{\ast} & 0\\
0 & 0
\end{array}
\right)  \right) \\
=  &  vec\left(  \left(
\begin{array}
[c]{cc}%
\hat{\Sigma}_{2}^{\ast} & 0\\
0 & 0
\end{array}
\right)  \right)  ^{\prime}\left[  \left(  \left(
\begin{array}
[c]{cc}%
\hat{N}_{2}^{\ast} & 0
\end{array}
\right)  \otimes\left(
\begin{array}
[c]{cc}%
\hat{L}_{2}^{\ast} & 0
\end{array}
\right)  \right)  ^{\prime}\right. \\
&  \left.  \hat{V}_{\hat{R}^{\ast}}\left(  \left(
\begin{array}
[c]{cc}%
\hat{N}_{2}^{\ast} & 0
\end{array}
\right)  \otimes\left(
\begin{array}
[c]{cc}%
\hat{L}_{2}^{\ast} & 0
\end{array}
\right)  \right)  \right]  ^{-}vec\left(  \left(
\begin{array}
[c]{cc}%
\hat{\Sigma}_{2}^{\ast} & 0\\
0 & 0
\end{array}
\right)  \right)
\end{align*}%
\begin{equation}%
\begin{array}
[c]{rl}%
= & n\times vec(\hat{\Sigma}_{2}^{\ast})^{\prime}\left[  \left(  \hat{N}%
_{2}^{\ast\prime}\otimes\hat{L}_{2}^{\ast\prime}\right)  \hat{V}_{\hat
{R}^{\ast}}\left(  \hat{N}_{2^{\ast}}^{\ast}\otimes\hat{L}_{2}^{\ast}\right)
\right]  ^{-1}vec(\hat{\Sigma}_{2,R^{\ast}})\\
= & KPST^{\ast},
\end{array}
\label{eq: kpst* alt}%
\end{equation}
which is the KPST expression using $\hat{R}^{\ast}$ so it differs from KPST.
Here we used that%
\[%
\begin{array}
[c]{rl}%
\hat{V}_{\bar{R}}= & \left(  D_{k}(D_{k}^{\prime}D_{k})^{-\frac{1}{2}}\otimes
D_{p}(D_{p}^{\prime}D_{p})^{-\frac{1}{2}}\right)  \hat{V}_{\hat{R}^{\ast}%
}\left(  D_{k}(D_{k}^{\prime}D_{k})^{-\frac{1}{2}}\otimes D_{p}(D_{p}^{\prime
}D_{p})^{-\frac{1}{2}}\right)  ^{\prime},\\
\bar{L}_{2}= & \left(  D_{p}(D_{p}^{\prime}D_{p})^{-\frac{1}{2}}\hat
{L}_{2,R^{\ast}}\text{ }\vdots\text{ }D_{p\perp}\right)  ,\\
\bar{N}_{2}= & \left(  D_{k}(D_{k}^{\prime}D_{k})^{-\frac{1}{2}}\hat
{N}_{2,R^{\ast}}\text{ }\vdots\text{ }D_{k\perp}\right)  ,
\end{array}
\]
and that an expression like (\ref{eq: simple KPST}) can be similarly shown to
apply to KPST$^{\ast}.$

Because $\hat{R}^{\ast}$ has the non-degenerate limiting distribution
(\ref{eq: CLT R_n}), the limiting distribution of KPST using $\hat{R}^{\ast}$
directly results from \cite[Corollary 1]{kleibergen2006grr} and is also
$\chi_{df}^{2}$. \medskip

\textbf{c. } To show (non-) invariance to orthonormal transformations of $\hat{V}_{i}$ and
$Z_{i},$ we consider a $p\times p$ dimensional orthonormal matrix $Q$ using
which we rotate $\hat{V}_{i}$ to become $Q\hat{V}_{i}$ so%
\[%
\begin{array}
[c]{rl}%
vec(Q\hat{V}_{i}\hat{V}_{i}^{\prime}Q^{\prime})= & (Q\otimes Q)vec(\hat{V}%
_{i}\hat{V}_{i}^{\prime})\\
= & (Q\otimes Q)D_{p}vech(\hat{V}_{i}\hat{V}_{i}^{\prime})\\
vech(Q\hat{V}_{i}\hat{V}_{i}^{\prime}Q^{\prime})= & (D_{p}^{\prime}D_{p}%
)^{-1}D_{p}^{\prime}(Q\otimes Q)D_{p}vech(\hat{V}_{i}\hat{V}_{i}^{\prime}),
\end{array}
\]
which implies that if we also rotate $Z_{i}$ by the $k\times k$ orthonormal
matrix $H:$
\[%
\begin{array}
[c]{l}%
\frac{1}{n}\sum_{i=1}^{n}vec(Q\hat{V}_{i}\hat{V}_{i}^{\prime}Q^{\prime
})vec(HZ_{i}Z_{i}^{\prime}H^{\prime})^{\prime}=\\
(Q\otimes Q)D_{p}\left[  \frac{1}{n}\sum_{i=1}^{n}vec(\hat{V}_{i}\hat{V}%
_{i}^{\prime})vec(Z_{i}Z_{i}^{\prime})^{\prime}\right]  D_{k}^{\prime
}(H\otimes H)^{\prime},
\end{array}
\]
and
\[%
\begin{array}
[c]{l}%
\frac{1}{n}\sum_{i=1}^{n}vech(Q\hat{V}_{i}\hat{V}_{i}^{\prime}Q^{\prime
})vech(HZ_{i}Z_{i}^{\prime}H^{\prime})^{\prime}=\\
(D_{p}^{\prime}D_{p})^{-1}D_{p}^{\prime}(Q\otimes Q)D_{p}\left[  \frac{1}%
{n}\sum_{i=1}^{n}vech(\hat{V}_{i}\hat{V}_{i}^{\prime})vech(Z_{i}Z_{i}^{\prime
})^{\prime}\right] \\
\times D_{k}^{\prime}(H\otimes H)^{\prime}D_{k}(D_{k}^{\prime}D_{k})^{-1}.
\end{array}
\]
Hence, because $Q$ and $H$ are orthonormal, this implies that the different
components of the SVD decomposition of $\mathcal{R}(\hat{R})$ in
(\ref{eq: SVD general}) with the transformed $\hat{R}$ become $(Q\otimes
Q)\hat{L},\hat{\Sigma}$ and $(H\otimes H)\hat{N}.$ Because these rotations
also transform the covariance matrix $\hat{V}$ to $(Q\otimes Q)\hat
{V}(H\otimes H)^{\prime},$ it immediately follows from the expression in KPST
in (\ref{eq: simple KPST}) that KPST is invariant to rotations of $V_{i}$ and
$Z_{i}.$

Let $\hat{R}_{T}^{\ast}$ be the $\hat{R}^{\ast}$ in (\ref{eq: R star})
obtained using the transformed data $Q\hat{V}_{i}$ and $HZ_{i},$ whose SVD is
\[%
\begin{array}
[c]{rl}%
\hat{R}_{T}^{\ast}= & \frac{1}{n}\sum_{i=1}^{n}vech(Q\hat{V}_{i}\hat{V}%
_{i}^{\prime}Q^{\prime})vech(HZ_{i}Z_{i}^{\prime}H^{\prime})^{\prime}\\
= & (D_{p}^{\prime}D_{p})^{-1}D_{p}^{\prime}(Q\otimes Q)D_{p}\left[  \frac
{1}{n}\sum_{i=1}^{n}vech(\hat{V}_{i}\hat{V}_{i}^{\prime})vech(Z_{i}%
Z_{i}^{\prime})^{\prime}\right] \\
& D_{k}^{\prime}(H\otimes H)^{\prime}D_{k}(D_{k}^{\prime}D_{k})^{-1}\\
= & (D_{p}^{\prime}D_{p})^{-1}D_{p}^{\prime}(Q\otimes Q)D_{p}\hat{L}^{\ast
}\hat{\Sigma}^{\ast}\hat{N}^{\ast\prime}D_{k}^{\prime}(H\otimes H)^{\prime
}D_{k}(D_{k}^{\prime}D_{k})^{-1}\\
= & \hat{L}_{T}^{\ast}\hat{\Sigma}_{T}^{\ast}\hat{N}_{T}^{\ast\prime},
\end{array}
\]
with $\hat{L}_{T}^{\ast}$ and $\hat{N}_{T}^{\ast}$ orthonormal $(\frac{1}%
{2}p(p+1)-1)\times(\frac{1}{2}p(p+1)-1)$ and $(\frac{1}{2}k(k+1)-1)\times
(\frac{1}{2}k(k+1)-1)$ dimensional matrices and $\hat{\Sigma}_{T}^{\ast}$ a
diagonal $(\frac{1}{2}p(p+1)-1)\times(\frac{1}{2}p(p+1)-1)$ dimensional matrix
with the singular values in non-increasing order on the main diagonal. Because
$(D_{p}^{\prime}D_{p})^{-1}D_{p}^{\prime}(Q\otimes Q)D_{p}$ and $(D_{k}%
^{\prime}D_{k})^{-1}D_{k}^{\prime}(H\otimes H)^{\prime}D_{k}$ are not
orthonormal, it follows that%
\[
\hat{L}_{T}^{\ast}\neq(D_{p}^{\prime}D_{p})^{-1}D_{p}^{\prime}(Q\otimes
Q)D_{p}\hat{L}^{\ast}\text{ and }\hat{N}_{T}^{\ast}\neq(D_{k}^{\prime}%
D_{k})^{-1}D_{k}^{\prime}(H\otimes H)^{\prime}D_{k}\hat{N}^{\ast}%
\]
and therefore also
\[
\hat{\Sigma}_{T}^{\ast}\neq\hat{\Sigma}^{\ast}.
\]
When substituting these expressions into the expression of KPST$^{\ast}$ in
(\ref{eq: kpst* alt}), it then follows that KPST$^{\ast}$ is not invariant to
orthonormal transformations of the data.

\textbf{d. }Under H$_{0}$ and joint limit sequences of $k,$ $p$ and
$n,$ we have to consider all components of $vec(\hat{\Lambda})$ in
(\ref{eq: lambda}) and its covariance matrix estimator.%
\begin{align*}
vec(\hat{\Lambda})=  &  \left(  vec(\hat{G}_{2})_{\perp}\otimes vec(\hat
{G}_{1})_{\perp}\right)  ^{\prime}vec(\mathcal{R}(\hat{R}))\\
=  &  \left(  \left[  vec(G_{2})_{\perp}+vec(\hat{G}_{2})_{\perp}%
-vec(G_{2})_{\perp}\right]  \otimes\left[  vec(G_{1})_{\perp}+vec(\hat{G}%
_{1})_{\perp}-vec(G_{1})_{\perp}\right]  \right)  ^{\prime}\\
&  vec\left(  \mathcal{R}(R)+\mathcal{R}(\hat{R})-\mathcal{R}(R)\right) \\
=  &  \left(  vec(G_{2})_{\perp}\otimes vec(G_{1})_{\perp}\right)  ^{\prime
}vec(\mathcal{R}(R))+\left(  \left[  vec(\hat{G}_{2})_{\perp}-vec(G_{2}%
)_{\perp}\right]  \otimes I_{p^{2}-1}\right)  ^{\prime}\\
&  vec(vec(G_{1})_{\perp}^{\prime}\mathcal{R}(R))+(vec(G_{2})_{\perp}\otimes
vec(G_{1})_{\perp})^{\prime}vec\left(  \mathcal{R}(\hat{R})-\mathcal{R}%
(R)\right)  +\\
&  \left(  \left[  vec(\hat{G}_{2})_{\perp}-vec(G_{2})_{\perp}\right]  \otimes
vec(G_{1})_{\perp}\right)  ^{\prime}vec\left(  \mathcal{R}(\hat{R}%
)-\mathcal{R}(R)\right)  +\\
&  \left(  I_{k^{2}-1}\otimes\left[  vec(\hat{G}_{1})_{\perp}-vec(G_{1}%
)_{\perp}\right]  \right)  ^{\prime}vec(\mathcal{R}(R)vec(G_{2})_{\perp})+\\
&  \left(  \left[  vec(\hat{G}_{2})_{\perp}-vec(G_{2})_{\perp}\right]
\otimes\left[  vec(\hat{G}_{1})_{\perp}-vec(G_{1})_{\perp}\right]  \right)
^{\prime}vec(\mathcal{R}(R))+\\
&  \left(  \left[  vec(\hat{G}_{2})_{\perp}-vec(G_{2})_{\perp}\right]
\otimes\left[  vec(\hat{G}_{1})_{\perp}-vec(G_{1})_{\perp}\right]  \right)
^{\prime}vec\left(  \mathcal{R}(\hat{R})-\mathcal{R}(R)\right)  +\\
&  \left(  vec(G_{2})_{\perp}\otimes\left[  vec(\hat{G}_{1})_{\perp}%
-vec(G_{1})_{\perp}\right]  \right)  ^{\prime}vec\left(  \mathcal{R}(\hat
{R})-\mathcal{R}(R)\right) \\
=  &  \left(  vec(G_{2})_{\perp}\otimes vec(G_{1})_{\perp}\right)  ^{\prime
}vec\left(  \mathcal{R}(\hat{R})-\mathcal{R}(R)\right)  +\\
&  \left(  \left[  vec(\hat{G}_{2})_{\perp}-vec(G_{2})_{\perp}\right]  \otimes
vec(G_{1})_{\perp}\right)  ^{\prime}vec\left(  \mathcal{R}(\hat{R}%
)-\mathcal{R}(R)\right)  +\\
&  \left(  \left[  vec(\hat{G}_{2})_{\perp}-vec(G_{2})_{\perp}\right]
\otimes\left[  vec(\hat{G}_{1})_{\perp}-vec(G_{1})_{\perp}\right]  \right)
^{\prime}vec(\mathcal{R}(R))+\\
&  \left(  \left[  vec(\hat{G}_{2})_{\perp}-vec(G_{2})_{\perp}\right]
\otimes\left[  vec(\hat{G}_{1})_{\perp}-vec(G_{1})_{\perp}\right]  \right)
^{\prime}vec\left(  \mathcal{R}(\hat{R})-\mathcal{R}(R)\right)  +\\
&  (vec(G_{2})_{\perp}\otimes\left[  vec(\hat{G}_{1})_{\perp}-vec(G_{1}%
)_{\perp}\right]  )^{\prime}vec\left[  \mathcal{R}(\hat{R})-\mathcal{R}%
(R)\right] \\
=  &  a+b+c
\end{align*}
for
\begin{align*}
a:=  &  (vec(G_{2})_{\perp}\otimes vec(G_{1})_{\perp})^{\prime}vec\left(
\mathcal{R}(\hat{R})-\mathcal{R}(R)\right)  ,\\
b:=  &  \left(  \left[  vec(\hat{G}_{2})_{\perp}-vec(G_{2})_{\perp}\right]
\otimes vec(G_{1})_{\perp}\right)  ^{\prime}vec\left(  \mathcal{R}(\hat
{R})-\mathcal{R}(R)\right)  +\\
&  \left(  \left[  vec(\hat{G}_{2})_{\perp}-vec(G_{2})_{\perp}\right]
\otimes\left[  vec(\hat{G}_{1})_{\perp}-vec(G_{1})_{\perp}\right]  \right)
^{\prime}vec(\mathcal{R}(R))+\\
&  \left(  vec(G_{2})_{\perp}\otimes\left[  vec(\hat{G}_{1})_{\perp}%
-vec(G_{1})_{\perp}\right]  \right)  ^{\prime}vec\left(  \mathcal{R}(\hat
{R})-\mathcal{R}(R)\right)  ,\\
c:=  &  \left(  \left[  vec(\hat{G}_{2})_{\perp}-vec(G_{2})_{\perp}\right]
\otimes\left[  vec(\hat{G}_{1})_{\perp}-vec(G_{1})_{\perp}\right]  \right)
^{\prime}vec\left(  \mathcal{R}(\hat{R})-\mathcal{R}(R)\right)  .
\end{align*}

In the derivation above, we use that under H$_{0}$, $\mathcal{R}%
(R)=vec(G_{1})vec(G_{2})^{\prime},$ see (\ref{eq: red rank cov}). Therefore
$vec(G_{1})_{\perp}^{\prime}\mathcal{R}(R)=0,$ $\mathcal{R}(R)vec(G_{2}%
)_{\perp}=0$. The limit behavior of KPST results from the limit behavior of
$a.$ We specify both $vec(G_{1})_{\perp}$ and $vec(G_{2})_{\perp},$ whose
dimensions increase as $k$ and $p$ get larger, as orthonormal matrices,
$vec(G_{1})_{\perp}^{\prime}vec(G_{1})_{\perp}\equiv I_{\frac{1}{2}k(k-1)}$
and $vec(G_{2})_{\perp}^{\prime}vec(G_{2})_{\perp}\equiv I_{\frac{1}{2}%
p(p-1)}.$ Hence the length of each column of $vec(G_{1})_{\perp}$ and
$vec(G_{2})_{\perp}$ equals one and does not change when $k$ and/or $p$ increase.

From (\ref{eq: CLT R_n}), it follows that $\mathcal{R}(\hat{R})-\mathcal{R}%
(R)=O_{p}(n^{-\frac{1}{2}})$, and so the same holds for the convergences rates
of $vec(\hat{G}_{1})_{\perp}-vec(G_{1})_{\perp}$ and $vec(\hat{G}_{2})_{\perp
}-vec(G_{2})_{\perp},$ see \cite{kleibergen2006grr}. Because $vec(\hat{G}%
_{1})_{\perp}$ and $vec(\hat{G}_{2})_{\perp}$ are solved from $\mathcal{R}%
(\hat{R}),$ it follows that $\mathcal{R}(\hat{R})-\mathcal{R}(R),$
$vec(\hat{G}_{1})_{\perp}-vec(G_{1})_{\perp}$ and $vec(\hat{G}_{2})_{\perp
}-vec(G_{2})_{\perp}$ are all jointly dependent. In a limiting sequence where
the dimensions $p$ and $k$ jointly increase with the sample size $n,$ we then
have the following convergence rates:%
\[%
\begin{array}
[c]{rrl}%
1. & (vec(G_{2})_{\perp}\otimes vec(G_{1})_{\perp})^{\prime}vec\left(
\mathcal{R}(\hat{R})-\mathcal{R}(R)\right)  & =O_{p}\left(  n^{-\frac{1}{2}%
}\right) \\
2. & \left(  \left[  vec(\hat{G}_{2})_{\perp}-vec(G_{2})_{\perp}\right]
\otimes vec(G_{1})_{\perp}\right)  ^{\prime}vec\left(  \mathcal{R}(\hat
{R})-\mathcal{R}(R)\right)  & =O_{p}%
%TCIMACRO{\QOVERD{(}{)}{k^{2}}{n}}%
%BeginExpansion
\genfrac{(}{)}{}{}{k^{2}}{n}%
%EndExpansion
\\
3. & \left(  vec(G_{2})_{\perp}\otimes\left[  vec(\hat{G}_{1})_{\perp
}-vec(G_{1})_{\perp}\right]  \right)  ^{\prime}vec\left(  \mathcal{R}(\hat
{R})-\mathcal{R}(R)\right)  & =O_{p}%
%TCIMACRO{\QOVERD{(}{)}{p^{2}}{n}}%
%BeginExpansion
\genfrac{(}{)}{}{}{p^{2}}{n}%
%EndExpansion
\\
4. & \left(  \left[  vec(\hat{G}_{2})_{\perp}-vec(G_{2})_{\perp}\right]
\otimes\left[  vec(\hat{G}_{1})_{\perp}-vec(G_{1})_{\perp}\right]  \right)
^{\prime}vec(\mathcal{R}(R)) & =O_{p}%
%TCIMACRO{\QOVERD{(}{)}{(pk)^{2}}{n}}%
%BeginExpansion
\genfrac{(}{)}{}{}{(pk)^{2}}{n}%
%EndExpansion
\\
5. & \left(  \left[  vec(\hat{G}_{2})_{\perp}-vec(G_{2})_{\perp}\right]
\otimes\left[  vec(\hat{G}_{1})_{\perp}-vec(G_{1})_{\perp}\right]  \right)
^{\prime} & \\
& vec\left(  \mathcal{R}(\hat{R})-\mathcal{R}(R)\right)  & =O_{p}%
%TCIMACRO{\QOVERD{(}{)}{p^{2}k^{2}}{n\sqrt{n}}}%
%BeginExpansion
\genfrac{(}{)}{}{}{p^{2}k^{2}}{n\sqrt{n}}%
%EndExpansion
.
\end{array}
\]
The individual elements of each of the above five components result from
multiplying the first KPS matrix with the second vectorized matrix. This
multiplication implies that the individual elements equal weighted summations
where the number of elements where we sum over increases with the sequence of
$k$ and $p.$ This affects the convergence rate of the individual elements. The
convergence rate of the individual elements is then a function of the sum of
the involved weights and the convergence rates of the multiplied components.
Along these lines, we next establish the convergence rate for, say, the $q$-th
element of each of the five components in the above expression:%
\begin{align*}
1.  &  \left[  \left(  vec(G_{2})_{\perp}\otimes vec(G_{1})_{\perp}\right)
^{\prime}vec\left(  \mathcal{R}(\hat{R})-\mathcal{R}(R)\right)  \right]
_{q}\\
&
\begin{array}
[c]{cl}%
=\sum_{i=1}^{p^{2}}\sum_{j=1}^{k^{2}} & \left[  vec(G_{2})_{\perp}\right]
_{jm}\left[  vec(G_{1})_{\perp}\right]  _{il}\left[  vec\left(  \mathcal{R}%
(\hat{R})-\mathcal{R}(R)\right)  \right]  _{(j-1)k^{2}+i},
\end{array}
\end{align*}
for $m=1+\lfloor(q-1)/(k^{2}-1)\rfloor,$ $l=q-(p^{2}-1)(m-1),$ with $\lfloor
b\rfloor$ the entier function of a scalar $b$, which is of order $O_{p}\left(
n^{-\frac{1}{2}}\right)  .$ This convergence rate follows because $vec\left(
\mathcal{R}(\hat{R})-\mathcal{R}(R)\right)  $ is $O_{p}\left(  n^{-\frac{1}%
{2}}\right)  $ and $vec(G_{1})_{\perp}$ and $vec(G_{2})_{\perp}$ are both
orthonormal matrices. The sum of the weights $\left[  vec(G_{2})_{\perp
}\right]  _{jm}$ and $\left[  vec(G_{1})_{\perp}\right]  _{il}$ $i=1,\ldots
p^{2},$ $j=1,\ldots,k^{2}$ in the above summation is therefore finite and does
not grow with the sequence of $k$ and $p.$ Hence, it does not affect the
convergence rate which then results from $\mathcal{R}(\hat{R})-\mathcal{R}%
(R)=O_{p}\left(  n^{-\frac{1}{2}}\right)  .$%

\begin{align*}
2.  &  \left[  \left(  \left[  vec(\hat{G}_{2})_{\perp}-vec(G_{2})_{\perp
}\right]  \otimes vec(G_{1})_{\perp}\right)  ^{\prime}vec\left(
\mathcal{R}(\hat{R})-\mathcal{R}(R)\right)  \right]  _{q}\\
&
\begin{array}
[c]{cl}%
=\sum_{i=1}^{p^{2}}\sum_{j=1}^{k^{2}} & \left[  vec(\hat{G}_{2})_{\perp
}-vec(G_{2})_{\perp}\right]  _{jm}\left[  vec(G_{1})_{\perp}\right]
_{il}\left[  vec\left(  \mathcal{R}(\hat{R})-\mathcal{R}(R)\right)  \right]
_{(j-1)k^{2}+i},
\end{array}
\end{align*}
for $m=1+\lfloor(q-1)/(k^{2}-1)\rfloor,$ $l=q-(p^{2}-1)(m-1),$ which is of
order $O_{p}%
%TCIMACRO{\QOVERD{(}{)}{k^{2}}{n}}%
%BeginExpansion
\genfrac{(}{)}{}{}{k^{2}}{n}%
%EndExpansion
.$ This order results from the $k^{2}$ dependent components $\left[
vec(\hat{G}_{2})_{\perp}-vec(G_{2})_{\perp}\right]  _{jm}$ and $\left[
vec\left(  \mathcal{R}(\hat{R})-\mathcal{R}(R)\right)  \right]  _{(j-1)k^{2}%
+i}$ that we sum over and that the sum of the weights in the summation is
proportional to $k^{2}.$ Each of the (dependent) components in $\left[
vec(\hat{G}_{2})_{\perp}-vec(G_{2})_{\perp}\right]  _{jm}$ and $\left[
vec\left(  \mathcal{R}(\hat{R})-\mathcal{R}(R)\right)  \right]  _{(j-1)k^{2}%
+i}$ are $O_{p}(n^{-\frac{1}{2}})$ so summing over $k^{2}$ of them and
multiplying through results in $O_{p}%
%TCIMACRO{\QOVERD{(}{)}{k^{2}}{n}}%
%BeginExpansion
\genfrac{(}{)}{}{}{k^{2}}{n}%
%EndExpansion
.$ The additional weights $\left[  vec(G_{1})_{\perp}\right]  _{il},$
$i=1,\ldots,p^{2},$ are again such that their sum is finite so it does not
grow with the sequence of $k$ and $p$ because $vec(G_{1})_{\perp}$ is
orthonormal. Hence, they do not affect the convergence rate.
\begin{align*}
3.  &  \left[  \left(  \left[  vec(G_{2})_{\perp}\right]  \otimes\left[
vec(\hat{G}_{1})_{\perp}-vec(G_{1})_{\perp}\right]  \right)  ^{\prime
}vec\left(  \mathcal{R}(\hat{R})-\mathcal{R}(R)\right)  \right]  _{q}\\
&
\begin{array}
[c]{cl}%
=\sum_{i=1}^{p^{2}}\sum_{j=1}^{k^{2}} & \left[  vec(G_{2})_{\perp}\right]
_{jm}\left[  vec(\hat{G}_{1})_{\perp}-vec(G_{1})_{\perp}\right]  _{il}\left[
vec\left(  \mathcal{R}(\hat{R})-\mathcal{R}(R)\right)  \right]  _{(j-1)k^{2}%
+i},
\end{array}
\end{align*}
which is of order $O_{p}%
%TCIMACRO{\QOVERD{(}{)}{p^{2}}{n}}%
%BeginExpansion
\genfrac{(}{)}{}{}{p^{2}}{n}%
%EndExpansion
.$ The argument for this convergence rate is identical to the one for $2.$%
\begin{align*}
4.  &  \left[  \left(  \left[  vec(\hat{G}_{2})_{\perp}-vec(G_{2})_{\perp
}\right]  \otimes\left[  vec(\hat{G}_{1})_{\perp}-vec(G_{1})_{\perp}\right]
\right)  ^{\prime}vec\left(  \mathcal{R}(R)\right)  \right]  _{q}\\
&
\begin{array}
[c]{cl}%
=\sum_{i=1}^{p^{2}}\sum_{j=1}^{k^{2}} & \left[  vec(\hat{G}_{2})_{\perp
}-vec(G_{2})_{\perp}\right]  _{jm}\left[  vec(\hat{G}_{1})_{\perp}%
-vec(G_{1})_{\perp}\right]  _{il}\left[  vec\left(  \mathcal{R}(R)\right)
\right]  _{(j-1)k^{2}+i},
\end{array}
\end{align*}
which is of order $O_{p}%
%TCIMACRO{\QOVERD{(}{)}{(pk)^{2}}{n}}%
%BeginExpansion
\genfrac{(}{)}{}{}{(pk)^{2}}{n}%
%EndExpansion
.$ This order results from the double sum over $p^{2}$ random variables in
$\left[  vec(\hat{G}_{2})_{\perp}-vec(G_{2})_{\perp}\right]  $ and $k^{2}$
random variables in $\left[  vec(\hat{G}_{1})_{\perp}-vec(G_{1})_{\perp
}\right]  $ which are dependent. The sum of the weights is then proportional
to $(pk)^{2}$ and because the convergence rates of $\left[  vec(\hat{G}%
_{1})_{\perp}-vec(G_{1})_{\perp}\right]  $ and $\left[  vec(\hat{G}%
_{2})_{\perp}-vec(G_{2})_{\perp}\right]  $ are both $O_{p}(n^{-\frac{1}{2}}),$
this then leads to the $O_{p}%
%TCIMACRO{\QOVERD{(}{)}{(pk)^{2}}{n}}%
%BeginExpansion
\genfrac{(}{)}{}{}{(pk)^{2}}{n}%
%EndExpansion
$ convergence rate.
\begin{align*}
5.  &  \left[  \left(  \left[  vec(\hat{G}_{2})_{\perp}-vec(G_{2})_{\perp
}\right]  \otimes\left[  vec(\hat{G}_{1})_{\perp}-vec(G_{1})_{\perp}\right]
\right)  ^{\prime}vec\left(  \mathcal{R}(\hat{R})-\mathcal{R}(R)\right)
\right]  _{q}\\
&
\begin{array}
[c]{cl}%
=\sum_{i=1}^{p^{2}}\sum_{j=1}^{k^{2}} & \left[  vec(\hat{G}_{2})_{\perp
}-vec(G_{2})_{\perp}\right]  _{jm}\left[  vec(\hat{G}_{1})_{\perp}%
-vec(G_{1})_{\perp}\right]  _{il}\left[  vec\left(  \mathcal{R}(\hat
{R})-\mathcal{R}(R)\right)  \right]  _{(j-1)k^{2}+i},
\end{array}
\end{align*}
is of order $O_{p}%
%TCIMACRO{\QOVERD{(}{)}{(pk)^{2}}{n\sqrt{n}}}%
%BeginExpansion
\genfrac{(}{)}{}{}{(pk)^{2}}{n\sqrt{n}}%
%EndExpansion
$ which follows allong the lines of the above results.

For the limit behavior of $\sqrt{n}\hat{\Lambda}$ to just result from 1 (and
in consequence, the limit distribution of KPST to remain unaffected) it is
then sufficient to have joint limit sequences that satisfy:
\[%
\begin{array}
[c]{c}%
\frac{(pk)^{2}}{\sqrt{n}}\rightarrow0.
\end{array}
\]

For the estimator of the covariance matrix of $\hat{\Lambda},$ we further have%
\[%
\begin{array}
[c]{cc}%
\left(  \left[  vec(\hat{G}_{2})\right]  _{\perp}^{\prime}\otimes\left[
vec(\hat{G}_{1})\right]  _{\perp}^{\prime}\right)  \widehat{\text{cov(}%
}vec(\mathcal{R}(\hat{R})))\left(  \left[  vec(\hat{G}_{2})\right]  _{\perp
}^{\prime}\otimes\left[  vec(\hat{G}_{1})\right]  _{\perp}^{\prime}\right)  &
=\\
\left(  \left[  vec(G_{2})_{\perp}+vec(\hat{G}_{2})_{\perp}-vec(G_{2})_{\perp
}\right]  ^{\prime}\otimes\left[  vec(G_{1})_{\perp}+vec(\hat{G}_{1})_{\perp
}-vec(G_{1})_{\perp}\right]  \right)  ^{\prime} & \\
\left(  \text{cov(}vec(\mathcal{R}(\hat{R})))+\widehat{\text{cov(}%
}vec(\mathcal{R}(\hat{R})))-\text{cov(}vec(\mathcal{R}(\hat{R})))\right)  & \\
\left(  \left[  vec(G_{2})_{\perp}+vec(\hat{G}_{2})_{\perp}-vec(G_{2})_{\perp
}\right]  \otimes\left[  vec(G_{1})_{\perp}+vec(\hat{G}_{1})_{\perp}%
-vec(G_{1})_{\perp}\right]  \right)  & =\\
\left(  vec(G_{2})_{\perp}^{\prime}\otimes vec(G_{1})_{\perp}^{\prime}\right)
\text{cov(}vec(\mathcal{R}(\hat{R})))\left(  vec(G_{2})_{\perp}^{\prime
}\otimes vec(G_{1})_{\perp}^{\prime}\right)  ^{\prime}+U & =\\
A_{1}+B_{1}+B_{2}+B_{3}+C_{1}+C_{2}+C_{3}+C_{4}+C_{5}+C_{6}+D_{1}+\text{
}\ldots &
\end{array}
\]
where below we show that the maximal convergence rates besides the zero-th
order component are $O_{p}(n^{-\frac{1}{2}})$, $O_{p}\left(  \frac{(pk)^{2}%
}{n}\right)  ,$ $O_{p}\left(  \frac{k^{4}}{n}\right)  ,$ $O_{p}\left(
\frac{p^{4}}{n}\right)  $ and $O_{p}\left(  \frac{k^{4}p^{4}}{n\sqrt{n}%
}\right)  .$ All these rates appear in an identical manner in the inverse of
the estimator of the covariance matrix.\footnote{To show this, one can use the
Woodbury matrix identity which implies that for invertible $m\times m$
matrices $H$ and $G,$ with $H+G$ also invertible: $(H+G)^{-1}=H^{-1}%
-H^{-1}(G^{-1}+H^{-1})^{-1}H^{-1}.$} When taking the resulting inverse and
accounting for the summations over the $k^{2}p^{2}$ components in
$vec(\hat{\Lambda}),$ we obtain a slightly stronger condition than just for
$\hat{\Lambda}:$%
\[%
\begin{array}
[c]{c}%
\frac{(pk)^{16}}{n^{3}}\rightarrow0,
\end{array}
\]
which results from the $O_{p}(n^{-\frac{1}{2}})$ components from the inverse
of the covariance matrix estimator paired with the $O_{p}%
%TCIMACRO{\QOVERD{(}{)}{(pk)^{2}}{n}}%
%BeginExpansion
\genfrac{(}{)}{}{}{(pk)^{2}}{n}%
%EndExpansion
$ components from $\hat{\Lambda}$ corrected for the multiplication by $n$ and
the double summation over $p^{2}k^{2}$ components.\footnote{All combined we
get: $O_{p}\left(  n%
%TCIMACRO{\QOVERD{(}{)}{p^{2}k^{2}}{n}}%
%BeginExpansion
\genfrac{(}{)}{}{}{p^{2}k^{2}}{n}%
%EndExpansion%
%TCIMACRO{\QOVERD{(}{)}{p^{2}k^{2}}{n}}%
%BeginExpansion
\genfrac{(}{)}{}{}{p^{2}k^{2}}{n}%
%EndExpansion
\frac{1}{\sqrt{n}}(k^{2}p^{2})^{2}\right)  =O_{p}\left(
%TCIMACRO{\QOVERD{(}{)}{(p^{2}k^{2})^{4}}{n\sqrt{n}}}%
%BeginExpansion
\genfrac{(}{)}{}{}{(p^{2}k^{2})^{4}}{n\sqrt{n}}%
%EndExpansion
\right)  =O_{p}%
%TCIMACRO{\QOVERD{(}{)}{(pk)^{16}}{N^{3}}}%
%BeginExpansion
\genfrac{(}{)}{}{}{(pk)^{16}}{N^{3}}%
%EndExpansion
.$} The rate that would result from $\hat{\Lambda}$ is $\frac{(pk)^{12}}%
{n^{3}}\rightarrow0.$ The convergence rate is in between the rate implied by
\cite{newey2009gmm} which would be $\frac{k^{4}p^{4}}{n}$ for
$\hat{\Lambda}$ and $\frac{k^{6}p^{6}}{n}$ for convergence of the test
statistic which is slightly stricter than our rate of $\frac{(pk)^{16}}{n^{3}%
}\rightarrow0.$

Below, we state the rates of the different $A,$ $B,$ $C$ and $D$ (third order
error) components where we only provide the rate for one of the $D$ components
because we just showed that they do not lead to the largest error rate because
the $O_{p}\left(  \frac{k^{4}p^{4}}{n\sqrt{n}}\right)  $ is less than the
$O_{p}\left(  \frac{p^{2}k^{2}}{n}\right)  $ that results from some of the $C$
components.%
\begin{align*}
A_{1}=  &  \left(  vec(G_{2})_{\perp}^{\prime}\otimes vec(G_{1})_{\perp
}^{\prime}\right)  \text{cov(}vec(\mathcal{R}(\hat{R})))\left(  vec(G_{2}%
)_{\perp}\otimes vec(G_{1})_{\perp}\right)  &  &  =O(1)\\
B_{1}=  &  \left(  \left[  vec(\hat{G}_{2})_{\perp}-vec(G_{2})_{\perp}\right]
^{\prime}\otimes vec(G_{1})_{\perp}^{\prime}\right)  \text{cov(}%
vec(\mathcal{R}(\hat{R}))) &  & \\
&  \left(  vec(G_{2})_{\perp}\otimes vec(G_{1})_{\perp}\right)  +\left(
vec(G_{2})_{\perp}^{\prime}\otimes vec(G_{1})_{\perp}^{\prime}\right)  &  & \\
&  \text{cov(}vec(\mathcal{R}(\hat{R})))\left(  \left[  vec(\hat{G}%
_{2})_{\perp}-vec(G_{2})_{\perp}\right]  \otimes vec(G_{1})_{\perp}\right)  &
&  =O_{p}(n^{-\frac{1}{2}})\\
B_{2}=  &  \left(  vec(G_{2})_{\perp}^{\prime}\otimes\left[  vec(\hat{G}%
_{1})_{\perp}-vec(G_{1})_{\perp}\right]  ^{\prime}\right)  \text{cov(}%
vec(\mathcal{R}(\hat{R}))) &  & \\
&  \left(  vec(G_{2})_{\perp}^{\prime}\otimes vec(G_{1})_{\perp}^{\prime
}\right)  ^{\prime}+\left(  vec(G_{2})_{\perp}^{\prime}\otimes vec(G_{1}%
)_{\perp}^{\prime}\right)  &  & \\
&  \text{cov(}vec(\mathcal{R}(\hat{R})))\left(  vec(G_{2})_{\perp}%
\otimes\left[  vec(\hat{G}_{1})_{\perp}-vec(G_{1})_{\perp}\right]  \right)  &
&  =O_{p}(n^{-\frac{1}{2}})\\
B_{3}=  &  \left(  vec(G_{2})_{\perp}^{\prime}\otimes vec(G_{1})_{\perp
}^{\prime}\right)  \left[  \widehat{\text{cov(}}vec(\mathcal{R}(\hat
{R})))-\text{cov(}vec(\mathcal{R}(\hat{R})))\right]  &  & \\
&  \left(  vec(G_{2})_{\perp}^{\prime}\otimes vec(G_{1})_{\perp}^{\prime
}\right)  ^{\prime} &  &  =O_{p}(n^{-\frac{1}{2}})\\
C_{1}=  &  \left(  \left[  vec(\hat{G}_{2})_{\perp}-vec(G_{2})_{\perp}\right]
^{\prime}\otimes\left[  vec(\hat{G}_{1})_{\perp}-vec(G_{1})_{\perp}\right]
^{\prime}\right)  &  & \\
&  \text{cov(}vec(\mathcal{R}(\hat{R})))\left(  vec(G_{2})_{\perp}\otimes
vec(G_{1})_{\perp}\right)  +\left(  vec(G_{2})_{\perp}^{\prime}\otimes
vec(G_{1})_{\perp}^{\prime}\right)  &  & \\
&  \text{cov(}vec(\mathcal{R}(\hat{R})))\left(  \left[  vec(\hat{G}%
_{2})_{\perp}-vec(G_{2})_{\perp}\right]  \otimes\left[  vec(\hat{G}%
_{1})_{\perp}-vec(G_{1})_{\perp}\right]  \right)  &  &  =O_{p}\left(
\frac{(pk)^{2}}{n}\right) \\
C_{2}=  &  \left(  \left[  vec(\hat{G}_{2})_{\perp}-vec(G_{2})_{\perp}\right]
^{\prime}\otimes vec(G_{1})_{\perp}^{\prime}\right)  \text{cov(}%
vec(\mathcal{R}(\hat{R}))) &  & \\
&  \left(  \left[  vec(\hat{G}_{2})_{\perp}-vec(G_{2})_{\perp}\right]  \otimes
vec(G_{1})_{\perp}\right)  &  &  =O_{p}\left(  \frac{k^{4}}{n}\right) \\
C_{3}=  &  \left(  vec(G_{2})_{\perp}^{\prime}\otimes\left[  vec(\hat{G}%
_{1})_{\perp}-vec(G_{1})_{\perp}\right]  ^{\prime}\right)  \text{cov(}%
vec(\mathcal{R}(\hat{R}))) &  & \\
&  \left(  vec(G_{2})_{\perp}\otimes\left[  vec(\hat{G}_{1})_{\perp}%
-vec(G_{1})_{\perp}\right]  \right)  &  &  =O_{p}\left(  \frac{p^{4}}%
{n}\right) \\
C_{4}=  &  \left(  \left[  vec(\hat{G}_{2})_{\perp}-vec(G_{2})_{\perp}\right]
^{\prime}\otimes vec(G_{1})_{\perp}^{\prime}\right)  \text{cov(}%
vec(\mathcal{R}(\hat{R}))) &  & \\
&  \left(  vec(G_{2})_{\perp}\otimes\left[  vec(\hat{G}_{1})_{\perp}%
-vec(G_{1})_{\perp}\right]  \right)  + &  & \\
&  \left(  vec(G_{2})_{\perp}^{\prime}\otimes\left[  vec(\hat{G}_{1})_{\perp
}-vec(G_{1})_{\perp}\right]  ^{\prime}\right)  \text{cov(}vec(\mathcal{R}%
(\hat{R}))) &  & \\
&  \left(  \left[  vec(\hat{G}_{2})_{\perp}-vec(G_{2})_{\perp}\right]  \otimes
vec(G_{1})_{\perp}\right)  &  &  =O_{p}\left(  \frac{p^{2}k^{2}}{n}\right) \\
C_{5}=  &  \left(  \left[  vec(\hat{G}_{2})_{\perp}-vec(G_{2})_{\perp}\right]
^{\prime}\otimes vec(G_{1})_{\perp}^{\prime}\right)  \left[
\widehat{\text{cov(}}vec(\mathcal{R}(\hat{R})))-\right.  &  & \\
&  \left.  \text{cov(}vec(\mathcal{R}(\hat{R})))\right]  \left(
vec(G_{2})_{\perp}\otimes vec(G_{1})_{\perp}\right)  + &  & \\
&  \left(  vec(G_{2})_{\perp}^{\prime}\otimes vec(G_{1})_{\perp}^{\prime
}\right)  \left[  \widehat{\text{cov(}}vec(\mathcal{R}(\hat{R})))-\text{cov(}%
vec(\mathcal{R}(\hat{R})))\right]  &  & \\
&  \left(  \left[  vec(\hat{G}_{2})_{\perp}-vec(G_{2})_{\perp}\right]  \otimes
vec(G_{1})_{\perp}\right)  &  &  =O_{p}\left(  \frac{p^{2}k^{2}}{n}\right) \\
C_{6}=  &  \left(  vec(G_{2})_{\perp}^{\prime}\otimes\left[  vec(\hat{G}%
_{1})_{\perp}-vec(G_{1})_{\perp}\right]  ^{\prime}\right)  \left[
\widehat{\text{cov(}}vec(\mathcal{R}(\hat{R})))-\right.  &  & \\
&  \left.  \text{cov(}vec(\mathcal{R}(\hat{R})))\right]  \left(
vec(G_{2})_{\perp}\otimes vec(G_{1})_{\perp}\right)  + &  & \\
&  \left(  vec(G_{2})_{\perp}^{\prime}\otimes vec(G_{1})_{\perp}^{\prime
}\right)  \left[  \widehat{\text{cov(}}vec(\mathcal{R}(\hat{R})))-\text{cov(}%
vec(\mathcal{R}(\hat{R})))\right]  &  & \\
&  \left(  vec(G_{2})_{\perp}\otimes\left[  vec(\hat{G}_{1})_{\perp}%
-vec(G_{1})_{\perp}\right]  \right)  &  &  =O_{p}\left(  \frac{p^{2}k^{2}}%
{n}\right) \\
D_{1}=  &  \left(  \left[  vec(\hat{G}_{2})_{\perp}-vec(G_{2})_{\perp}\right]
^{\prime}\otimes\left[  vec(\hat{G}_{1})_{\perp}-vec(G_{1})_{\perp}\right]
^{\prime}\right)  &  & \\
&  \text{cov(}vec(\mathcal{R}(\hat{R})))^{\prime}\left(  \left[  vec(\hat
{G}_{2})_{\perp}-vec(G_{2})_{\perp}\right]  _{\perp}\otimes vec(G_{1})_{\perp
}\right)  + &  & \\
&  \left(  \left[  vec(\hat{G}_{2})_{\perp}-vec(G_{2})_{\perp}\right]
_{\perp}\otimes vec(G_{1})_{\perp}\right)  ^{\prime}\text{cov(}vec(\mathcal{R}%
(\hat{R}))) &  & \\
&  \left(  \left[  vec(\hat{G}_{2})_{\perp}-vec(G_{2})_{\perp}\right]
\otimes\left[  vec(\hat{G}_{1})_{\perp}-vec(G_{1})_{\perp}\right]  \right)  &
&  =O_{p}\left(  \frac{k^{4}p^{4}}{n\sqrt{n}}\right)
\end{align*}

\noindent\textbf{Proof of Theorem \ref{th: kps test power}.} 
Note first that for $i=1,2$,
\begin{equation}
n^{1/2}(\bar{G}_{i,n}-G_{i})\rightarrow g_{i}  \label{eq: conv}
\end{equation}
as $n\rightarrow \infty $ for certain matrices $g_{i}.$ This can be shown as
follows. Consider $M(t)=vec(G_{1})vec(G_{2})^{\prime }+t\mathcal{R}(A_{0})$
for $t\in R.$ For $|t|$ small enough, $\sigma _{1}(M(t)),$ the largest
singular value of $M(t),$ is simple and obviously $M(t)$ depends
differentiably on $t.$ By Theorems 7 and 8 in \cite{lax2007} it follows that for
small $|t|,$ $\sigma _{1}(M(t))$ depends differentiably on $t$ and that
there exists left and right eigenvectors corresponding to $\sigma _{1}(M(t))$
that depend differentiably on $t.$\footnote{%
The results in \cite{lax2007} are formulated for square matrices and spectral
decompositions but immediately translate to rectangular matrices $M$ and
their SVDs by considering $M^{\prime }M.$%
\par
{}} From (\ref{eq: g1g2 est}) we know that $vec(\bar{G}_{1,n})$ equals the
normalized left eigenvector, $L_{1}(n^{-1/2})$ say, corresponding to $\sigma
_{1}(M(n^{-1/2})).$ Given that $L_{1}(t)$ is differentiable at $t=0$ it
follows that $n^{1/2}(L_{1}(n^{-1/2})-L_{1}(0))$ converges to some vector $%
vec(g_{1})\in \mathbb{R}^{p^{2}}.$ But this proves the claim for $\bar{G}%
_{1,n}.$ The proof for $\bar{G}_{2,n}$ is identical. Intuitively, $(vec(G_1)+O(n^{-a}))(vec(G_2)+O(n^{-b}))'=vec(G_1)vec(G_2)'+O(n^{-1/2})$ implies that $a=b\geq 1/2$, see \cite{kleibergen2006grr}.

Second, note that $vec(\bar{G}_{i,n})_{\perp }$ for $i=1,2$ can be specified
such that 
\begin{equation}
vec(\bar{G}_{i,n})_{\perp }-vec(G_{i})_{\perp }=O(n^{-1/2}).
\label{eq: conv2}
\end{equation}%
Namely, let $vec(G_{1})_{\perp }=(v_{2},...,v_{p^{2}}).$ Clearly $(vec(\bar{G%
}_{1,n}),v_{2},...,v_{p^{2}})$ will be of full rank for all $n$ large enough
and $vec(\bar{G}_{1,n})_{\perp }$ can be obtained as the last $p^{2}-1$
vectors by performing Gram-Schmidt orthogonalization to $(vec(\bar{G}%
_{1,n}),v_{2},...,v_{p^{2}})$. E.g., the first column of $vec(\bar{G}%
_{1,n})_{\perp }$ (before normalizing its length to one) then equals 
\begin{eqnarray*}
&&v_{2}-[vec(\bar{G}_{1,n})^{\prime }v_{2}]vec(\bar{G}_{1,n})/||vec(\bar{G}%
_{1,n})|| \\
&=&v_{2}-[(vec(G_{1})+n^{-1/2}g_{1}+o(n^{-1/2}))^{\prime }v_{2}]vec(\bar{G}%
_{1,n})/||vec(\bar{G}_{1,n})|| \\
&=&v_{2}-[(n^{-1/2}g_{1}+o(n^{-1/2}))^{\prime }v_{2}]vec(\bar{G}%
_{1,n})/||vec(\bar{G}_{1,n})|| \\
&=&v_{2}+O(n^{-1/2}),
\end{eqnarray*}%
where in the first equality we use (\ref{eq: conv}) and in the second
equality $vec(G_{1})^{\prime }v_{2}=0.$ Continuing further with Gram-Schmidt
orthogonalization with the other columns of $vec(\bar{G}_{1,n})_{\perp }$
yields the desired result (\ref{eq: conv2}). The result for $%
vec(G_{2})_{\perp }$ is established analogously.

Next, by the definition of $\Lambda _{n}$ in (\ref{eq: Lambda_n}), we have%
\begin{equation*}
\lim_{n\rightarrow \infty }\sqrt{n}\Lambda _{n}=\lim_{n\rightarrow \infty }%
\sqrt{n}vec(\bar{G}_{1,n})_{\perp }^{\prime }\left(
vec(G_{1})vec(G_{2})^{\prime }+\frac{1}{\sqrt{n}}\mathcal{R}(A_{0})\right)
vec(\bar{G}_{2,n})_{\perp }=a_{0},
\end{equation*}%
where for the last equality we use (\ref{eq: conv2}). Namely,%
\begin{eqnarray*}
&&\sqrt{n}vec(\bar{G}_{1,n})_{\perp }^{\prime }vec(G_{1})vec(G_{2})^{\prime
}vec(\bar{G}_{2,n})_{\perp } \\
&=&(\sqrt{n}vec(G_{1})_{\perp }+O(1))^{\prime }vec(G_{1})vec(G_{2})^{\prime
}(vec(G_{2})_{\perp }+O(n^{-1/2})) \\
&=&O(1)vec(G_{1})vec(G_{2})^{\prime }O(n^{-1/2}) \\
&=&O(n^{-1/2}).
\end{eqnarray*}

Recall that KPST is defined as a quadratic form in $\sqrt{n}vec(\hat{%
\Lambda}).$ To derive the limiting distribution of the latter quantity, note
first that its asymptotic variance is given by 
\begin{equation*}
\begin{array}{rl}
V_{\Lambda }:= & \lim_{n\rightarrow \infty }\left[ vec(\bar{G}_{2,n})_{\perp
}^{\prime }\otimes vec(\bar{G}_{1,n})_{\perp }^{\prime }\right] cov\left( 
\mathcal{R}\left( \hat{R}\right) \right) \\ 
& \left[ vec(\bar{G}_{2,n})_{\perp }\otimes vec(\bar{G}_{1,n})_{\perp }%
\right] \\ 
= & \left( vec(G_{2})_{\perp }^{\prime }\otimes vec(G_{1})_{\perp }^{\prime
}\right) (D_{k}\otimes D_{p})V_{R^{\ast }} \\ 
& (D_{k}\otimes D_{p})^{\prime }\left( vec(G_{2})_{\perp }\otimes
vec(G_{1})_{\perp }\right) .%
\end{array}%
\end{equation*}%
The limiting distribution of $\sqrt{n}vec(\hat{\Lambda})$ under local to KPS
alternatives can be derived along the same lines as under KPS, which was
done in the proof of Theorem \ref{th: kps test}a, see (\ref{eq: lim
lambdahat}). It follows that $\sqrt{n}vec(\hat{\Lambda})\rightarrow
_{d}N\left( a_{0},V_{\Lambda }\right) $ and thus $KPST\rightarrow _{d}\chi
_{df}^{2}(\delta )$ as claimed.

{\small
\bibliographystyle{chicago}
\bibliography{KPS}

\begin{thebibliography}{}

\bibitem[\protect\citeauthoryear{Acemoglu, Cantoni, Johnson, and
  Robinson}{Acemoglu et~al.}{2011}]{ACJR2011}
Acemoglu, D., D.~Cantoni, S.~Johnson, and J.~A. Robinson (2011).
\newblock The {{Consequences}} of {{Radical Reform}}: {{The French
  Revolution}}.
\newblock {\em American Economic Review\/}~{\em 101\/}(7), 3286--3307.

\bibitem[\protect\citeauthoryear{Acemoglu and Johnson}{Acemoglu and
  Johnson}{2005}]{AJ2005}
Acemoglu, D. and S.~Johnson (2005).
\newblock Unbundling {{Institutions}}.
\newblock {\em Journal of Political Economy\/}~{\em 113\/}(5), 949--995.

\bibitem[\protect\citeauthoryear{Acemoglu, Johnson, Robinson, and
  Yared}{Acemoglu et~al.}{2008}]{AJRY2008}
Acemoglu, D., S.~Johnson, J.~A. Robinson, and P.~Yared (2008).
\newblock Income and {{Democracy}}.
\newblock {\em American Economic Review\/}~{\em 98\/}(3), 808--42.

\bibitem[\protect\citeauthoryear{Alesina, Giuliano, and Nunn}{Alesina
  et~al.}{2013}]{AGN2013}
Alesina, A., P.~Giuliano, and N.~Nunn (2013).
\newblock On the {{Origins}} of {{Gender Roles}}: {{Women}} and the {{Plough}}.
\newblock {\em The Quarterly Journal of Economics\/}~{\em 128\/}(2), 469--530.

\bibitem[\protect\citeauthoryear{Andrews}{Andrews}{1991}]{Andr91}
Andrews, D. W.~K. (1991).
\newblock Heteroskedasticity and autocorrelation consistent covariance matrix
  estimation.
\newblock {\em Econometrica\/}~{\em 59\/}(3), 817--858.

\bibitem[\protect\citeauthoryear{Andrews}{Andrews}{2017}]{and17}
Andrews, D. W.~K. (2017).
\newblock Identification-robust subvector inference.
\newblock Cowles foundation discussion papers 3005, Cowles Foundation for
  Research in Economics, Yale University.

\bibitem[\protect\citeauthoryear{Autor and Dorn}{Autor and Dorn}{2013}]{AD2013}
Autor, D.~H. and D.~Dorn (2013).
\newblock The {{Growth}} of {{Low}}-{{Skill Service Jobs}} and the
  {{Polarization}} of the {{US Labor Market}}.
\newblock {\em American Economic Review\/}~{\em 103\/}(5), 1553--97.

\bibitem[\protect\citeauthoryear{Autor, Dorn, and Hanson}{Autor
  et~al.}{2013}]{ADG2013}
Autor, D.~H., D.~Dorn, and G.~H. Hanson (2013).
\newblock The {{China Syndrome}}: {{Local Labor Market Effects}} of {{Import
  Competition}} in the {{United States}}.
\newblock {\em American Economic Review\/}~{\em 103\/}(6), 2121--68.

\bibitem[\protect\citeauthoryear{Chen and Fang}{Chen and Fang}{2019}]{Chen2019}
Chen, Q. and Z.~Fang (2019).
\newblock Improved inference on the rank of a matrix.
\newblock {\em Quantitative Economics\/}~{\em {\bf 10}}, 1787--1824.

\bibitem[\protect\citeauthoryear{Dahl and Lochner}{Dahl and
  Lochner}{2012}]{DL2012}
Dahl, G.~B. and L.~Lochner (2012).
\newblock The {{Impact}} of {{Family Income}} on {{Child Achievement}}:
  {{Evidence}} from the {{Earned Income Tax Credit}}.
\newblock {\em American Economic Review\/}~{\em 102\/}(5), 1927--56.

\bibitem[\protect\citeauthoryear{Donald, Fortuna, and Pipiras}{Donald
  et~al.}{2007}]{dfp07}
Donald, S.~G., N.~Fortuna, and V.~Pipiras (2007).
\newblock {On Rank Estimation in Suymmetric Matrices: The Case of Indefinite
  Matrix Estimators}.
\newblock {\em Journal of Econometrics\/}~{\em {\bf 23}}, 1217--1232.

\bibitem[\protect\citeauthoryear{Dufour and Taamouti}{Dufour and
  Taamouti}{2005}]{dufour2005projection}
Dufour, J.-M. and M.~Taamouti (2005).
\newblock Projection-based statistical inference in linear structural models
  with possibly weak instruments.
\newblock {\em Econometrica\/}~{\em 73\/}(4), 1351--1365.

\bibitem[\protect\citeauthoryear{Duranton and Turner}{Duranton and
  Turner}{2011}]{DT2011}
Duranton, G. and M.~A. Turner (2011).
\newblock The {{Fundamental Law}} of {{Road Congestion}}: {{Evidence}} from
  {{US Cities}}.
\newblock {\em American Economic Review\/}~{\em 101\/}(6), 2616--52.

\bibitem[\protect\citeauthoryear{Genton}{Genton}{2007}]{Genton07}
Genton, M.~G. (2007).
\newblock {Separable approximatioms of space-time covariance matrices}.
\newblock {\em Environmetrics\/}~{\em {\bf 18}}, 681--695.

\bibitem[\protect\citeauthoryear{Guggenberger, Kleibergen, and
  Mavroeidis}{Guggenberger et~al.}{2019}]{gkm19}
Guggenberger, P., F.~Kleibergen, and S.~Mavroeidis (2019).
\newblock {A more powerful subvector Anderson Rubin test in linear instrumental
  variable regression}.
\newblock {\em Quantitive Economics\/}~{\em 10\/}(2), 487--526.

\bibitem[\protect\citeauthoryear{Guggenberger, Kleibergen, and
  Mavroeidis}{Guggenberger et~al.}{2021}]{GKM21}
Guggenberger, P., F.~Kleibergen, and S.~Mavroeidis (2021).
\newblock A powerful subvector anderson rubin test in linear instrumental
  variables regression with conditional heteroskedasticity.
\newblock {\em arXiv preprint arXiv:2103.11371\/}.

\bibitem[\protect\citeauthoryear{Hansen}{Hansen}{1982}]{Hans82l}
Hansen, L.~P. (1982).
\newblock Large sample properties of generalized method of moments estimators.
\newblock {\em Econometrica\/}~{\em 50}, 1029--1054.

\bibitem[\protect\citeauthoryear{Hansen, Heaton, and Yaron}{Hansen
  et~al.}{1996}]{Hans96l}
Hansen, L.~P., J.~Heaton, and A.~Yaron (1996).
\newblock Finite sample properties of some alternative {GMM} estimators.
\newblock {\em Journal of Business and Economic Statistics\/}~{\em 14},
  262--280.

\bibitem[\protect\citeauthoryear{Hansford and Gomez}{Hansford and
  Gomez}{2010}]{HG2010}
Hansford, T.~G. and B.~T. Gomez (2010).
\newblock Estimating the {{Electoral Effects}} of {{Voter Turnout}}.
\newblock {\em American Political Science Review\/}~{\em 104\/}(2), 268--288.

\bibitem[\protect\citeauthoryear{Johnson, Parker, and Souleles}{Johnson
  et~al.}{2006}]{JPS2006}
Johnson, D.~S., J.~A. Parker, and N.~S. Souleles (2006).
\newblock Household {{Expenditure}} and the {{Income Tax Rebates}} of 2001.
\newblock {\em American Economic Review\/}~{\em 96\/}(5), 1589--1610.

\bibitem[\protect\citeauthoryear{Kan and Zhang}{Kan and Zhang}{1999}]{kz99}
Kan, R. and C.~Zhang (1999).
\newblock Two-pass tests of asset pricing models with useless factors.
\newblock {\em Journal of Finance\/}~{\em 54\/}(1), 203--235.

\bibitem[\protect\citeauthoryear{Kleibergen}{Kleibergen}{2005}]{Kleibergen2005}
Kleibergen, F. (2005).
\newblock Testing parameters in {GMM} without assuming that they are
  identified.
\newblock {\em Econometrica\/}~{\em 73\/}(4), 1103--1123.

\bibitem[\protect\citeauthoryear{Kleibergen}{Kleibergen}{2009}]{kf09}
Kleibergen, F. (2009).
\newblock {Tests of Risk Premia in Linear Factor Models}.
\newblock {\em Journal of Econometrics\/}~{\em {\bf 149}}, 149--173.

\bibitem[\protect\citeauthoryear{Kleibergen}{Kleibergen}{2021}]{kf21}
Kleibergen, F. (2021).
\newblock {Efficient size correct subset inference in homoskedastic linear
  instrumental variables regression}.
\newblock {\em Journal of Econometrics\/}~{\em 221\/}(1), 78--96.

\bibitem[\protect\citeauthoryear{Kleibergen and Paap}{Kleibergen and
  Paap}{2006}]{kleibergen2006grr}
Kleibergen, F. and R.~Paap (2006).
\newblock Generalized reduced rank tests using the singular value
  decomposition.
\newblock {\em Journal of Econometrics\/}~{\em 133\/}(1), 97--126.

\bibitem[\protect\citeauthoryear{Kleibergen and Zhan}{Kleibergen and
  Zhan}{2020}]{kz2020}
Kleibergen, F. and Z.~Zhan (2020).
\newblock Robust inference for consumption-based asset pricing.
\newblock {\em The Journal of Finance\/}~{\em 75\/}(1), 507--550.

\bibitem[\protect\citeauthoryear{Lax}{Lax}{2007}]{lax2007}
Lax, P.~D. (2007).
\newblock {\em {Linear Algebra and its Applications}\/} (2nd ed.).
\newblock Wiley-Interscience.

\bibitem[\protect\citeauthoryear{Ledoit and Wolf}{Ledoit and Wolf}{2012}]{lw12}
Ledoit, O. and M.~Wolf (2012).
\newblock {Nonlinear shrinkage estimation of large-dimensional covariance
  matrices}.
\newblock {\em Annals of Statistics\/}~{\em {\bf 40}}, 1024--1060.

\bibitem[\protect\citeauthoryear{Ledoit and Wolf}{Ledoit and Wolf}{2015}]{lw15}
Ledoit, O. and M.~Wolf (2015).
\newblock {Spectrum Estimation: a unified approach for covariance estimation
  and PCA in large dimensions}.
\newblock {\em Journal of Multivariate Analysis\/}~{\em {\bf 139}}, 360--384.

\bibitem[\protect\citeauthoryear{Ledoit and Wolf}{Ledoit and Wolf}{2018}]{lw18}
Ledoit, O. and M.~Wolf (2018).
\newblock {Optimal estimation of a large-dimensional covariance matrix under
  Stein's loss}.
\newblock {\em Bernouilli\/}~{\em {\bf 24}}, 3791--3832.

\bibitem[\protect\citeauthoryear{Lu and Zimmermann}{Lu and
  Zimmermann}{2005}]{LuZim05}
Lu, N. and D.~L. Zimmermann (2005).
\newblock {The likelihood ratio test for a separable covariance matrix}.
\newblock {\em {Statistics and Probability Letters}\/}~{\em {\bf 73}},
  449--457.

\bibitem[\protect\citeauthoryear{Miguel, Satyanath, and Sergenti}{Miguel
  et~al.}{2004}]{MSS2004}
Miguel, E., S.~Satyanath, and E.~Sergenti (2004).
\newblock Economic {{Shocks}} and {{Civil Conflict}}: {{An Instrumental
  Variables Approach}}.
\newblock {\em Journal of Political Economy\/}~{\em 112\/}(4), 725--753.

\bibitem[\protect\citeauthoryear{Mitchell, Genton, and Gumpertz}{Mitchell
  et~al.}{2006}]{mgg06}
Mitchell, M., M.~Genton, and M.~Gumpertz (2006).
\newblock {A likelihood ratio test for separability of ccovariance}.
\newblock {\em Journal of Multivariate Analysis\/}~{\em {\bf 97}}, 1025--1043.

\bibitem[\protect\citeauthoryear{Newey and Windmeijer}{Newey and
  Windmeijer}{2009}]{newey2009gmm}
Newey, W. and F.~Windmeijer (2009).
\newblock {GMM with many weak moment conditions}.
\newblock {\em Econometrica\/}~{\em 77\/}(3), 687--719.

\bibitem[\protect\citeauthoryear{Newey and West}{Newey and West}{1987}]{Newe87}
Newey, W.~K. and K.~D. West (1987).
\newblock A simple, positive semidefinite, heteroskedasticity and
  autocorrelation consistent covariance matrix.
\newblock {\em Econometrica\/}~{\em 55\/}(3), 703--708.

\bibitem[\protect\citeauthoryear{Nunn}{Nunn}{2008}]{Nunn2008}
Nunn, N. (2008).
\newblock The {{Long}}-{{Term Effects}} of {{Africa}}'s {{Slave Trades}}.
\newblock {\em The Quarterly Journal of Economics\/}~{\em 123\/}(1), 139--176.

\bibitem[\protect\citeauthoryear{Parker, Souleles, Johnson, and
  McClelland}{Parker et~al.}{2013}]{PSJM2013}
Parker, J.~A., N.~S. Souleles, D.~S. Johnson, and R.~McClelland (2013).
\newblock Consumer {{Spending}} and the {{Economic Stimulus Payments}} of 2008.
\newblock {\em American Economic Review\/}~{\em 103\/}(6), 2530--53.

\bibitem[\protect\citeauthoryear{Tanaka, Camerer, and Nguyen}{Tanaka
  et~al.}{2010}]{TCN2010}
Tanaka, T., C.~F. Camerer, and Q.~Nguyen (2010).
\newblock Risk and {{Time Preferences}}: {{Linking Experimental}} and
  {{Household Survey Data}} from {{Vietnam}}.
\newblock {\em American Economic Review\/}~{\em 100\/}(1), 557--71.

\bibitem[\protect\citeauthoryear{Van~Loan and Pitsianis}{Van~Loan and
  Pitsianis}{1993}]{vanloanpit93}
Van~Loan, C. and N.~Pitsianis (1993).
\newblock Approximation with kronecker products.
\newblock In {\em Linear algebra for large scale and real-time applications},
  NATO Adv. Sci. Inst. Ser. E Appl. Sci. 232, pp.\  293--314. Kluwer Academic
  Publishers.

\bibitem[\protect\citeauthoryear{Velu and Herman}{Velu and Herman}{2017}]{VH17}
Velu, R. and K.~Herman (2017).
\newblock {Separable Covariance Matrices and Kronecker Approximations}.
\newblock {\em Procedia Computer Science\/}~{\em {\bf 108}}, 1019--1029.

\bibitem[\protect\citeauthoryear{Voors, Nillesen, Verwimp, Bulte, Lensink, and
  Van~Soest}{Voors et~al.}{2012}]{Voors2012}
Voors, M.~J., E.~E. Nillesen, P.~Verwimp, E.~H. Bulte, R.~Lensink, and D.~P.
  Van~Soest (2012).
\newblock Violent {{Conflict}} and {{Behavior}}: {{A Field Experiment}} in
  {{Burundi}}.
\newblock {\em American Economic Review\/}~{\em 102\/}(2), 941--64.

\bibitem[\protect\citeauthoryear{Werner, Jansson, and Stoica}{Werner
  et~al.}{2008}]{wjs08}
Werner, K., M.~Jansson, and P.~Stoica (2008).
\newblock {On estimation of covariance matrices with Kronecker product
  structure}.
\newblock {\em IEEE Transactions of Signal Processing\/}~{\em {\bf 56}},
  478--491.

\bibitem[\protect\citeauthoryear{White}{White}{1980}]{White80}
White, H. (1980).
\newblock A heteroskedasticity-consistent covariance matrix estimator and a
  direct test for heteroskedasticity.
\newblock {\em Econometrica\/}~{\em 48\/}(4), 817--38.

\bibitem[\protect\citeauthoryear{White}{White}{1984}]{wh84}
White, H. (1984).
\newblock {\em {Asymptotic Theory for Econometricians}}.
\newblock Academic Press.

\bibitem[\protect\citeauthoryear{Yogo}{Yogo}{2004}]{Yogo2004}
Yogo, M. (2004).
\newblock Estimating the {{Elasticity}} of {{Intertemporal Substitution}} when
  {{Instruments}} are {{Weak}}.
\newblock {\em Review of Economics and Statistics\/}~{\em 86\/}(3), 797--810.

\end{thebibliography}
}

\newpage

\section{Supplementary Appendix: Detailed empirical results}

Tables \ref{tab:KPS app} and \ref{tab:clust} give detailed empirical results
in the applications considered, with non-clustered and clustered data,
respectively. The acronyms refer to the different papers listed in Table
\ref{TableKey papers}.%

%TCIMACRO{\TeXButton{B}{\begin{table}[tbp] \centering}}%
%BeginExpansion
\begin{table}[tbp] \centering
%EndExpansion%
\begin{tabular}
[c]{|l|l|}\hline
\textbf{Acronym} & \textbf{Paper}\\
ACJR 11 & Acemoglu et al. (2011)\\\hline
AD 13 & Autor and Dorn (2013)\\\hline
ADG 13 & Autor et al. (2013)\\\hline
AGN 13 & Alesina et al. (2013)\\\hline
AJ 05 & Acemoglu and Johnson (2005)\\\hline
AJRY 08 & Acemoglu et al. (2008)\\\hline
DT 11 & Duranton and Turner (2011)\\\hline
HG 10 & Hansford and Gomez (2010)\\\hline
JPS 06 & Johnson et al. (2006)\\\hline
MSS 04 & Miguel et al. (2004)\\\hline
Nunn 08 & Nunn (2008)\\\hline
PSJM 13 & Parker et al. (2013)\\\hline
TCN 10 & Tanaka et al. (2010)\\\hline
V et al 12 & Voors et al. (2012)\\\hline
Yogo 04 & Yogo (2004)\\\hline
\end{tabular}
\caption{List of papers used in the empirical applications.}\label{TableKey papers}%
%TCIMACRO{\TeXButton{E}{\end{table}}}%
%BeginExpansion
\end{table}%
%EndExpansion

\begin{landscape}
\begin{longtable}{l|l|p{6.25cm}|p{6.25cm}|c|c|r|r|r}
\caption{Applications of KPST.} \label{tab:KPS app} \\
\hline \multicolumn{1}{l|}{\textbf{Paper}} & \multicolumn{1}{l|}{\textbf{Specif.}} & \multicolumn{1}{c|}{$Y$} & \multicolumn{1}{c|}{$Z$} & \multicolumn{1}{c|}{$p$} & \multicolumn{1}{c|}{$k$} & \multicolumn{1}{r|}{$n$} & \multicolumn{1}{r|}{\textbf{KPST}} & \multicolumn{1}{c}{\textbf{p val}} \\ \hline
\endfirsthead
\multicolumn{9}{c}%
{{\bfseries \tablename\ \thetable{} -- continued from previous page}} \\
\hline \multicolumn{1}{l|}{\textbf{Paper}} & \multicolumn{1}{l|}{\textbf{Specif.}} & \multicolumn{1}{c|}{$Y$} & \multicolumn{1}{c|}{$Z$} & \multicolumn{1}{c|}{$p$} & \multicolumn{1}{c|}{$k$} & \multicolumn{1}{r|}{$n$} & \multicolumn{1}{r|}{\textbf{KPST}}&  \multicolumn{1}{c}{\textbf{p val}}\\ \hline
\endhead
\multicolumn{9}{r}{{Continued on next page}} \\
\endfoot
\endlastfoot
TCN 10  & T5.P2.C1                              & Value function curvature, Income                                                               & Rainfall, Head of Household Cannot Work (dummy variable)                                        & 2 & 2 & 181   & 4.944   & 0.293    \\
& T5.P2.C2                              & Value function curvature, Relative Income, Mean Income                                         & Rainfall, Head of Household Cannot Work (dummy variable)                                        & 3 & 2 & 181   & 14.859  & 0.137    \\ \hline
Nunn 08 & T4.C1                                 & Log income in 2000, Slave exports                                                              & Atlantic distance, Indian distance, Saharan distance, Red Sea distance                          & 2 & 4 & 52    & 32.307  & 0.02     \\
& T4.C2                                 & Log income in 2000, Slave exports, \newline (X: Colonization effect)                                           & Atlantic distance, Indian distance, Saharan distance, Red Sea distance                          & 2 & 4 & 52    & 30.922  & 0.029    \\
& T4.C3                                 & Log income in 2000, Slave exports, \newline (X: Col.~effect, geographical controls)                  & Atlantic distance, Indian distance, Saharan distance, Red Sea distance                          & 2 & 4 & 52    & 34.597  & 0.011    \\
& T4.C4                                 & Log income in 2000, Slave exports, \newline (X: Col.~effect, geographical controls)                                                              & Atlantic distance, Indian distance, Saharan distance, Red Sea distance                          & 2 & 4 & 42    & 28.263  & 0.058    \\ \hline
AJ 05   & T4.P1.C1                              & Log GDP per capita, legal formalism, constraint on executive                                   & English legal origin, settler mortality                                                         & 3 & 2 & 51    & 8.18    & 0.611    \\
& T4.P1.C2                              & Log GDP per capita, legal formalism, constraint on executive                                   & English legal origin, population density 1500                                                   & 3 & 2 & 60    & 25.969  & 0.004    \\
& T4.P1.C3                              & Log GDP per capita, constraint on executive, procedural complexity                             & English legal origin, settler mortality                                                         & 3 & 2 & 60    & 5.574   & 0.85     \\
& T4.P1.C4                              & Log GDP per capita, constraint on executive, number of procedures                              & English legal origin, settler mortality                                                         & 3 & 2 & 61    & 10.916  & 0.364    \\
& T4.P1.C5                              & Log GDP per capita, legal formalism, average protection against risk of expropriation          & English legal origin, settler mortality                                                         & 3 & 2 & 51    & 7.075   & 0.718    \\
& T4.P1.C6                              & Log GDP per capita, legal formalism, private property                                          & English legal origin, settler mortality                                                         & 3 & 2 & 52    & 8.646   & 0.566    \\
& T4.P2.C1                              & Investment-GDP ratio, legal formalism, constraint on executive                                 & English legal origin, settler mortality                                                         & 3 & 2 & 51    & 13.068  & 0.22     \\
& T4.P2.C2                              & Investment-GDP ratio, legal formalism, constraint on executive                                 & English legal origin, population density 1500                                                   & 3 & 2 & 60    & 36.298  & 0        \\
& T4.P2.C3                              & Investment-GDP ratio, constraint on executive, procedural complexity                           & English legal origin, settler mortality                                                         & 3 & 2 & 61    & 16.838  & 0.078    \\
& T4.P2.C4                              & Investment-GDP ratio, constraint on executive, number of procedures                            & English legal origin, settler mortality                                                         & 3 & 2 & 62    & 14.82   & 0.139    \\
& T4.P2.C5                              & Investment-GDP ratio, legal formalism, average protection against risk of expropriation        & English legal origin, settler mortality                                                         & 3 & 2 & 51    & 13.75   & 0.185    \\
& T4.P2.C6                              & Investment-GDP ratio, legal formalism, private property                                        & English legal origin, settler mortality                                                         & 3 & 2 & 52    & 8.582   & 0.572    \\
& T5.P1.C1                              & Private credit, legal formalism, constraint on executive                                       & English legal origin, settler mortality                                                         & 3 & 2 & 51    & 9.296   & 0.504    \\
& T5.P1.C2                              & Private credit, legal formalism, constraint on executive                                       & English legal origin, population density 1500                                                   & 3 & 2 & 60    & 31.406  & 0.001    \\
& T5.P1.C3                              & Private credit, constraint on executive, procedural complexity                                 & English legal origin, settler mortality                                                         & 3 & 2 & 60    & 13.721  & 0.186    \\
& T5.P1.C4                              & Private credit, constraint on executive, number of procedures                                  & English legal origin, settler mortality                                                         & 3 & 2 & 61    & 11.605  & 0.312    \\
& T5.P1.C5                              & Private credit, legal formalism, average protection against risk of expropriation              & English legal origin, settler mortality                                                         & 3 & 2 & 51    & 12.206  & 0.272    \\
& T5.P1.C6                              & Private credit, legal formalism, private property                                              & English legal origin, settler mortality                                                         & 3 & 2 & 52    & 19.304  & 0.037    \\
& T5.P2.C1                              & Stock market capitalization, legal formalism, constraint on executive                          & English legal origin, settler mortality                                                         & 3 & 2 & 50    & 19.178  & 0.038    \\
& T5.P2.C2                              & Stock market capitalization, legal formalism, constraint on executive                          & English legal origin, population density 1500                                                   & 3 & 2 & 59    & 19.405  & 0.035    \\
& T5.P2.C3                              & Stock market capitalization, constraint on executive, procedural complexity                    & English legal origin, settler mortality                                                         & 3 & 2 & 59    & 34.566  & 0        \\
& T5.P2.C4                              & Stock market capitalization, constraint on executive, number of procedures                     & English legal origin, settler mortality                                                         & 3 & 2 & 59    & 28.06   & 0.002    \\
& T5.P2.C5                              & Stock market capitalization, legal formalism, average protection against risk of expropriation & English legal origin, settler mortality                                                         & 3 & 2 & 50    & 35.531  & 0        \\
& T5.P2.C6                              & Stock market capitalization, legal formalism, private property                                 & English legal origin, settler mortality                                                         & 3 & 2 & 51    & 21.344  & 0.019    \\ \hline
HG 10   & T1.C2                                 & Democratic vote share, turnout, turnout * partisan composition, turnout * Republican incumbent & Rainfall, rainfall*partisan composition, raninfall*Republican incumbent                         & 4 & 3 & 27401 & 507.919 & 0        \\
& T1.C3                                 & Democratic vote share, turnout, turnout * partisan composition, turnout * Republican incumbent & Rainfall, rainfall*partisan composition, raninfall*Republican incumbent                         & 4 & 3 & 27401 & 457.962 & 0        \\ \hline
AGN 13  & T8.P3.C1                              & Female LF participation, Traditional plough use                                                & Plough-neg. environment, Plough-pos. environment                                                & 2 & 2 & 160   & 6.191   & 0.185    \\
& T8.P3.C2                              & Female LF participation, Traditional plough use                                                & Plough-neg. environment, Plough-pos. environment                                                & 2 & 2 & 160   & 4.939   & 0.294    \\
& T8.P3.C3                              & Share firm ownership female, Traditional plough use                                            & Plough-neg. environment, Plough-pos. environment                                                & 2 & 2 & 122   & 3.586   & 0.465    \\
& T8.P3.C4                              & Share firm ownership female, Traditional plough use                                            & Plough-neg. environment, Plough-pos. environment                                                & 2 & 2 & 122   & 6.785   & 0.148    \\
& T8.P3.C5                              & Share political position female, Traditional plough use                                        & Plough-neg. environment, Plough-pos. environment                                                & 2 & 2 & 140   & 9.29    & 0.054    \\
& T8.P3.C6                              & Share political position female, Traditional plough use                                        & Plough-neg. environment, Plough-pos. environment                                                & 2 & 2 & 140   & 10.982  & 0.027    \\ \hline
Yogo 04 & AUL                                   & cons growth, risk-free rtn                                                                     & \multirow{22}{*}{\parbox{6.25cm}{\hfill Twice lagged nominal interest rate, inflation,consumption growth, and log dividend-price ratio}} & 2 & 4 & 114   & 16.628  & 0.549    \\
&                                       & cons growth, stk mkt rtn                                                                       &                                                                                                & 2 & 4 & 114   & 22.879  & 0.195    \\
& CAN                                   & cons growth, risk-free rtn                                                                     &  & 2 & 4 & 115   & 24.078  & 0.152    \\
&                                       & cons growth, stk mkt rtn                                                                       &                                                                                                & 2 & 4 & 115   & 32.528  & 0.019    \\
& FRA                                   & cons growth, risk-free rtn                                                                     & & 2 & 4 & 113   & 28.015  & 0.062    \\
&                                       & cons growth, stk mkt rtn                                                                       &                                                                                                & 2 & 4 & 113   & 25.608  & 0.109    \\
& GER                                   & cons growth, risk-free rtn                                                                     & & 2 & 4 & 79    & 25.452  & 0.113    \\
&                                       & cons growth, stk mkt rtn                                                                       &                                                                                               & 2 & 4 & 79    & 31.24   & 0.027    \\
& ITA                                   & cons growth, risk-free rtn                                                                     & & 2 & 4 & 106   & 18.266  & 0.438    \\
&                                       & cons growth, stk mkt rtn                                                                       &                                                                                                & 2 & 4 & 106   & 25.889  & 0.102    \\
& JAP                                   & cons growth, risk-free rtn                                                                     &  & 2 & 4 & 114   & 22.835  & 0.197    \\
&                                       & cons growth, stk mkt rtn                                                                       &                                                                                                & 2 & 4 & 114   & 16.132  & 0.583    \\
& NTH                                   & cons growth, risk-free rtn                                                                     &  & 2 & 4 & 86    & 20.969  & 0.281    \\
&                                       & cons growth, stk mkt rtn                                                                       &                                                                                                & 2 & 4 & 86    & 21.762  & 0.243    \\
& SWD                                   & cons growth, risk-free rtn                                                                     &  & 2 & 4 & 116   & 18.967  & 0.394    \\
&                                       & cons growth, stk mkt rtn                                                                       &                                                                                                & 2 & 4 & 116   & 29.714  & 0.04     \\
& SWT                                   & cons growth, risk-free rtn                                                                     &  & 2 & 4 & 91    & 14.889  & 0.67     \\
&                                       & cons growth, stk mkt rtn                                                                       &                                                                                                & 2 & 4 & 91    & 43.768  & 0.001    \\
& UK                                    & cons growth, risk-free rtn                                                                     &  & 2 & 4 & 115   & 30.148  & 0.036    \\
&                                       & cons growth, stk mkt rtn                                                                       &                                                                                                & 2 & 4 & 115   & 19.94   & 0.336    \\
& US                                    & cons growth, risk-free rtn                                                                     &  & 2 & 4 & 114   & 18.478  & 0.425    \\
&                                       & cons growth, stk mkt rtn                                                                       &                                                                                                & 2 & 4 & 114   & 22.373  & 0.216
\\ \hline \hline
\multicolumn{9}{@{}l}{\parbox[t]{\linewidth}{Specification T: table; P: panel; C: column.}}
\end{longtable}
\end{landscape}

\begin{landscape}
\begin{longtable}{l|p{4cm}|p{4cm}|c|c|r|r|r|r|r|r}
\caption{Applications of cluster KPST.} \label{tab:clust} \\
\hline \multicolumn{1}{l|}{\textbf{Specif.}} & \multicolumn{1}{c|}{$Y$} & \multicolumn{1}{c|}{$Z$} & \multicolumn{1}{c|}{$p$} & \multicolumn{1}{c|}{$k$} & \multicolumn{1}{r|}{$n$} & \multicolumn{1}{r|}{\textbf{KPST}} & \multicolumn{1}{c|}{\textbf{p val}}
& \multicolumn{1}{r|}{$n_c$} & \multicolumn{1}{r|}{\textbf{KPST$_c$}} & \multicolumn{1}{c}{\textbf{p val}}\\ \hline
\endfirsthead
\multicolumn{11}{c}%
{{\bfseries \tablename\ \thetable{} -- continued from previous page}} \\
\hline \multicolumn{1}{l|}{\textbf{Specif.}} & \multicolumn{1}{c|}{$Y$} & \multicolumn{1}{c|}{$Z$} & \multicolumn{1}{c|}{$p$} & \multicolumn{1}{c|}{$k$} & \multicolumn{1}{r|}{$n$} & \multicolumn{1}{r|}{\textbf{KPST}}&  \multicolumn{1}{c|}{\textbf{p val}}
& \multicolumn{1}{r|}{$n_c$} & \multicolumn{1}{r|}{\textbf{KPST$_c$}} & \multicolumn{1}{c}{\textbf{p val}}\\ \hline
\endhead
\multicolumn{11}{r}{{Continued on next page}} \\
\endfoot
\endlastfoot
\textit{AJRY 08}   & & & & & & & & & & \\                                              T5.C5     & Freedom House measure of democracy, Log GDP per capita in t-1                    & Savings rate in t-2, Democracy in t-1                                                                    & 2 & 2 & 891   & 23.86   & 0.000 & 134         & 20.204     & 0.001 \\
T5.C7     & Freedom House measure of democracy, Log GDP per capita in t-1                    & Savings rate in t-2, labour share of income                                                              & 2 & 2 & 471   & 21.85   & 0.000 & 98          & 6.037      & 0.303 \\
T5.C8.S1  & Freedom House measure of democracy, Log GDP per capita in t-1                    & Savings rate in t-2, democracy in t-1  \newline X: democracy in t-2, t-3                                                                  & 2 & 2 & 471   & 17.21   & 0.002 & 98          & 13.500     & 0.019 \\
T5.C8.S2  & Freedom House measure of democracy, Log GDP per capita in t-1                    & Savings rate in t-2, democracy in t-2 \newline X: democracy in t-1, t-3         & 2 & 2 & 471   & 14.96   & 0.005 & 98          & 11.738     & 0.039 \\
T5.C8.S3  & Freedom House measure of democracy, Log GDP per capita in t-1                    & Savings rate in t-2, democracy in t-3 \newline X: democracy in t-1, t-2          & 2 & 2 & 471   & 6.83    & 0.145 & 98          & 4.388      & 0.495 \\
T5.C9     & Freedom House measure of democracy, Log GDP per capita in t-1                    & Savings rate in t-2, t-3                                                                & 2 & 2 & 796   & 12.14   & 0.016 & 125         & 18.960     & 0.002 \\
T6.C5     & Freedom House measure of democracy, Log GDP per capita in t-1                    & Trade-weighted (tw) log GDP in t-1, democracy in t-1                                                          & 2 & 2 & 796   & 4.71    & 0.318 & 125         & 12.970     & 0.024 \\
T6.C7     & Freedom House measure of democracy, Log GDP per capita in t-1                    & tw log GDP in t-1, tw democracy in t-1                                           & 2 & 2 & 796   & 10.18   & 0.037 & 125         & 11.808     & 0.038 \\
T6.C9     & Freedom House measure of democracy, Log GDP per capita in t-1                    & tw log GDP in t-1, t-2                                             & 2 & 2 & 796   & 12.83   & 0.012 & 125         & 12.121     & 0.033 \\ \hline
\textit{JPS 06}           & & & & & & & & & & \\
T4.P1.C5  & Dollar change in strict non-durables, rebate in t+1, t                & I (rebate t+1), I (rebate t)                                                                             & 3 & 2 & 12730 & 1062.30 & 0.000 & 6253        & 386.388    & 0.000 \\
T4.P1.C6  & Dollar change in nondurable goods, rebate in t+1, t                    & I (rebate t+1), I (rebate t)                                                                             & 3 & 2 & 12730 & 1062.05 & 0.000 & 6253        & 377.982    & 0.000 \\
T4.P2.C5  & Dollar change in strict non-durables, rebate in t+1, t, t-1 & I (rebate t+1), I (rebate t), I (rebate t-1)                                                             & 4 & 3 & 15022 & 1635.13 & 0.000 & 6295        & 1128.150   & 0.000 \\
T4.P2.C6  & Dollar change in nondurable goods, rebate in t+1, t, t-1     & I (rebate t+1), I (rebate t), I (rebate t-1)                                                             & 4 & 3 & 15022 & 1666.13 & 0.000 & 6295        & 1140.060   & 0.000 \\ \hline
\textit{PSJM 13}     & & & & & & & & & & \\
T4.P1.C5  & Nondurable spending, ESP by check, ESP by electronic transfer                    & I (ESP by check), I (ESP by electronic transfer)                                                         & 3 & 2 & 17281 & 457.30  & 0.000 & 8038        & 314.724    & 0.000 \\
T4.P1.C6  & All spending, ESP by check, ESP by electronic transfer                           & I (ESP by check), I (ESP by electronic transfer)                                                         & 3 & 2 & 17281 & 458.98  & 0.000 & 8038        & 288.445    & 0.000 \\ \hline
\textit{ADH 13}    & & & & & & & & & & \\
T10.P3.C1 & $\Delta$ mfg empl, $\Delta$ trade US-China net input pw (nipw)                                      & $\Delta$ trade other-China, $\Delta$ net input other-China                                                              & 2 & 2 & 1444  & 20.00   & 0.001 & 48          & 27.125     & 0.000 \\
T10.P3.C2 & $\Delta$ nonmfg empl, $\Delta$ trade US-China nipw                                     & $\Delta$ trade other-China, $\Delta$ net input other-China                                                              & 2 & 2 & 1444  & 22.95   & 0.000 & 48          & 24.312     & 0.000 \\
T10.P3.C3 & $\Delta$ mfg log wage,  $\Delta$ trade US-China nipw                                & $\Delta$ trade other-China, $\Delta$ net input other-China                                                              & 2 & 2 & 1444  & 31.27   & 0.000 & 48          & 19.553     & 0.002 \\
T10.P3.C4 & $\Delta$ mfg log wage, $\Delta$ trade US-China nipw                                 & $\Delta$ trade other-China, $\Delta$ net input other-China                                                              & 2 & 2 & 1444  & 19.40   & 0.001 & 48          & 22.269     & 0.000 \\
T10.P3.C5 & $\Delta$ nonmfg log wage, $\Delta$ trade US-China nipw                                & $\Delta$ trade other-China, $\Delta$ net input other-China                                                              & 2 & 2 & 1444  & 100.88  & 0.000 & 48          & 10.514     & 0.062 \\
T10.P3.C6 & $\Delta$ log transfers, $\Delta$ trade US-China nipw                            & $\Delta$ trade other-China, $\Delta$ net input other-China                                                              & 2 & 2 & 1444  & 21.82   & 0.000 & 48          & 16.716     & 0.005 \\
T10.P4.C1 & $\Delta$ mfg empl, $\Delta$ US-China net imports pw                                               & $\Delta$ trade other-China, $\Delta$ net exports other-China                                                                & 2 & 2 & 1444  & 16.52   & 0.002 & 48          & 10.187     & 0.070 \\
T10.P4.C2 & $\Delta$ nonmfg empl, $\Delta$ US-China net imp pw                                              & $\Delta$ trade other-China, $\Delta$ net exports other-China                                                                & 2 & 2 & 1444  & 18.44   & 0.001 & 48          & 10.014     & 0.075 \\
T10.P4.C3 & $\Delta$ mfg log wage, $\Delta$ US-China net imp pw                                          & $\Delta$ trade other-China, $\Delta$ net exports other-China                                                                & 2 & 2 & 1444  & 37.44   & 0.000 & 48          & 13.290     & 0.021 \\
T10.P4.C4 & $\Delta$ nonmfg log wage, $\Delta$ US-China net imp pw                                         & $\Delta$ trade other-China, $\Delta$ net exports other-China                                                                & 2 & 2 & 1444  & 11.21   & 0.024 & 48          & 11.072     & 0.050 \\
T10.P4.C5 & $\Delta$ log transfers, $\Delta$ US-China net imp pw                                     & $\Delta$ trade other-China, $\Delta$ net exports other-China                                                                & 2 & 2 & 1444  & 41.77   & 0.000 & 48          & 9.138      & 0.104 \\
T10.P4.C6 & $\Delta$ avg household inc, $\Delta$ US-China net imp pw                                 & $\Delta$ trade other-China, $\Delta$ net exports other-China                                                                & 2 & 2 & 1444  & 18.08   & 0.001 & 48          & 13.395     & 0.020 \\
T10.P6.C1 & $\Delta$ mfg empl, $\Delta$ net trade factor (ntf) US-China                                     & $\Delta$ ntf other-China, $\Delta$ net export factor (nef) other-China                                                & 2 & 2 & 1444  & 16.57   & 0.002 & 48          & 14.213     & 0.014 \\
T10.P6.C2 & $\Delta$ nonmfg empl, $\Delta$ ntf US-China                                    & $\Delta$ ntf other-China, $\Delta$ nef other-China                                                & 2 & 2 & 1444  & 43.88   & 0.000 & 48          & 15.611     & 0.008 \\
T10.P6.C3 & $\Delta$ mfg log wage, $\Delta$ ntf US-China                                & $\Delta$ ntf other-China, $\Delta$ nef other-China                                                & 2 & 2 & 1444  & 24.54   & 0.000 & 48          & 12.087     & 0.034 \\
T10.P6.C4 & $\Delta$ nonmfg log wage, $\Delta$ ntf US-China                               & $\Delta$ ntf other-China, $\Delta$ nef other-China                                                & 2 & 2 & 1444  & 10.81   & 0.029 & 48          & 18.869     & 0.002 \\
T10.P6.C5 & $\Delta$ log transfers, $\Delta$ ntf US-China                           & $\Delta$ ntf other-China, $\Delta$ nef other-China                                                & 2 & 2 & 1444  & 15.56   & 0.004 & 48          & 16.692     & 0.005 \\
T10.P6.C6 & $\Delta$ avg household inc, $\Delta$ ntf US-China                       & $\Delta$ ntf other-China, $\Delta$ nef other-China                                                & 2 & 2 & 1444  & 16.46   & 0.002 & 48          & 29.073     & 0.000 \\ \hline
\textit{AD 13}    & & & & & & & & & & \\
T5.P2.C1  & Growth of service employment, Share of routine employment (t-1)                  & 1950 employment share by commuting zone excluding those corresponding to observation:  1980; 1990;2000.` & 2 & 3 & 2166  & 141.50  & 0.000 & 48          & 57.891     & 0.000 \\
T5.P2.C2  & Growth of service employment, Share of routine employment (t-1)                  & 1951 employment share by commuting zone excluding those corresponding to observation:  1980; 1990;2000.` & 2 & 3 & 2166  & 122.97  & 0.000 & 48          & 41.735     & 0.000 \\
T5.P2.C3  & Growth of service employment, Share of routine employment (t-1)                  & 1952 employment share by commuting zone excluding those corresponding to observation:  1980; 1990;2000.` & 2 & 3 & 2166  & 140.57  & 0.000 & 48          & 52.603     & 0.000 \\
T5.P2.C4  & Growth of service employment, Share of routine employment (t-1)                  & 1953 employment share by commuting zone excluding those corresponding to observation:  1980; 1990;2000.` & 2 & 3 & 2166  & 118.33  & 0.000 & 48          & 47.893     & 0.000 \\
T5.P2.C5  & Growth of service employment, Share of routine employment (t-1)                  & 1954 employment share by commuting zone excluding those corresponding to observation:  1980; 1990;2000.` & 2 & 3 & 2166  & 106.08  & 0.000 & 48          & 47.248     & 0.000 \\
T5.P2.C6  & Growth of service employment, Share of routine employment (t-1)                  & 1955 employment share by commuting zone excluding those corresponding to observation:  1980; 1990;2000.` & 2 & 3 & 2166  & 146.81  & 0.000 & 48          & 43.400     & 0.000 \\
T5.P2.C7  & Growth of service employment, Share of routine employment (t-1)                  & 1956 employment share by commuting zone excluding those corresponding to observation:  1980; 1990;2000.` & 2 & 3 & 2166  & 101.50  & 0.000 & 48          & 32.647     & 0.002 \\ \hline
\textit{ACJR 11}    & & & & & & & & & & \\
T6.P.3.C2 & Urbanization in Germany, reform index                                            & French presence in 1850, 1875 and 1900                                                                   & 2 & 3 & 74    & 12.74   & 0.239 & 13          & 112.422    & 0.000 \\ \hline
\textit{MSS 04}   & & & & & & & & & & \\
T4.C5     & Civil conflict \textgreater{}25 deaths, Economic growth rate (t)                 & Current and lagged rainfall                                                                              & 3 & 2 & 743   & 10.30   & 0.414 & 41          & 31.022     & 0.003 \\
T4.C6     & Civil conflict \textgreater{}25 deaths, Economic growth rate (t)                 & Current and lagged rainfall                                                                              & 3 & 2 & 743   & 5.18    & 0.879 & 41          & 37.682     & 0.000 \\
T4.C7     & Civil conflict \textgreater{}1000 deaths, Economic growth rate (t)               & Current and lagged rainfall                                                                              & 3 & 2 & 743   & 5.35    & 0.867 & 41          & 42.052     & 0.000 \\ \hline
\textit{V etal 12}   & & & & & & & & & & \\
T3.C6     & Degree of altruism scale, Percentage dead in attacks                             & Distance to Bujumbura (log), Altitude (log)                                                              & 2 & 2 & 278   & 9.45    & 0.051 & 35          & 8.054      & 0.153 \\
T4.C6     & Risk preference, Percentage dead in attacks                                      & Distance to Bujumbura (log), Altitude (log)                                                              & 2 & 2 & 213   & 12.28   & 0.015 & 35          & 1.349      & 0.930 \\
T5.C6     & Discount rate, Percentage dead in attacks                                        & Distance to Bujumbura (log), Altitude (log)                                                              & 2 & 2 & 266   & 6.69    & 0.153 & 35          & 5.622      & 0.345 \\
T6.C4     & Degree of altruism scale, Percentage dead in attacks                             & Distance to Bujumbura (log), Altitude (log)                                                              & 2 & 2 & 212   & 6.36    & 0.174 & 35          & 6.931      & 0.226 \\
T6.C5     & Risk preference, Percentage dead in attacks                                      & Distance to Bujumbura (log), Altitude (log)                                                              & 2 & 2 & 158   & 8.88   & 0.064 & 35          & 6.860      & 0.231 \\
T6.C6     & Discount rate, Percentage dead in attacks                                        & Distance to Bujumbura (log), Altitude (log)                                                              & 2 & 2 & 205   & 2.34    & 0.673 & 35          & 4.451      & 0.487
\\ \hline \hline
\multicolumn{11}{@{}l}{\parbox[t]{\linewidth}{Specification: T: table; P: panel; C: column. $n_c$: number of clusters, KPST$_c$: cluster KPST statistic.}}
\end{longtable}
\end{landscape}

\end{document}